\documentclass{sig-alternate-10}
\usepackage{amsfonts}
\usepackage{balance}  % for  \balance command ON LAST PAGE  (only there!)
\usepackage{color, graphicx, MnSymbol, hyperref, subfigure}
\usepackage{multirow, framed, float, tabularx}
\usepackage{aliascnt,enumitem}
\usepackage{relsize}
\usepackage{listings}
\usepackage{colortbl}
\usepackage{tikz,pgf}
\usepackage{mathrsfs}
\usetikzlibrary{arrows}
\usepackage{algorithm}
\usepackage{algorithmicx}
\usepackage{algpseudocode}
\usepackage{fancyvrb}
\newcommand\indep{\protect\mathpalette{\protect\independenT}{\perp}} % symbols-a4, p.106
\def\independenT#1#2{\mathrel{\rlap{$#1#2$}\mkern2mu{#1#2}}}

\newcommand{\attr}{{\cal{A}}}
\newcommand{\I}{{\mathbb{I}}}

\newcommand{\ci}{{\cal{CI}}}

\newcommand\mvd{\twoheadrightarrow}

\newcommand{\X}{{\mathtt{X}}}
\newcommand{\Y}{{\mathtt{Y}}}
\newcommand{\Z}{{\mathtt{Z}}}
\newcommand{\W}{{\mathtt{W}}}

\definecolor{LightCyan}{rgb}{0.88,1,1}
\definecolor{Gray}{gray}{0.9}

\newcommand{\babak}[1]{{\texttt {\color{green} Babak: [{#1}]}}}

\newcommand{\sudeepa}[1]{{{\tt \color{blue} Sudeepa: [{#1}]}}}

\newcommand{\red}[1]{{\textbf {\tt {\color{red} {#1}}}}}

\newcommand{\proj}[1]{{\Pi}}
\newcommand{\sel}[1]{{\sigma}}

\newcommand{\cut}[1]{}
\newcommand{\cutfull}[1]{}

\newcommand{\commentresolved}[1]{}

\newcommand{\ie}{{\it i.e.}} %\xspace}
\newcommand{\eg}{{\it e.g.}} %\xspace}
\newcommand{\etal}{{et al.}} %\xspace}

\newtheorem{theorem}{Theorem}[section]          	% Theorem environment.
\newaliascnt{lemma}{theorem}				% 1 alias counter
\newtheorem{lemma}[lemma]{Lemma}              	% Lemma environment.
\aliascntresetthe{lemma}  					% 3 set
\newaliascnt{conjecture}{theorem}			% 1 alias counter
\aliascntresetthe{conjecture}  				% 3 set
\newaliascnt{remark}{theorem}				% 1 alias counter

\aliascntresetthe{remark}  					% 3 set
\newaliascnt{corollary}{theorem}			% 1 alias counter
\newtheorem{corollary}[corollary]{Corollary}      % Corollary environment.
\aliascntresetthe{corollary}  				% 3 set
\newaliascnt{definition}{theorem}			% 1 alias counter
\newtheorem{definition}[definition]{Definition}    % Definition environment.
\aliascntresetthe{definition}  				% 3 set
\newaliascnt{proposition}{theorem}			% 1 alias counter
\newtheorem{proposition}[proposition]{Proposition}  % proposition environment.
\aliascntresetthe{proposition}  				% 3 set
\newaliascnt{example}{theorem}			% 1 alias counter
\newtheorem{example}[example]{Example}  	% 2 environment.
\aliascntresetthe{example}  				% 3 set
\newaliascnt{observation}{theorem}			% 1 alias counter
\newtheorem{observation}[observation]{Observation}  % proposition environment.
\aliascntresetthe{observation}  				% 3 set
\cut{

}

\allowdisplaybreaks
\begin{document}
\title{A Framework for Inferring Causality from Multi-Relational Observational Data using Conditional Independence} %
%\titlenote{This work was supported in part by NSF awards IIS-1408846
%  and IIS-1552538, and a Google Faculty Research Award.}}
\numberofauthors{2}
\author{
\alignauthor%
Sudeepa Roy\\
\affaddr{Duke University}\\
\email{sudeepa@cs.duke.edu}%
\alignauthor%
Babak Salimi\\
\affaddr{University of Washington}\\
\email{bsalimi@cs.washington.edu}%
}

\maketitle
\pagestyle{empty}

\begin{sloppypar}
\begin{abstract}
%A Framework for Inferring Causality from Multi-Relational Observational Data using Conditional Independence and Graphical Models

The study of causality or causal inference -- how much a given treatment causally affects a given outcome in a population -- goes way beyond correlation or association analysis of variables, and is critical in making sound data driven decisions  and policies in a multitude of applications. 
The gold standard in causal inference is performing \emph{controlled experiments}, which often is  not possible due to logistical or ethical reasons. As an alternative, inferring causality on \emph{observational data} based on the \emph{Neyman-Rubin potential outcome model} has been extensively used in statistics, economics, and social sciences over several decades.  In this paper,  we present a formal framework for sound causal analysis on observational datasets that are given as multiple relations and where the population under study is obtained by joining these base relations. We study a crucial condition for inferring causality from observational data, called the \emph{strong ignorability assumption} (the treatment and outcome variables should be independent in the joined relation given the observed covariates), using known conditional independences that hold in the base relations. We also discuss how the structure of the conditional independences in base relations given as graphical models help infer new conditional independences in the joined relation. The proposed framework combines concepts from  databases, statistics, and graphical models, and aims to initiate new research directions spanning these fields to facilitate powerful data-driven decisions in today's big data world.
\end{abstract}

\section{Introduction}\label{sec:introduction}
The problem of \emph{causal inference} goes far beyond simple correlation, association, or model-based prediction analysis, and is practically indispensable in health, medicine, social sciences, and other domains. For example, a medical researcher may want to find out whether a new drug is effective in curing cancer of a certain type. An economist may want to understand whether a job-training program helps improve employment prospects, or whether an economic depression has an effect on the spending habit of people. A sociologist may be interested in measuring the effect of domestic violence on children's education or the effect of a special curricular activity on their class performances. A public health researcher may want to find out whether giving incentives for not smoking in terms of reduction in insurance premium helps people quit smoking.
Causal inference lays the foundation of sound and robust policy making by providing a means to estimate the impact of a certain \emph{intervention} to the world. For instance, if non-smokers pay reduced insurance premium anyway, and introducing the plan of reduced premium does not help smokers quit smoking, then a simple correlation analysis between people who pay less premium and who do not smoke may not be sufficient to convince policy makers in the government or in insurance companies that the new policy should be introduced -- as cited widely in statistical studies, \emph{correlation does not imply causation}.
\par
The formal study of causality was initiated in statistical science, back in 1920s and 30s by Neyman \cite{neyman1923} and Fisher~\cite{Fisher1935design}, and later investigated by Rubin~\cite{Rubin1974,Rubin2005} and
Holland~\cite{Holland1986} among others. The gold standard in causal analysis is performing \emph{controlled experiments}  or \emph{randomized trials} with the following basic idea: given a \emph{population} consisting of individual \emph{units} or subjects (patients, students, people, plots of land, etc.), randomly divide them into \emph{treatment} (or, the \emph{active treatment}) and \emph{control} (or, the \emph{control treatment}) groups. The units in treatment group receives the treatment whose effect we desire to measure (the actual drug, special training program, discount on the premium, a fertilizer), whereas the units in the control group do not receive it. At the end of the experiment, the difference in the \emph{outcome} (on the status of the disease, grade in class, smoking status, production of crops) is measured as the \emph{causal effect} of the treatment. Of course, additional assumptions and considerations are needed in experiment designs to make the results reflect the true causal effect of the treatment \cite{william1992experimental}.
\par
On the other hand, often the causality questions under consideration, including some of the questions mentioned earlier, are difficult or even infeasible to answer by controlled experiments due to ethical or logistical reasons (time, monetary cost, unavailability of instruments to enforce the treatment, etc.). Some extreme examples studied in the past in sociology, psychology, and health sciences include \cite{Rosenbaum2005} studying effects of criminal violence of laws limiting access to handguns, effects on children from occupational exposures of parents to lead, or long term psychological effects of the death of a close relative, which are not possible to analyze by controlled experiments. Nevertheless, in many such cases we have an \emph{observational dataset} recording the units, their treatment assignment, and their outcomes, possibly recorded by a research agency or by the government through surveys. Using such observational data,  it is still possible to infer causal relationships between the treatment and the outcome variable under certain assumptions, which is known as the \emph{observational study} for causal analysis.
\par
%Observational study gives a powerful tool of inferring causal relationships in the scenarios when controlled experiments are infeasible.
%The  basic problem with observation study is to \red{asi}
In observational studies,  however, when units are not assigned treatment or control at random, or when their `environment' selects their treatment (\eg, people in the rich neighborhoods received the special training program as treatment, whereas people in poorer neighborhoods formed the control group), differences in their outcomes may exhibit effects due to these initial \emph{selection biases}, and \emph{not due to the treatment}.  Some of the sources of such biases may have been measured (called \emph{observed covariates}, or \emph{overt biases} \cite{Rosenbaum2005}, \eg, age, gender, neighborhood, etc), whereas some of these biases may remain unmeasured in the observed data (called \emph{unobserved covariates}, or \emph{hidden biases}, \eg, some unrecorded health conditions). Observed covariates are taken care of by adjustment  in observational studies, \eg, by \emph{matching} treated and control units with the same or similar values of such covariates (further discussed in Section~\ref{sec:background}).

\vspace{-1mm}
 %We will review some techniques for observational study in Section~\ref{sec:background}.
\subsection{A Framework for Causal Analysis on Multi-Relational Observational Data using Conditional Independence}
% using Techniques from Causality, Databases, and Graphical Models}
In the database literature, the study of causality has so far focused on problems like finding tuples or summaries of tuples that affect the query answers \cite{MeliouGMS2011, RoyS14, bertossi2017causes,SalimiTaPP16}, abductive reasoning and view updates \cite{bertossi2017abcauses}, data provenance \cite{GKT07-semirings}, why-not analysis \cite{DBLP:conf/sigmod/ChapmanJ09}, and mining causal graphs \cite{Silverstein+2000}. Until very recently \cite{FLAME2017, salimi2017zaliql}, the topic of causal analysis as done in statistical studies in practice has not been studied in database research. On the other hand,  observational causal studies, even as studied nowadays, can significantly benefit from database techniques.  For instance, the popular \emph{potential outcome model} by Neyman-Rubin \cite{Rubin2005} (discussed in Section~\ref{sec:background}) can be captured using the relational database model, and one of the most common methods for observational studies --  \emph{exact matching} of treated and control units using their observed covariates -- can be efficiently implemented using the \emph{group-by} operators in SQL queries in any standard database management system, thereby improving the scalability of matching methods \cite{FLAME2017,salimi2017zaliql}.

%Many of the existing datasets conidered in observational studies are generated by small scale surveys or existing records, whereas using scalable techniques from databases can extend such observational studies to \emph{big data} with millions of records.
\par
This paper proposes a framework that goes beyond employing database queries to efficiently implement existing techniques for observational studies. The standard observational studies are performed on a \emph{single table}, storing information about treatment, outcome, and covariates of all units in the population. On the other hand, many available datasets are naturally stored in \emph{multi-relational format}, where the data is divided into multiple related tables (\eg, \emph{author-authorship-publications, student-enrollment-course-background, customer-order-products, restaurant-customer-reviewed-reviews} etc.). These datasets are large (in contrast to relatively smaller datasets mostly used in observational studies recorded by research agencies through surveys), readily available (\eg, DBLP \cite{dblpdata}, Yelp \cite{yelpdata}, government \cite{datagov}, or other online data repository  used for research \cite{uci-Lichman:2013}), and pose interesting causal questions that can help  design policies in schools, businesses, or health.
\par
In addition, causal analysis on observational data inherently depends on the  \emph{strong ignorability} condition (Section~\ref{sec:background}), \ie, \emph{the treatment assignment and the potential outcomes are conditionally independent given a set of observed covariates}. The common practice is to try to include as many covariates as possible to match treated and control units in order to ensure this conditional independence. However, one dataset stored as a single relation may have a limited number of possible covariates, whereas integrating with other datasets may extend the set of available confounding covariates for matching (we discuss examples in Section~\ref{sec:framework-multi}).
\par
In some other scenarios, additional interesting causal questions may arise by integrating multiple datasets or combining multiple relations. For instance, in a \emph{restaurant-customer-reviewed-review} dataset like Yelp, one may be interested in finding whether awarding special status to a customer makes her write more favorable reviews with higher ratings. Here the preferred status belongs to the \emph{customer} table, whereas the ratings of the reviews belong to the \emph{review} table. Clearly, the set of possible causal questions will be far more limited if only one relation is considered for observational studies.
\par
Extending the scope of observational studies to multi-relational data, however, requires additional challenges to be addressed, and links observational studies to understanding \emph{conditional independences} in relational databases. Causal analysis on observational data crucially depends on a number of assumptions like \emph{SUTVA} and \emph{strong ignorability} involving independences and conditional independences (Section~\ref{sec:background}).
In fact, causal inference can only be done under some causal assumptions that
are not testable from data \cite{PearlBook2000}, and the results are only valid under these assumptions.
Therefore, to perform sound causal analysis on datasets in multiple relational tables, one important task is to understand the notion of conditional independences in the given tables (\emph{base relations}) and also in the \emph{joined table}.
\par
The natural approach is to define conditional independence using the probabilistic interpretation: if variables $A, B$ are independent, then  $P(AB) = P(A) \times P(B)$, and if $A, B$ are conditionally independent given $C$, then $P(A B|C) = P(A | C) \times P(B | C)$, where in the context of relational databases $A, B, C$ denote subsets of attributes from one or multiple base relations. However, measuring probability values from an exponential space on variable combinations to deduce conditional independences in the joined relation is not only impractical for high-dimensional data, but may also be incorrect when the available data represents only a \emph{sample} from the actual dataset (\eg, data on 50 patients vs. the data on all people in the world). On the other hand, as explained in the book by Pearl \cite{Pearl-PR-book}, humans are more likely to be able to make probabilistic judgment on a small set of variables based on domain knowledge, \eg, whether a patient suffering from disease A is conditionally independent on the patient suffering from disease B given her age, which can be obtained much more efficiently and correctly without looking at the data.
% performance exhibits the opposite pattern in inferring conditional independences -- probabilistic judgement based on a small set of variables based on domain knowledge, \eg, whether a patient suffering from disease A is conditionally independent on the patient suffering from disease B given her age, which can be obtained much more efficiently and correctly without looking at the data.
\par
%While we present a general framework for causal analysis on multi-relational data and discuss several challenges, the key
Hence one fundamental problem in the framework of multi-relational observational data for observational causal studies is understanding which conditional independences in the joined relation can be inferred from a set of conditional independences given on the base relations, which is the focus of this paper. Unfortunately, testing whether a conditional independence holds given a set of other conditional independences in general is undecidable  \cite{DBLP:journals/tods/Beeri80} (the conditional independence has to hold in an infinite number of possible distributions), although  in certain scenarios, conditional independences are decidable, \eg, when the underlying distribution is \emph{graph-isomorph} (has an equivalent undirected or directed graphical representations) \cite{Pearl-PR-book, PearlPaz86} (discussed further in Section~\ref{sec:background}). Nevertheless, if we still derive certain conditional independences in the joined relation using sound inference rules, we know that a correct conclusion has been made, and can use such conditional independences for covariate selection in observational studies on the joined relation. 
%\babak{I would mention here that there are some decidable fragment of conditional independencies for which the implication problem is decidable e.g., when they are graph isomorphic (refer to \cite{butz2000relational} for a detailed discussion))}
\vspace{-2mm}
\subsection{Our Contributions} %We make the following contributions in this paper: \red{shorten}
%\begin{itemize}[leftmargin=*]
%\itemsep0em
%\item
We propose a framework for performing sound causal analysis on \emph{multi-relational} observational data with multiple base relations. Toward this goal, we review the concepts of potential outcome model and observational studies for causal analysis as studied in statistics over many decades,
%which has not been studied in the database community until very recently \cite{FLAME2017, DBLP:journals/corr/SalimiS16} to the best of our knowledge,
and make connections with the relational database model.
%\item
\par

%\babak{We need to revise this section by making less (no) emphasize on strong ignoranbility. Instead we should advertise the potential application of our result for reducing the sufficient set fo covarite and the proposition which show in some cases we dont need to join.}
Performing sound causal analysis on observational data requires conditional independence assumptions. Using the standard probabilistic interpretation, we study conditional independences that hold in the joined relation for natural joins between two relations in general and  special cases (foreign key joins and one-one joins). We show that, for any conditional independence in a base relation of the form  ``$X$ is independent of $Y$ given $Z$'', if the set of join attributes is a subset of $X, Y,$ or $Z$, then this conditional independence also holds in the joined relation. As  applications, we show that (i) conditioning on join attributes (or a superset) satisfies strong ignorability but is not useful since the estimated causal effect will be zero, and (ii) in some cases the join can be avoided to achieve the same estimated causal effect. %if in the  joined relation all attributes other than the treatment and the outcome variables are used as the set of confounding covariates, we get the desired strong ignorability condition.
\par
%\item
We show that, in general, a conditional independence that holds in the base relation may not hold in the joined relation. However, if the conditional independences that hold in the base relation are \emph{graph-isomorph} \cite{Pearl-PR-book, PearlPaz86}, \ie, if an undirected graph exists (called a \emph{perfect map or P-map}, see Section~\ref{sec:gm}) such that \emph{(i) $X$ and $Y$ are separated by a vertex set $Z$ ($Z$ forms a cutset between $X, Y$)},  if and only if \emph{(ii) $X$ and $Y$ are conditionally independent given $Z$}, and if the join is on a single attribute, then all the conditional independences from the base relation propagate to the joined relation. We also show that, for join between two relations on a single join attribute, the \emph{union graph} of the P-maps of the base relations, if the given P-maps are connected, gives an \emph{independency map or I-map} of the joined relation, where separation in the I-map implies a valid conditional independence in the joined relation (although some conditional independences %\babak{ independencies ?} 
may not be captured in an I-map unlike a P-map). We also review the notions of graphoid axioms and undirected graphical model or Markov networks from the work by Pearl and Paz \cite{PearlPaz86, Pearl-PR-book} that give a toolkit to infer sound conditional independences in the joined relation.
\par
%\item
In addition to understanding conditional independences and use of undirected graphical models to their full generality, we discuss four fundamental research directions using our framework:  \emph{First,} we discuss using \emph{directed graphical models}, or \emph{causal Bayesian networks} for inferring conditional independences. In contrast to vertex separation in undirected graphs, conditional independences in directed Bayesian networks is given by \emph{d-separation} (Pearl and Verma, \cite{PearlBook1998,DBLP:conf/aaai/PearlV87}), which may capture additional conditional independences than the undirected model, but introduces new challenges in inferring  conditional independence for joined relation. % (for instance, the union of the P-maps of two base relations may no longer be an I-map of the joined relation unlike the undirected models).
\emph{Second,} we discuss extension of inferring conditional independences in joined relation to achieving strong ignorability, which involved \emph{potential outcomes} and missing data (instead of observed outcome).
 \emph{Third,} we discuss challenges that may arise in many-to-many joins, which may violate basic causal assumptions, and need to be taken care of in observational studies. \emph{Fourth}, we discuss whether other weaker concepts of conditional independences like \emph{embedded multi-valued dependency} (Fagin,  \cite{DBLP:journals/tods/Fagin77}) is more suitable for causal analysis for multi-relational data, since using the natural probabilistic interpretation, some conditional independences may not propagate to the joined relation, which is contrary to the intuitive idea that the independences are inherent property of different relations and should be unchanged whether or not the relation is integrated with other relations.
%\end{itemize}

%On one hand our framework extends the
\vspace{-2mm}
\subsection{Our Vision for the Framework}
We envision several potential impact of our framework in both databases and causality research. (1) It extends the well-studied Neyman-Rubin potential outcome model to multiple relations, enriching the possible set of covariates that can be included or possible set of causal questions that can be asked. Further, using database techniques like group-by queries makes the popular techniques for observational studies like matching more scalable. (2) It presents the rich causality research in statistics to databases, which has significant practical implications in avoiding `incorrect causal claims' in the study of big data and in policy making.
%, but has not been much studied in the database community until recently.
(3) It brings together techniques from causality in statistics \cite{Rubin2005} and artificial intelligence \cite{PearlBook2000}, database theory (\eg, embedded multi-valued dependency \cite{DBLP:journals/tods/Fagin77}), database queries (for matching), and undirected and directed graphical models (for conditional independences) \cite{Pearl-PR-book, KollerF-PGM-book}. Thereby, it creates new research problems spanning multiple domains, and  creates scope of collaboration among database theoreticians and practitioners, statisticians, researchers in artificial intelligence, and domain scientists who are interested in solving causal questions for specific applications.
%
%Integrating multiple data sources or multiple relations also help select an extended set of covariates, thus facilitating accurate estimation of causal effects in observational studies.
%\red{The key idea in this paper is to present a framework to bring the research areas of causality and databases together, which will allow database researchers to work on a new domain that is crucial for many practical applications, and will benefit the study of causality bringing new tools and techniques from database theory and database management.}   Not only the existing matching techniques can be made more efficient and new matching techniques can be developed using database queries, the database techniques would allow robust causal analysis on much more complex and rich datasets that are often naturally stored in multiple relations.
%%The framework is aligned to the recent big data movement: due to available rich datasets
%In particular, our contributions are as follows: \red{revisit}
%%\par
% \red{ADD from Pearl Chapter}.
%\begin{itemize}
%\item We propose a new and general framework for causal analysis on observational data, which are stored in multi-relational form.
%\item We introduce a
%\end{itemize}
%%Many datasets available today are naturally
%Even if the question is simpler and is doable by a controlled experiment, in today's surge of

 \textbf{Roadmap.~} In Section~\ref{sec:background}, we review some concepts from the causality and graphical models literature. In Section~\ref{sec:prelim}, we present our framework and define the notion of conditional independence. Section~\ref{sec:condl_indep} explores conditional independences in the joined relation, while Section~\ref{sec:undirected} discusses inference using undirected graphs. We conclude by discussing further research directions in Section~\ref{sec:applications}. All proofs and further discussions on related work appear in the appendix. 
\section{Background}\label{sec:background}

In this section, we review some material on causality, observational studies, graphoid axioms, and graphical models that we will use in our framework. We use $X \indep Y$ to denote that (sets of) variables $X, Y$ are marginally independent, and $X \indep Y | Z$ or $\I(X, Z, Y)$ to denote that $X$ and $Y$ are conditionally independent given $Z$.

\subsection{Potential Outcome Model}\label{sec:pom}
The commonly used model for causal analysis in statistics is known as the \emph{potential outcome model} (or \emph{Neyman-Rubin potential outcome model}), which was first proposed in Neyman's work \cite{neyman1923}, and was later popularized by Rubin \cite{Rubin1974, Rubin2005}. % using the concept of controlled experiments by Fisher \cite{Fisher1935design}.
In causal analysis, the goal is to estimate the causal effect of a \emph{treatment} $T$  (\eg, a drug to treat fever) on an \emph{outcome} $Y$ (\eg, the body temperature). The treatment variable $T$ assumes binary value, where $T = 1$ means that the \emph{treatment (or, active treatment)} has been applied, and $T = 0$ means that the \emph{control treatment (or control)} has been applied to the unit.  %\red{multi valued treatment cite?}.
For unit $i$, we denote its outcome and treatment by $Y_i$ and $T_i$.
In contrast to predictive analysis, where one computes the distribution (probability or expectation) of $Y$ given $T = t$, in causal analysis, selected units (the treatment group) are \emph{administered} the treatment, \ie, the value of $Y_i$ is observed by (intuitively) \emph{forcing} unit $i$ to assume $T_i = 1$; this is called the \emph{intervention} mechanism. Denoting the value of the outcome as $Y_i(0)$ when $T_i = 0$, and  $Y_i(1)$ when $T_i = 1$, the goal is to estimate the causal effect by computing the \emph{average treatment effect (ATE):}
\begin{equation}\label{equn:ATE}
ATE = E[Y(1) - Y(0)]
\end{equation}
The variables $Y(1), Y(0) = \{Y_i(1), Y_i(0)\}$ are called the \emph{potential outcomes}. To estimate the ATE, ideally, we want to estimate the difference in effects of administering both $T_i = 1$ and $T_i = 0$ to the same unit $i$, \ie, we want to compute both $Y_i(1), Y_i(0)$. But the \emph{fundamental problem of causal inference} is that
for each unit we only know either $Y_i(1)$ or $Y_i(0)$ but not both, reducing the causal inference to a missing data problem \cite{Holland1986, RosenbaumRubin1983}. The potential outcome model on $n$ units %in the population
can be represented in a tabular form as shown in Table~\ref{tab:potential-outcome}; we will explain the set of \emph{covariates} $X$ in Section~\ref{sec:obs-data}.

\begin{table}[ht]
{ \small
\centering
\begin{center}
  \begin{tabular}{|c|c|c|c|c|c|} \hline
    %Unit & Covariate & Treatment & Treatment & Control & Causal  Effect \\
   \textbf{Unit} & $\mathbf{X}$ & $\mathbf{T}$ & $\mathbf{Y(1)}$ & $\mathbf{Y(0)}$ & $\mathbf{Y(1)-Y(0)}$ \\
    \hline
    1 & $X_1$ & $T_1$ & $Y_1(1)$ & $Y_1(0)$ & $Y_1(1) - Y_1(0)$ \\
    2 & $X_2$ & $T_2$ & $Y_2(1)$ & $Y_2(0)$ & $Y_2(1) - Y_2(0)$ \\
    $\ldots$ & $\ldots$ & $\ldots$ & $\ldots$ & $\ldots$ & $\ldots$ \\
    n & $X_n$ & $T_n$ & $Y_n(1)$ & $Y_n(0)$ & $Y_n(1) - Y_n(0)$ \\ \hline
  \end{tabular}
\end{center}
}
\vspace{-2mm}
\caption{Neyman-Rubin's potential outcome framework \cite{Rubin2005}}
\vspace{-2mm}
\label{tab:potential-outcome}
\end{table}

In randomized controlled experiments, however, randomly assigning treatments to units (each patient is randomly given either the drug for fever or a placebo) %and observing the outcome values
gives an unbiased estimate of ATE. In this case, \emph{the treatment assignment $T$ is independent of the potential outcomes $Y(1), Y(0)$}, \ie,
\begin{equation}\label{equn:T-indep-Y1-Y0}
T \indep Y(1), Y(0).
\end{equation}
Therefore, $E[Y(1)] = E[Y(1) | T = 1]$ and $E[Y(0)] = E[Y(0) | T = 0]$, and from (\ref{equn:ATE}) we get,
\begin{eqnarray}
\nonumber
ATE & = & E[Y(1)] - E[Y(0)] \\
& = & E[Y(1) | T = 1] - E[Y(0) | T = 0].\label{equn:ATE-avg}
\end{eqnarray}
Now the ATE can be estimated %The above equation suggests that now the ATE can be estimated
by taking the difference of average outcomes of the treated and control units under an %. However, to be able to estimate $E[Y(1) | T = 1]$ and $E[Y(0) | T = 0]$ from the results, an
additional assumption called SUTVA  \cite{Rubin2005, cox1958planning}:
%the \emph{Stable Unit Treatment Value Assumption (SUTVA)}, containing two sub-assumptions \cite{Rubin2005, cox1958planning}:
\begin{definition}\label{def:SUTVA}
\textbf{Stable Unit Treatment Value Assumption or SUTVA} (Rubin \cite{Rubin2005}, Cox \cite{cox1958planning}):
\begin{enumerate}[leftmargin=*]
\itemsep0em
\item There is no \emph{interference} among units, \ie, both $Y_i(1), Y_i(0)$ of a unit $i$ are unaffected by what action $T_j$ any other unit $j$ receives.
\item There are no hidden versions of treatments, \ie, no matter how unit $i$ received treatment $T_i = 1$ (resp. 0), the outcome will be observed $Y_i(1)$ (resp. $Y_i(0)$).
\end{enumerate}
\end{definition}
%Under SUTVA, and given the results of a randomized control experiment, using (\ref{equn:ATE-avg}), one can estimate the value of ATE by taking the difference of the average outcome values of the treated units and the average outcome values of the control units.

\subsection{Causality for Observational Data}\label{sec:obs-data}
%\cut{
%Table~\ref{tab:potential-outcome} shows the potential outcome model for $n$ units for an \emph{observational dataset} when the causal analysis is performed on a set of units with pre-recorded information in a dataset; \eg, if one is estimating the causal effect of a new drug on fever, the table will store all $n$ patients as units (or subjects), $T_i = 1$ if a patient receives the drug (treatment), and $= 0$ if she receives a placebo (control), $Y_i(1), Y_i(0)$ (however, one one appears in the dataset while the other will be missing)  denote the effect on the outcome (say body temperature) given treatment and control respectively, the vector $X_i$ denote all other variables recorded for the patient (age, gender, education, medical history, salary, ethnicity etc.), and $Y_i(1) - Y_i(0)$ denotes the individual causal effect, which, once again, cannot be estimated from the data.
%\par
%}
The potential outcome model gives a formal method to reason about `potential outcomes' and estimate ATE in controlled experiments. However, once we attempt to do causal analysis on \emph{observational data} -- a dataset containing $Y_i(1), Y_i(0), T_i$ %and a vector of other variables $X_i$ recorded
for each unit $i$ -- the independence assumption in (\ref{equn:T-indep-Y1-Y0}) typically fails. As mentioned in the introduction, this happens due to \emph{selection biases}  \cite{Rosenbaum2005}, when the treatment assignment depends on the environment of the unit (\eg, it may happen that male patients between age 20-30 received the drug and the other patients received placebo, students already performing well in a class enrolled in a special training program, etc.). In observational studies, some of these potential factors are also recorded in the dataset as variables $X$ (called \emph{confounding covariates}), while the others may remain unobserved. Table~\ref{tab:potential-outcome} for the potential outcome model also shows $X_i$ for each unit $i$, which is a vector containing possible confounding covariates (\eg, in the study of the drug for fever, $X$ may include age, gender, medical history and conditions, and ethnicity of the patient).
A controlled experiment takes care of biases due to both observed and unobserved covariates by randomization. For observational studies, one can still adjust for selection biases due to observed covariates with additional \emph{unconfoundedness assumption} as follows:
%stated as an assumption -- \emph{strong ignorability assumption} by Rosenbaum and Rubin \cite{RosenbaumRubin1983}:
\begin{definition}\label{def:SITA}
\textbf{Strong ignorability assumption} (Rosenbaum and Rubin, \cite{RosenbaumRubin1983}):
\begin{enumerate}[leftmargin=*]
\itemsep0em
\item each individual has a positive probability
of being assigned to treatment, and
\item potential outcomes
$(Y (0),Y (1))$ and $T$ are conditionally independent given relevant covariates $X$:
\begin{equation}\label{equn:T-Y1-Y0-X}
T \indep Y(0), Y(1) | X
\end{equation}
\end{enumerate}
\end{definition}
Using (\ref{equn:T-Y1-Y0-X}), we can write ATE (\ref{equn:ATE}) as

\begin{equation} \small \label{equn:ae}
  ATE = E_X[E[Y(1) | T = 1, X]] - E_X[E[Y(0) | T = 0, X]] \\
\end{equation}
 where  $E_X[E[Y(t) | T = t, X]]=\sum_{X=x} P(X = x) E[Y|T = t , X = x]$, $t = \{0, 1\}$,
 can be estimated from the observational data, of course if SUTVA (Definition~\ref{def:SUTVA}) holds too. The RHS of Equation \ref{equn:ae} is known as the {\em adjusted estimand} and is denoted by $A(T,Y,X)$. 
 
 \begin{definition}\label{def:c-eq}
 \textbf{$c$-equivalence} (Pearl and Paz,   \cite{pearl2014confounding}):
 Two set of variables $X$ and $X'$ are called {\em c-equivalent} if $A(T,Y,X)=A(T,Y,X')$, \ie, adjusting based on both $X$ and $X'$ would yield a same result. 
\end{definition} 
 Thus, if adjusting for one of them, say $X$, is sufficient for computing the causal effect of $T$ on $Y$, adjusting for the
$X'$ is also sufficient.  The following is shown in  \cite{pearl2014confounding}.

\begin{theorem} \label{th:ceq}
  (Pearl and Paz,   \cite{pearl2014confounding}): A sufficient condition for $X$ and $X'$ to be c-equivalent is that they satisfy one of the following two
  conditions:
   \begin{eqnarray*}
     (i)  \ T \indep X'|X & and & Y \indep X|X',T \\
     (ii) \ T \indep X|X' & and & Y \indep X'|X,T
   \end{eqnarray*}

\end{theorem}

\textbf{Matching methods for observational studies.~}
%\babak{I revised the following. The original is included. (propensity score matching is different from subclassifications but it has been used for that purpose too. CEM and Exact matching are indeed for of subclassification.)}
If strong ignorability holds for a set of covariates $X$, one can estimate ATE using Equation \ref{equn:ae} by (a) dividing  units with the same value of all covariates in $X$ into groups (called \emph{exact matching}), (b) 
taking the difference of average $Y$ values of treated and control units for each group. 
In practice, however, a direct application of this method is
impossible, because the data is typically very sparse: for any value
$X=x$ we either have no data values at all, or very few such values, or only treated or only control units for some groups, thereby %which means that estimating $E[Y(T)|T=t,X=x]$ becomes impossible. 
estimating ATE becomes infeasible. 
%A solution adopted in statistics is {\em matching}. The idea 
In general, matching methods used in statistics group treated and control units based on the same or similar values of covariates (exact or approximate matchings) to create a balance between treated and control units in each matched group. Popular approximate matching methods are \emph{propensity score matching}  \cite{RosenbaumRubin1983} and {\em coarsen exact matching (CEM)} \cite{iacus2009cem}.  In the former, units with the same value of $B(X_i) = b$ are matched together, $B(X) = P[T = 1 | X = x]$ being the propensity score, \cite{rubin2006matched}.
%\cut{	; Rosenbaum and Rubin  \cite{RosenbaumRubin1983} showed that propensity score is a \emph{balancing score} (within each matched group, treated and control units are independent), and is the most coarsened balancing score, thereby allowing valid matching of many units.} 
In the latter, the vector of covariates is coarsened according to a set of user-defined cutpoints and then exact matching used to match units with similar value of the coarsened covariates \cite{iacus2009cem}.
\par
Exact matching and CEM on observational data bears high resemblance with the \emph{group-by} operator in SQL queries, leading to recent applications of database techniques for matching in causal inference \cite{FLAME2017,salimi2017zaliql}. Further discussion on matching techniques can be found in Appendix~\ref{sec:related}.

\cut{
A natural approach to observational studies is \emph{matching} treated and control units based on the same or similar values of covariates (exact and approximate matchings) to create a balance between treated and control units in each matched group.
%Matching has been been studied since 1940s for observational studies \cite{chapin1947experimental,greenwood1945experimental}.
Another popular matching approach is \emph{propensity score matching}, where units with the same value of $B(X_i) = b$ are matched together, $B(X) = P[T = 1 | X = x]$ being the propensity score, which is a special case of \emph{subclassification} \cite{rubin2006matched}.
}
\cut{
; Rosenbaum and Rubin  \cite{RosenbaumRubin1983} showed that propensity score is a \emph{balancing score} (within each matched group, treated and control units are independent), and is the most coarsened balancing score, thereby allowing valid matching of many units. Propensity score matching also belongs to another set of techniques for matching called \emph{subclassification} \cite{rubin2006matched}.
}
\cut{Exact matching on observational data bears high resemblance with the \emph{group-by} operator in SQL queries, leading to recent applications of database techniques for matching in causal inference \cite{FLAME2017,salimi2017zaliql}. Further discussion on matching techniques can be found in Appendix~\ref{sec:related}.}
%, another one being the recent \emph{}

\textbf{Covariate selection for observational studies.~} The process of covariate selection chooses a good subset of variables from available variables $X$ to be used in a matching method for observational studies. Covariate selection is a challenging problem: to maintain strong ignorability, we require that (i) each valid matched group has to have at least one treated and one control unit, favoring selection of as few variables as possible, as well as (ii) the treatment and potential outcomes should be independent given the selected covariates, which may require choosing many covariates.  The efficiency of matching in the presence of a large number of covariates is another practical concern. The popular matching techniques aim to select as many variables as possible with the intuition that collecting more information will make the treatment and potential outcomes conditionally independent, whereas a more methodical approach using the underlying graphical causal structure of the covariates has been studied by Pearl and others \cite{PearlBook2000}.
% which we use in this paper.
% \red{(also see the related work section)}.

\subsection{Graphoid Axioms and Graphoids}\label{sec:graphoid}

%\paragraph{ \bf Dependency Models and Graphoids}
Given a \emph{probabilistic model} $P$ on a finite set of variables $U$ with discrete values, and three subsets $X, Y,  Z \subseteq U$, the variables $X, Y$  are said to be \emph{conditionally independent} given $Z$ (denoted by $X \indep Y |_P Z$ or $\I(X, Z, Y)$) if for all values $X = x, Y = y, Z = z$, $P(x|y, z) = P(x, z)$ whenever $P(y, z) > 0$.
The seminal work on \emph{graphoid axoims} by Pearl and Paz \cite{Pearl-PR-book, PearlPaz86} gives a set of logical conditions for constraining the set of triplets $(X, Y, Z)$ such that $X \indep Y |_P Z$ holds in $P$.
\begin{theorem}
\label{thm:graphoid}
\textbf{Graphoid axioms \cite{Pearl-PR-book, PearlPaz86}~}
For three disjoint subsets of variables $X, Y, Z$, if $\I(X, Z, Y)$ %denotes $X$ is independent of $Y$ given $Z$ in
holds in some probabilistic model $P$, then $\I$ must satisfy the following four independent conditions:
\begin{itemize}[leftmargin=*]
\itemsep-1em
\item
\textbf{Symmetry:}
\vspace{-2mm}
\begin{eqnarray}
\I(X, Z, Y) \Leftrightarrow \I(Y, Z,  X) \label{equn:sym}
\end{eqnarray}
\vspace{-1mm}
\item \textbf{Decomposition:}
\vspace{-2mm}
\begin{eqnarray}
 \I(X, Z, Y\cup W) \Rightarrow \I(X,Z, Y)~\&~\I(X, Z, W) \label{equn:dec}
\end{eqnarray}
\vspace{-1mm}
\item \textbf{Weak Union:}
\vspace{-2mm}
\begin{eqnarray}
  \I(X, Z, Y\cup W) \Rightarrow \I(X, Z\cup W, Y) \label{equn:wu}
\end{eqnarray}
\vspace{-1mm}
\item  \textbf{Contraction:}
\vspace{-2mm}
\begin{eqnarray}
 \I(X, Z, Y)~\&~\I(X, Z\cup Y, W)\Rightarrow \I(X, Y\cup W, Z) \label{equn:con}
\end{eqnarray}
\vspace{-1mm}
\end{itemize}
\vspace{-1mm}
 If $P$ is strictly positive, \ie, $P(x) > 0$ for all combination of variables $x$, then a fifth condition holds:
\begin{itemize}[leftmargin=*]
\itemsep-1em
\item \textbf{Intersection:}
\vspace{-2mm}
\begin{eqnarray}
 \I(X, Z \cup W, Y)~\&~\I(X, Z \cup Y, W)\Rightarrow \I(X, Z, Y \cup W) \label{equn:int}
\end{eqnarray}
\end{itemize}
\vspace{-1mm}
\end{theorem}
Here $Y \cup W$ denotes the union of variable sets $Y, W$.
%The above set of axioms has also been conjectured to be \emph{complete} by Pearl and Paz  \cite{PearlPaz86} ($\I$ satisfies the axioms if and only if there is a probabilistic model $P$ such that $X \indep Y|_P Z \Leftrightarrow \I(X, Z, Y)$) when $\I$ is a conditional independence relation. Although this conjecture is not proved yet, all known properties of conditional independences are derivable from the graphoid axioms, and completeness results for special cases have been derived \cite{DBLP:journals/amai/GeigerP90}.
The axiomatic characterization of probabilistic independences give a powerful tool to discuss conditional independences that hold in the model but are not obvious from the numeric values of the probabilities. It also allows us to derive new  conditional independences starting with a small set of initial independences, possibly handcrafted by an expert. Although the \emph{membership problem} -- testing whether a triplet $(X, Z, Y)$ satisfies $\I(X, Z, Y)$ given a set of conditional independences --  is undecidable \cite{DBLP:journals/tods/Beeri80},
% (the conditional independence has to hold in an infinite number of possible distributions),
still if a conditional independence can be derived using the graphoid axioms, we know that it holds in the model.
%\red{add more interpretatin from pearl chapter}
%\babak{You may add something about the decidability results here}

The graphoid axioms have a correspondence with \emph{vertex separation} in undirected graphs  \cite{Pearl-PR-book}. If $\I(X, Z, Y)$ denotes $Z$ separates $X$ from $Y$ in an undirected graph $G$ (\ie, removing $Z$ destroys all paths between $X, Y$), then it can be easily checked that $\I$ satisfies the graphoid axioms. 
%\babak{I don't get this statement. Graph separation does satisfy the graphoid axiom. I guess you mean implication problem becomes easy to check? }
%However, vertex separation may satisfy stronger inference rules than the graphoid axioms.
In general, other models in addition to conditional independence relations in probability models can also satisfy the graphoid axioms, and in that case, they are called graphoids:
\begin{definition}\label{def:graphoid}
\textbf{(Graphoids \cite{PearlPaz86}):~}
Let $\I$ denote an independence relation consisting of a set of triplets $(X, Z, Y)$ where $\I(X, Z, Y)$ denotes that $X$ and $Y$ are independent given $Z$. If $\I$ satisfies inference rules (\ref{equn:sym}) to (\ref{equn:con}), it is called a \emph{semi-graphoid}. If $\I$ also satisfies rule (\ref{equn:int}), it is called a \emph{graphoid}.
\end{definition}
Other than conditional independences and vertex separation in undirected graphs, vertex separation in directed graphs, \emph{embedded multi-valued dependencies} in relational databases \cite{DBLP:journals/tods/Fagin77} (see Section~\ref{sec:weaker-CIs}), and qualitative constraints \cite{DBLP:conf/uai/ShenoyS88} are semi-graphoids.

%A {\em dependency model} $M$ is a subset of triplets $(X,Y, Z)$ for which the predicate $\I(X,Y|Z)$:``$X$ is independent of $Y$ given $Z$", is true\footnote{Pearl (\eg, \red{cite}) uses the notation $\I(X, Z, Y)$ to denote $X$ and $Y$ are conditionally independent given $Z$. We use the notation$\I(X,Y|Z)$ for clarity to denote general conditional independence, and later in Section~\ref{sec:prelim}, we use $X \indep Y|_R Z$ to denote that (subsets of) attributes $X$ and $Y$ are independent given $Z$ in a relation $R$ using the probabilistic interpretation.}.
%A dependency model is called {\em Graphoid} if it is closed under the following five axioms:
%A dependency model is called semi-graphoid if it is closed under \ref{equn:sym}-\ref{equn:con}.

\cut{
\paragraph{ \bf Graph Interpretation.}
Consider an undirected $G = (U, E)$. For $X$,$Y$ and $Z$, three disjoint subsets of $V$, let $\I(X;Y|Z)_G$ stand for
 the phrase ``$Z$ {\em separates} $X$ and $Y$", i.e., $X$ and $Y$
 are disconnected in $G_{\underline{Z}}$, the graph obtained by
removing $Z$ and the associated edges from $G$. It is easy to verify that $\I(.)_G$
satisfies the Graphoid axioms. Thus, graph separation forms a Graphoid \cite{pearl1985Graphoids}.

\paragraph{ \bf Probabilistic Interpretation.} If $\I(X;Y|Z)_p$ stands for ``$X$ is independent of
$Y$ given $Z$ in some probability distribution $p$", then $\I(.)_p$ is a semi-Graphoids. If $p$ is strictly positive then $\I(.)_p$ forms a Graphoid \cite{pearl1985Graphoids}.
}

\subsection{Dependencies by Undirected Graphs}\label{sec:gm}
Let $\I$ be an arbitrary \emph{dependency model} $M$ encoding an independence relation $\I = \I(M)$ with a subset of triples $(X, Z, Y)$ where $\I(X, Z, Y)_M$ denotes $X$ and $Y$ are indepdendent given $Z$. As discussed above, graphoid axioms have a correspondence with vertex separation in graphs. However, it does not say whether a set of given independences $\I$ can be captured using vertex separation. Nevertheless, for some dependency models $M$, it is possible to find a graph representation on the variable set such that the conditional independence corresponds to vertex separation.
\par
We denote by $X \indep Y |_G Z$ if two subsets of vertices $X, Y$ are separated by $Z$, or $Z$ forms a \emph{cutset} between $X, Y$.
Note that the independence in graph by vertex separation has no relation with independences in a probability space or in a dependency model in general.

\begin{definition}\label{def:idp-maps}
\textbf{(D-map, I-map, P-map) \cite{Pearl-PR-book}} An undirected graph $G$ on the variable set  is a \emph{dependency map} or \emph{D-map} of $M$ if for all disjoint subsets $X, Y, Z$, we have
\begin{equation}\label{equn:dmap}
\I(X, Z, Y)_M \Rightarrow X \indep Y |_{G} Z
\end{equation}
$G$ is an \emph{independency map} or \emph{I-map} if
\begin{equation}\label{equn:imap}
\I(X, Z, Y)_M \Leftarrow X \indep Y |_{G} Z
\end{equation}
$G$ is said to be a \emph{perfect map} or \emph{P-map} of $M$ if it is both a D-map and an I-map.
\end{definition}

A D-map guarantees that vertices that are connected in $G$ are indeed dependent in $M$, but a pair of separated vertex sets in $G$ may be dependent in $M$. An I-map guarantees that separated vertices in $G$ are independent in $M$, although if they are not separated in $G$, they may still be independent in $M$. Empty graphs are trivial D-maps and complete graphs are trivial I-maps, although they are not useful in practice. In particular, obtaining an I-map encoding as many independences as possible is useful for causal analysis on observational data since it helps in covariate selection satisfying strong ignorability (Definition~\ref{def:SITA}).
In Section~\ref{sec:undirected} we will obtain I-maps for the joined relation using P-maps of base relations.
\par
Ideally, we want to obtain an I-map that is also a P-map, but there are simple dependency models that do not have any P-maps (\eg, the ones where $\I(X, Z, Y)$ but $\neg \I(X, Z\cup W, Y)$, since a superset of a cutset is also a cutset). %Therefore, we want I-maps that can encode as many independences as possible.
%In general. undirected graphical models have limited expressive power. As a special case of the above example that it cannot encode, it cannot also encode when two variables $X, Y$ are mutually independent, but becomes dependent when a third variable $Z$ is observed.
Some such dependencies can be captured using directed graphical models that we discuss in Section~\ref{sec:future-directed} as a research direction using this framework.
\par
A dependency model $M$ is \emph{graph-isomorph} if there exists an undirected graph that is a P-map of $M$. The following theorem gives a necessary and sufficient conditions for dependency models to be graph-isomorph.
\vspace{-2mm}
\begin{theorem}\label{thm:PP-undirected}
%\textbf{(Characterization of graph-isomorph dependency models,
(Pearl and Paz \cite{PearlPaz86, Pearl-PR-book})
A necessary and sufficient conditions for a dependency model $M$ to be graph-isomorph is that $\I(X, Z, Y) = \I(X, Y, Z)_M$ satisfies the following five independent axioms on disjoint set of variables:
\begin{itemize}[leftmargin=*]
\itemsep-1em
\item
\textbf{Symmetry:}
\vspace{-2mm}
\begin{eqnarray}
\I(X, Z, Y) \Leftrightarrow \I(Y, Z,  X) \label{equn:sym-1}
\end{eqnarray}
\vspace{-1mm}
\item \textbf{Decomposition:}
\begin{eqnarray}
 \I(X, Z, Y\cup W) \Rightarrow \I(X,Z, Y)~\&~\I(X, Z, W) \label{equn:dec-1}
\end{eqnarray}
\vspace{-1mm}
\item \textbf{Intersection:}
\begin{eqnarray}
 \I(X, Z \cup W, Y)~\&~\I(X, Z \cup Y, W)\Rightarrow \I(X, Z, Y \cup W) \label{equn:int-1}
\end{eqnarray}
\vspace{-1mm}
\item \textbf{Strong Union:}
\begin{eqnarray}
  \I(X, Z, Y) \Rightarrow \I(X, Z\cup W, Y) \label{equn:su-1}
\end{eqnarray}
\vspace{-1mm}
\item  \textbf{Transitivity:}~~ For all other variables $\gamma$
\begin{eqnarray}
 \I(X, Z, Y) \Rightarrow \I(X, Z, \gamma) ~\textrm{or}~ \I(\gamma, Z, Y) \label{equn:tran-1}
\end{eqnarray}
\end{itemize}
\vspace{-1mm}
\end{theorem}
It is possible to construct a \emph{minimal I-map} of a probability distribution (removing any edge would destroy the I-map) and to check if a given a graph is an I-map of a given  probability  distribution by constructing the \emph{Markov Network} of the probability distribution; when we quantify the links of the undirected graph, the theory of Markov fields allows constructing a complete and consistent quantitative model preserving the dependency structure of an arbitrary graph $G$ \cite{Pearl-PR-book}. Conditional independences can also be captured using directed acyclic graphs (Bayesian networks and Causal graphs) that we discuss in Section~\ref{sec:applications}.

%\subsection{Directed Bayesian Networks and Causal Graphs}

\section{A Formal Framework for Causal Analysis on Multi-Relational Data}\label{sec:prelim}
Using the concepts in the previous section, we now describe a framework for causal analysis on observational data given in multiple relations.
First we describe some notations used in the rest of the paper.
%, and define the problem formally.
\par
Let $D$ be a database schema with $k$ relations $R_1, \cdots, R_k$ (called \emph{base relations}). We use $R_i$, where $i \in [1, k]$ both as the name of a relation and the subset of attributes contained in the relation.
$\attr = \cup_i R_i$ denote the set of all attributes. %For an attribute $A \in \attr$, $\dom(A)$ denotes the domain of $A$.
\par
Any join without a join condition in the subscript denotes the \emph{natural join} with equality condition on the common set of attributes, \ie, $R \Join S = R \Join_{R.A_1 = S.A_1 \cdots R.A_p = S.A_p} S$, where $A_1, \cdots, A_p = R \cap S$ denotes the common set of attributes in $R$ and $S$.
%Here $R, S$ are two relations in $D$, \ie, $R = R_i$ and $S = S_j$ for some $i, j \in [1, k]$.
We use $A, B, C, \cdots \in \attr$ to denote individual attributes, and unless mentioned otherwise, $X, Y, Z, W, \cdots \subseteq \attr$ to denote subsets of attributes.
For two subsets of attributes $X, Y \subseteq \attr$, $XY$ denotes $X \cup Y$. We will use $A \in \attr$ both as an attribute and as a singleton subset $\{A\}$.  For a tuple $t \in R$, and  attribute(s) $A \in R$, we use $t[A]$ to denote the value of attribute(s) $A$ of $R$.
\par
For an attribute $A \in \attr$ and its value $a$, if $A \in R$, then  $N_{R, (A, a)}$ denotes the number of tuples in $R$ with $A = a$, \ie, $N_{R, (A, a)} = |\{t \in R~ : t[A] = a\}|$. Unless mentioned otherwise, we assume the \emph{bag semantics}, \ie, the relations can have duplicates, and projections are duplicate-preserving. For multiple attributes $A_1, \cdots, A_p$ and their respective values $a_1, \cdots, a_p$, $N_{R, (A_1, a_1), \cdots, (A_p, a_p)} =  |\{t \in R~ : t[A_1] = a_1, \cdots, t[A_p] = a_p\}|$ denotes the number of all tuples matching all the values of all the attributes. When the attributes are clear from the context, we will use $N_{R, a}$ instead of $N_{R, (A, a)}$, and $N_{R, a_1, \cdots, a_p}$ instead of $N_{R, (A_1, a_1), \cdots, (A_p, a_p)}$ for simplicity.
$N_R$ denotes the number of tuples in $R$.\\

\subsection{Conditional Probability and Independence in a Relation}
%\textbf{Probability and conditional probability in a relation.~}
Given an instance of a relation $R$, the probability distribution of an attribute $A$ is given by
\begin{equation}
\Pr_R[A = a] = \frac{N_{R, (A, a)}}{N_R}
\end{equation}
This notion of probability has also been used in \cite{DBLP:conf/pods/DalkilicR00} to define information dependencies.
The joint probability distribution of two attributes $A, B$ is given by (similarly, for two subsets of attributes)
\begin{equation}
\Pr_R[A = a, B = b] = \frac{N_{R, (A, a), (B, b)}}{N_R}
\end{equation}

%\babak{How do you denote an instance? I assumed R is used for both schema and instance}

Note that for an attribute $A \in R \cap S$ belonging to two relations $R, S$, the distribution of $A$ in $R$ and $S$ may be different, \ie, in general, $\Pr_R[A = a] \neq \Pr_S[A = a]$.
\par
Given two attributes $A$ and $B$, the conditional probability of $A$ given $B$ is given by
\begin{equation}
\Pr_R[A = a | B = b] = \frac{\Pr_R[A = a, B = b]}{\Pr_R[B = b]} = \frac{N_{R, (A, a), (B, b)}}{N_{R, (B = b)}}
\end{equation}

%\textbf{Conditional independence in a relation.~}
For every relation $R$ in $D$, we also have a set of conditional independence statements $\ci_R$ defined as follows:
\begin{definition}\textbf{(Conditional independence)~} Let $X, Y, Z$ be three mutually disjoint subset of attributes in $R$. We say that $X$ and $Y$ are \emph{conditionally independent} given $Z$, denoted by $X \indep Y |_R Z$, if for all values $x, y, z$ of $X, Y, Z$ (all $X, Y, Z, x, y, z$ are subsets of attributes or values),
\begin{equation}
\Pr_R[X = x| Z = z] \times \Pr_R[Y = y | Z = z] = \Pr_R[X = x, Y = y | Z = z]
\end{equation}
If $X \indep Y |_R Z$ does not hold in $R$, we write $\neg(X \indep Y |_R Z)$.
\end{definition}
Similarly, $X$ and $Y$ are (marginally) independent if $\Pr_R[X = x] \times \Pr_R[Y = y] = \Pr_R[X = x, Y = y]$.
Note that for multi-relational databases, the subscript $R$ is important, since even if all of $X, Y, Z$ belong to two relations $R, S$, it may hold that $X \indep Y |_R Z$ whereas
$\neg (X \indep Y |_S Z)$.
We will use \emph{CI} as an abbreviation of \emph{Conditional Independence}. For a CI $X \indep Y |_R Z$, we say that $X, Y$ belongs to the \emph{LHS} (left hand side) of the CI, and $Z$ belongs to the \emph{RHS} (right hand side).

\textbf{Entropy.~}
Under the defined distribution each subset $X \subseteq \attr$ %=\{x_1 \ldots x_n\} \subseteq \attr$ 
defines a marginal distribution $P(X)$ with 
\emph{entropy} $H(X)=-\sum_{x} P(X=x) \log P(X=x)$, where $x$ denotes a combination of values of attributes in $X$.
%We denote differe information measures for $X$,$Y$ and $Z$, arbitrary subsets of $\attr$ as follows:
Given $X, Y, Z \subseteq \attr$, the other information measures that we will use in the paper are (\eg, see \cite{InfoTheoryBookCoverThomas}):
(i)  \emph{conditional entropy:} $H(X|Y)=H(XY)-H(Y) $; (ii) \emph{mutual information} 
$I(X,Y)=H(X)+H(Y)-H(XY)$,  and (iii) \emph{conditional mutual information:}
	$I(X,Y|Z)=H(X|Z)-H(X|YZ)$. % (recall that $XY$ denotes the union of the attributes in $X$ and $Y$).  
	\cut{For any $X$, $Y$ and $Z$, it is known that 
	the following basic information inequalities hold: 
  $H(X) \geq 0$, $H(X|Y) \geq 0$, $I(X;Y) \geq 0$
	and  $I(X;Y|Z) \geq 0$ \cite{yeung2008information}.  }
	Note that we are using $\I$ for independence and $I$ for mutual information.

\par
\textbf{Independence in schema vs. instance.} Independence is typically considered as a property of a schema, \ie, a conditional or unconditional independence should hold on all possible instances of a relation (which may not be true in practice, since we mostly get a \emph{sample} of the real world). For instance, if we are looking at  students database, the number of courses taken by a student and whether the student has done an internship may not be independent (senior students are likely to take more courses and also do an internship). However, given a student or given the seniority of the students, these two attributes are conditionally independent. This conditional independence follows from the domain knowledge and not necessarily from a specific instance of the database. % (although conditional independence can be learnt from the data \red{true?? then cite}).
%Therefore, these conditional independence statements should hold on all possible instances of the database.
In this paper, we focus on conditional independence statements that can be inferred from the schema-level information.
%, and therefore should hold on all possible relation instances, and
%We revisit this question as a future direction in Section~\ref{sec:applications}. % \red{check, think about it, and revisit}

\cut{
\babak{From here on you confused the independence of the potential outcomes and the treatment with the inductance of the outcome and treatment. If treatment and outcome are independent then there is no causal relationship between them. }

\babak{Not sure id the example scenario is the right one. this is because we only join the tables based on ID. Concourse ID is not participate in the causal model of any tables.  }
}
\subsection{Why Multi-Relational Data}\label{sec:framework-multi}
Note that the potential outcome framework as shown in Table~\ref{tab:potential-outcome} resembles a single relation or table in relational database model. If all required information, \ie, treatment, outcome, and confounding covariates, is available as a single table, the standard potential outcome model with a single table suffices.
However, to do a sound and robust causal analysis on observational data, it is desirable to collect as much information about the units (\ie, possible confounding covariates) as possible. The first example below motivates why it is useful to include additional attributes as covariates by combining multiple relations.
\begin{example}\label{eg:1}
\textbf{(Extension of covariate set by joining multiple relations) }
Suppose we have a Students dataset of the form
$$\mathtt{Students(\underline{sid}, major, parents\_income, gpa)},$$
which stores the \emph{id} of the student, major, annual income of the parents, and the gpa.  The goal is to estimate the causal effect of income of the parents on the performance or gpa of the students, in particular, how much having an annual income of $> 100k$  ($T = 1$ if and only if $\mathtt{parents\_income} > 100k$) affects the gpas of the students.
If we do causal analysis using this dataset only, the only available covariate is the major of the student. Conditioning on $id$ does not help, since only one tuple will satisfy a given $id$ values, thereby violating strong ignorability condition (Definition~\ref{def:SITA}) and making matching unusable, since a group will not have one treated and one control units to estimate the causal effect within a group. Now if we only use major of the student as a covariate, (i) it may not satisfy the strong ignorability condition that $\mathtt{parents\_income' \indep gpa | major}$ (where $\mathtt{parents\_income'}$ denotes the potential outcomes for income for low and high gpas), so considering this covariate will give wrong estimate of ATE (\ref{equn:ATE}) and (ii) ignoring this covariate will lead to considering all students in a single giant group, thereby again giving an incorrect estimate, since it may not hold that   $\mathtt{parents\_income' \indep gpa}$.

\par
On the other hand, there may be several other datasets available that include additional information about this causal analysis. There may exist a \emph{course} relation and an \emph{enrollment} relation storing the courses the students took: $\mathtt{Enroll}(\mathtt{sid, cid})$, and $\mathtt{Course}(\mathtt{\underline{cid}, year, title, dept, instructor})$.
There also may exist relations storing the names of the parents, and their jobs, educational background, whether they own a house, their ethnicity, and age:
$\mathtt{Parents}(\mathtt{sid, pid}),$ and
$\mathtt{ParentInfo}(\mathtt{\underline{pid}}$, $\mathtt{name}$, $\mathtt{job}$, $\mathtt{edu}$, $\mathtt{owns\_house}$, $\mathtt{ethnicity, age})$.
For simplicity, we assume that information about only one parent is stored, and revisit this assumption in Section~\ref{sec:app-mult-rows}.
Now in the causal analysis, we can include (some of the) attributes $\mathtt{year, title, dept, instructor}$ from the Course relation, and $\mathtt{job, edu, owns\_house, ethnicity, age}$ as additional covariates. The ethnicity, age, education may affect the income of the parent, as well as can also affect the gpa of the student, whereas the information about the courses can affect the gpa. Including these additional as covariates has the potential to make the causal analysis better.
\end{example}
Not only using multiple relations extend the set of available covariates, it also extends the scope of causal analysis where the treatment $T$ belongs to one relation and the outcome $Y$ belongs to another relation, allowing additional causal questions that can be asked.
\begin{example}\label{eg:2}
\textbf{(Extension of causal questions by joining multiple relations) }
In Example~\ref{eg:1}, one may be interested in estimating the causal effect of the profession ($\mathtt{job}$) or education level ($\mathtt{edu}$) of the parents (which comes from the $\mathtt{ParentInfo}$ relation) in the gpa or major of the student (which comes from the $\mathtt{Student}$ relation). Clearly, this question cannot be answered if we look at only single relation, but can be answered if we consider the join of $\mathtt{Students, Parents},$ and $\mathtt{ParentsInfo}$ relations.
\end{example}
In summary, sometimes the data iteself is naturally stored in multiple relations to reduce redundancy, whereas in other cases integrating datasets to gather new information may extend the set of covariates that can be used, or the set of causal questions that can be asked. Therefore, allowing multiple relations in observational data significantly extends the potential outcome framework by Neyman-Rubin (Sections~\ref{sec:pom} and \ref{sec:obs-data}) to be further useful for practical purposes.

\subsection{Framework: Causal Analysis with Multiple Relations}
Suppose the data is stored in $k$ relations $R_1, \cdots, R_k$, and we want to perform the causal analysis on the joined relation $J = R_1 \Join \cdots \Join R_k$. If the intended \emph{treatment} variable $T$ is not already in binary form, we consider a derived attribute $T$, which is 1 if and only if a given predicate $\phi$ on one of the chosen columns $A_T \in \attr$ evaluates to true, \ie,  $T = 1 \Leftrightarrow \phi(A_T) = true$. In Example~\ref{eg:1}, $\phi =$ whether $\mathtt{parents\_income} > 100k$.
Another chosen column $Y \in \attr$ serves as the outcome variable, and may assume real values.
\begin{itemize}
\itemsep0em
\item \emph{The final goal  is to estimate the causal effect (ATE, (\ref{equn:ATE})) of $T$ on $Y$.}
\end{itemize}
To meet the above goal using techniques from observational studies, we need to solve several sub-goals. The first and foremost sub-goal is the following:
\begin{itemize}
\itemsep0em
\item \emph{Select a set of covariates $X \subseteq \attr \setminus \{TY\}$  such that the strong ignorability (Definition~\ref{def:SITA}) holds, \ie, }
\begin{equation}
T \indep Y(0), Y(1) |_{J} X
\end{equation}
\end{itemize}
%\babak{This is different from unconfoundedness}
Unfortunately, due to the fundamental problem of causal analysis \cite{Holland1986, RosenbaumRubin1983}, for any unit only one of $Y(0)$ and $Y(1)$ is observed. 
If a \emph{directed causal graph} on the variables is available,  Pearl \cite{PearlBook2000} gives sufficient conditions called \emph{backdoor criteria} (defined in Section~\ref{sec:applications}) to check for \emph{admissible} covariates $X$ satisfying the above condition. There is a long-standing debate among causality researchers whether the causal graphical model is a practical assumption. 

\subsection{Valid Units for Causal Analysis from Joins}\label{sec:valid-units-joins}
%\sudeepa{This will probably go to Section 3.}
For doing causal analysis on joined relation, one of the first tasks is to define the units. We illustrate the challenge in this task in the following example.
\begin{example}\label{eg:units-join}
Consider three relations:\\
 $\mathtt{P(\underline{iid}, training, seniority)}$,  $\mathtt{S(\underline{sid}, gpa, major)}$,
 $\mathtt{R(\underline{iid, sid}, class, sem, year, grade)}$, respectively denoting relations for \emph{Professor, Student}, and \emph{Teaches}. Here $\mathtt{iid}$ and $\mathtt{sid}$ denote the id-s of instructors and students (also foreign keys from $R$ to $P, S$), and other variables are $\mathtt{training}$ (Boolean variable denoting whether the professor had a special training or went to a top-10 school), $\mathtt{seniority}$ (of the professor: senior, mid-level, or junior), $\mathtt{class, sem, year, grade}$ (the class taught by the professor that the student took, semester,  year, and grade obtained by the student), $\mathtt{gpa}$ (average grade of the student), $\mathtt{major}$ (major of the student). The attributes $\mathtt{iid, sid}$ are foreign keys in $R$ referring to $P$ and $S$ respectively. 
\par
Suppose one asks the question
\begin{itemize}[leftmargin=*]
\itemsep0em
\item Estimate the causal effect of the training received by instructors on the grades of the students. 
\end{itemize}
There are several tasks to be solved to answer this question: (1) what should be the units, treatment $T$, and outcome $Y$, (2) will they satisfy the basic assumption SUTVA, (3) what would be a good choice of confounding covariates $X$ to satisfy the strong ignorability condition.
\par
From the question, intuitively $\mathtt{training}$ should be the treatment $T$ (we revisit this below), but for $Y$, we have two choices: $\mathtt{grade}$ from $P$ and $\mathtt{gpa}$ from $S$. 
\end{example}

We make the following observation:
\begin{observation}\label{obs:units-unique-outcome} 
%\babak{read my comment in section 4 regarding this. I argue that the potential outcome as well as the actual outcome of two units with exactly the same treatment/covariates should be the same. I am not sure if we can capture this easily with FDs. In this example for gpa, SUTVA is violated because outcome of a unit is a function of the treatment assignment and covariates of another units. Simply because its an aggregate.}
Let $R$ be the table containing a specified outcome $Y$ and let $U$ be the population table containing the units. 
For SUTVA to hold, each $Y$-value from the rows in $R$ can appear at most once in $P$, \ie, each tuple in $R$ can contribute to at most one tuple in $U$\footnote{If each relation has unique identifiers for tuples as in Example~\ref{eg:units-join}, we can say that $R_i \rightarrow U$ holds in $U$ is necessary, but if the identifiers do not exist, two different tuples may have the same value of $Y$ and other covariates in $R$. If each tuple in $R$ contributes to at most  one tuple in $U$, $R_i \rightarrow U$ is not necessary.}.
%$R \rightarrow U$ (the attributes from $R$ form a superkey in $U$) is a necessary condition.
\end{observation}

Consider Example~\ref{eg:units-join}. If $Y = \mathtt{grade}$ and $T = training$, $U = P \Join R$, due to the foreign keys from $R$ to $P$, the $Y$ values are not repeated in $U$. The pairs \emph{(student, professor)} constitute the units with unique outcome (not value-wise, two students may receive the same grade). Now the standard techniques (\eg, matching on some covariates $X$ in $U$ that satisfy ignorability) can be applied to do the causal analysis in $U$. The same holds if $U = P \Join R \Join S$, then also the outcomes are unique.
\par
Now suppose we choose $Y = \mathtt{gpa}$, $T = \mathtt{training}$, and population table $U = P \Join R \Join S$. Suppose a student $s$ has been taught by two professors $p_1, p_2$, where $p_1$ has $T = 1$ and $p_2$ has $T = 0$. Then in $U$, there will be two tuples $u_1, u_2$ for $s$, one with $p_1$ the other with $p_2$, both with the same $\mathtt{gpa}$ value say $g_s$. Now the units still are \emph{(student, professor)} pairs, but the treatment of $u_1$ affects the (same) outcome of $u_2$, thereby violating SUTVA. Hence in this case the units are not valid.
\par
Note that Observation~\ref{obs:units-unique-outcome} gives a necessary condition for defining valid units satisfying SUTVA making the joined relation amenable to observational causal analysis. 
For instance, if the treatment $T$ was in $R$, and if SUTVA was originally violated in $R$ ($T$ of a tuple $r_1 \in R$ affects the outcome $Y$ of another tuple $r_2 \in R$), then even if each tuple in $R$ contributes at most once to $U$, SUTVA will be violated in $U$ (so this is not a sufficient condition).
\par
On the other hand, there can be another unit table constructed after join satisfying SUTVA:
If we choose $Y = gpa$, there might be a plausible option to collapse  $P \Join R \Join S$ to define valid units with unique $\mathtt{gpa}$ (\ie, the (\emph{$student$}) becomes the unit), \eg, by aggregating over different $\mathtt{training}$ values to define $T$ (at least one instructor had training or the majority of the instructors had the training) as well as aggregating different values of covariates from $P$ or $R$. We leave this as a direction of future research (Section~\ref{sec:app-mult-rows}) and assume Observation~\ref{obs:units-unique-outcome} holds in this paper.

\subsection{Inferring CIs in a Joined Relation}
As a stepping stone toward understanding strong ignorability for potential outcomes in the presence of multiple relations, we need to understand how CIs propagate from base relations to the joined relation in the standard relational database model, which is the main focus of this paper. 
The problem of inferring $X \indep Y |_{J} Z$ in a joined relation  $J = R_1 \Join \cdots \Join R_k$., even if we ignore the missing data problem arising in the application of causal inference, is non-trivial. 
As discussed earlier, given the dataset, it may be either inefficient or incorrect to validate this CI in $J$ by computing the numeric probabilities: (i) $Z$ may be large, and the independence has to be checked for exponentially many combinations of values of variables in $Z$, and (ii) the available data itself may not be complete, \ie, the dataset may be a sample from the actual world.  On the other hand, prior knowledge or expert knowledge may result in some conditional independences in individual base relations involving a small set of variables or attributes. With this intuition, we define %a more general 
the \emph{problem of inferring CIs in a joined relation}, which has other potential applications as discussed in Section~\ref{sec:QO}.
%that may have other applications, like estimating histograms for collecting statistics in query optimization:
\begin{itemize}
\itemsep0em
\item \emph{Given conditional independences $\ci_i$ in base relations $R_i$-s, and three disjoint subset of attributes $X, Y, Z$, infer whether $X \indep Y |_{J} Z$ in the joined relation $J = R_1 \Join \cdots \Join R_k$.}
\end{itemize}
In other words, instead of directly inferring the CIs in the joined relation, if we have knowledge about CIs that hold in the base relations and the nature of join, how we can infer CIs in the joined relation. 
Unfortunately, testing CI is undecidable in general \cite{DBLP:journals/tods/Beeri80}. Nevertheless, using properties of join, we can still infer some CIs in the joined relation and use them for observational causal analysis (Section~\ref{sec:app-ci}). Further, when the base relations are graph-isomorphs, or at least have any non-trivial I-maps (Section~\ref{sec:gm}), we can infer a larger classes of CIs in the joined relation, as we illustrate in Section~\ref{sec:undirected} for a special case when the join is on single attribute.
%In this paper, we focus on this question, draw connections with causal inference, and discuss other applications in the following subsection.
%Apart from the main goal of inferring CIs for strong ignorability, there are several other challenges that are introduced when we consider multi-relational data.
We will discuss other sub-goals in Section~\ref{sec:applications} as further directions of research.
The  preliminary results on this framework presented in the following two sections primarily use a binary join between two relations $R \Join S$, and we discuss extensions to multiple relations as problems to study in the future. 

%\babak{Followings reads very well. I'd mention Amol deshpande paper too. Note that they use Markov network (undirected graphs)
%for selectivity estimation. It is good to mention that both Bayesian and Markov networks are used for this purpose. Since our results are more about Markov networks }
%\par
\subsection{Application of Inferring CIs in Query Optimization}\label{sec:QO}
%\textbf{Other applications of inferring CIs in joins in query optimization.~~}
Other than helping understand CIs for causal inference with multiple relations, inferring CIs in a joined relation is useful in fine-grained \emph{selectivity estimation} for query optimization as used in modern query optimizers \cite{getoor2001selectivity,Tzoumas2013,deshpande2001independence}
%use fine-grained histograms for distribution of values 
instead of simply considering textbook assumptions (like uniform distribution or independence among attributes). 
Knowing additional CIs in base relations reduce the number of different combinations of attribute values that has to be maintained. For instance, if in $R(A, B, C, D)$, we know that $A \indep B | C$, then by chain rule, $\Pr_R(A, B, C) = \Pr_R(AB|C)\Pr_R(C) = \Pr_R(A|C) \Pr_R(B | C) \Pr_R(C)$, and instead of keeping frequencies for all possible combinations of $A, B, C$, we can have a table for $C$, and two other tables for $A | C$ and $B | C$. These tables are likely to be much smaller, thereby making the computation of joint distribution more efficient.  
\par
Selectivity estimation using graphical models has been studied in different contexts in the literature (probabilistic graphical models for select and foreign key join queries using probabilistic relational model \cite{getoor2001selectivity}, undirected graphical model-based compact synopses for a single table \cite{deshpande2001independence}, and model-based approach using directed and undirected gaphical models for multiple relations \cite{Tzoumas2013}); 
%where interactions between attributes of a multi-dimensional data in a single table is captured through a graphical model (as undirected Markov networks or directed Bayesian networks). .
efficiently learning undirected and directed graphical models has also been extensively studied \cite{KollerF-PGM-book, PearlBook2000} which unfortunately is computationally expensive \cite{neapolitan2004learning,pearl2014probabilistic}. 
%So in the presence of joins, learning CIs as graphical models would require running an expensive algorithm on many attributes. 
On the other hand, in this paper we study inferring CIs in the joined relation structurally given CIs and graphical models on the base relations (that are either constructed by existing algorithms or obtained using domain knowledge), as well as the properties of the join,  without looking again at the data, which can reduce the complexity significantly as well as help in selectivity estimation of subsequent steps in a query optimizer. For instance, we show in Section~\ref{sec:undirected}, for a special case of join, an undirected I-map of the joined relation can be obtained by taking the union of P-maps of two base relations. 

\section{Conditional Independence in the Joined Relation}\label{sec:condl_indep}
In this section we study the problem stated in the previous section for the special case of binary joins:
%-- given CIs on the base relations whether we can deduce CIs in the join relation. The problem we investigate in this section is the following:
\emph{Given two relations $R, S$, with conditional independences $\ci_R, \ci_S$ respectively,
%\begin{itemize}
%\item
(1) Which of the conditional independences in $\ci_R, \ci_S$  hold in $R \Join S$ for arbitrary $R$ and $S$?
%\item
(2) What other conditional independences hold in $R \Join S$?
%\end{itemize}
}

\subsection{CI in Joined Relation for Binary Joins}\label{sec:join-ci-general}
The main result in this section is that, if the join attributes (the common attributes of $R, S$) belong to one of the three subsets participating in a CI  that holds in a base relation, then the CI holds in the joined relation. Further, any pairs of subsets of attributes from the two relations participating in the join are independent in the joined relation given the join attributes:
%In this section, we show the following theorem answering the above questions for general schema:

\begin{theorem}\label{thm:ci-join}
Given two relations $R, S$, with CIs $\ci_R, \ci_S$ respectively, the joined relation $R \Join S$ satisfies the following CIs:
\begin{itemize}
\item For all $R' \subseteq R \setminus S$ and $S' \subseteq S \setminus R$, it holds that $R' \indep S' |_{R \Join S} (R \cap S)$,
\item For all CI $X \indep Y|_R Z \in \ci_R$, if $(R \cap S) \subseteq Z$, then $X \indep Y|_{R \Join S} Z$.

\item For all CI $X \indep Y|_R Z \in \ci_R$, if $(R \cap S) \subseteq X$ (or $Y$), then $X \indep Y|_{R \Join S} Z$.
\end{itemize}
\end{theorem}
%\babak{the proofs look good to me.}

The proof of the theorem follows from Corollary~\ref{cor:cl-join-0}, Lemma~\ref{lem:ci-join-1}, and Lemma~\ref{lem:ci-join-2} below, all proofs are given in the appendix.
Note that the above theorem does not say that no other CIs hold in the joined relation. In fact, all CIs that can be obtained by applying graphoid axioms (Theorem~\ref{thm:graphoid}) on the CIs stated in the theorem will hold in the joined relation.

%\subsubsection{Type-0 CIs: on Join Attributes}\label{sec:join-type-0}
\textbf{(I) CIs conditioning on the join attributes:}
The join introduces new CIs in the joined relation $R \Join S$:
\begin{lemma}\label{lem:cl-join-0}
In the joined relation, $(R \setminus S) \indep (S \setminus R) |_{R \Join S} (R \cap S)$.
\end{lemma}

\cut{
\begin{proof}
Fix any value $z$ (as a set) of the attributes $Z = R \cap S$. Given $Z = z$, all tuples $r \in R$ with $r.Z = z$ join with all tuples $s \in S$ with $s. Z = z$, and with no other tuples.
Let $X = R \setminus S$, $Y = S \setminus R$, $T = R \Join S$. For any values $X = x, Y = y$, $N_{T, xz} = N_{R, xz} \times N_{S, z}$, $N_{T, yz} = N_{S, yz} \times N_{R, z}$,
$N_{T, xyz} = N_{R, xz} \times N_{S, yz}$, and $N_{T, z} = N_{R, z} \times N_{S, z}$. Hence it follows that
$$\frac{N_{T, xyz}}{N_{T, z}} = \frac{N_{T, xz}}{N_{T, z}} \times \frac{N_{T, yz}}{N_{T, z}}$$
\ie, $X \indep Y |_{T} Z$.
\end{proof}
}

Proof is in the appendix. The following corollary follows from the graphoid axioms using the decomposition rule (\ref{equn:dec}) multiple times:
\begin{corollary}\label{cor:decomp-subset}
For any probability space $P$, if $X \indep Y |_P Z)$, then for all subsets $X' \subseteq X$ and $Y' \subseteq Y$, $X' \indep Y' |_P Z$.
\end{corollary}

Using Corollary~\ref{cor:decomp-subset} and Lemma~\ref{lem:cl-join-0}, the following corollary holds:
\begin{corollary}\label{cor:cl-join-0}
For all subsets of attributes $R' \subseteq R \setminus S$ and $S' \subseteq S \setminus R$ (including singleton attributes), $R' \indep S' |_{R \Join S} (R \cap S)$.
\end{corollary}

%\subsubsection{Type-I CIs: on RHS}\label{sec:join-type-1}
\textbf{(II) CIs with join attributes on the RHS:} %\label{sec:join-type-1}
Here we show that if the joined attributes belong to the RHS of a CI in a base relation, then the CI propagates to the joined relation.
\begin{lemma}\label{lem:ci-join-1}
For any $X \indep Y |_R Z$ in $\ci_R$,  in the joined relation $X \indep Y |_{R \Join S} Z$ if $R \cap S \subseteq  Z$, \ie, all join attributes $R \cap S$ belongs to $Z$.
\end{lemma}

Proof is in the appendix.
\cut{
\begin{proof}[of Lemma~\ref{lem:ci-join-1}
By the definition of conditional independence, and since $X \indep Y |_{R \Join S} Z$, for all values $x, y, z$ of $X, Y, Z$ we have
\begin{equation}\label{equn:6}
\frac{N_{R, xyz}}{N_{R, z}} = \frac{N_{R, xz}}{N_{R, z}} \times \frac{N_{R, yz}}{N_{R, z}}
\end{equation}
Fix arbitrary $x, y, z$.
Let $Z = Z_1 \cup Z_2$, where $Z_1 = R \cap S$ and $Z_2  = Z \setminus Z_1$.
Let $Z_1 = z_1$ and $Z_2 = z_2$ in $z$.
Note that $N_{S, (Z_1, z_1)}$ is the number of tuples in $S$ with the value of the join attributes as $z_1$. Each tuple $r$ in $R$ with $Z = z$, \ie, with $Z_1 = z_1$, joins with all tuples in $S$ with $Z_1 = z_1$, and joins with no other tuples in  $S$. Hence multiplying the numerator and denominator all terms in equation (\ref{equn:6}) above, we get
$$\frac{N_{R \Join S, xyz}}{N_{R \Join S, z}} = \frac{N_{R \Join S, xz}}{N_{R \Join S, z}} \times \frac{N_{R \Join S, yz}}{N_{R \Join S, z}}$$
In other words, $X \indep Y |_{R \Join S} Z$.
\end{proof}
}

%\red{TODO: add an example here.}

%\red{TODO: Make the above claim stronger by showing that FOR ALL CIs in R, if the above condition does not hold, then they do not propagate. }

%\textbf{New CIs in $R \Join S$:}

%\subsubsection{Type-II CIs: on LHS}\label{sec:join-type-2}
\textbf{(III) CIs with join attributes on the LHS:}
Here we show that if the join attributes belong to one of the subsets of attributes on the LHS, then also the CI propagates to the joined relation.
\begin{lemma}\label{lem:ci-join-2}
For any $X \indep Y |_R Z$ in $\ci_R$,  in the joined relation $X \indep Y |_{R \Join S} Z$ if $(R \cap S) \subseteq X$, \ie, all join attributes $R \cap S$ belongs to $X$ (similarly $Y$).
\end{lemma}

The proof is in the appendix. The following example shows that the conditions in Lemmas~\ref{lem:ci-join-1} and \ref{lem:ci-join-2} are necessary in the sense that if the join attribute do not participate in the a CI, it may not extend to the joined relation.
\begin{proposition}\label{prop:join-1-tight}
There exist relation instances for $R, S$ and a CI such that violating the conditions in both Lemmas~\ref{lem:ci-join-1} and \ref{lem:ci-join-2}  prohibits the propagation of the CI to the joined relation $R \Join S$.
\end{proposition}
\begin{proof}
Consider the relations $R, S, $ and $R \Join S$ below:\\

{\small
\begin{tabular}{|c|c|c|c|}
\multicolumn{4}{c}{$\mathbf{R}$}\\
\hline
$A$ & $B$ & $C$ & $D$\\
\hline
$a_1$ & $b_1$ & $c$ & $d_1$ \\
$a_1$ & $b_2$ & $c$ & $d_2$ \\
$a_2$ & $b_1$ & $c$ & $d_3$ \\
$a_2$ & $b_2$ & $c$ & $d_4$ \\\hline
\end{tabular}
\begin{tabular}{|c|c|}
\multicolumn{2}{c}{$\mathbf{S}$}\\
\hline
$D$ & $E$\\
\hline
$d_1$ & $e_1$ \\
$d_1$ & $e_2$ \\
$d_2$ & $e_1$ \\
$d_2$ & $e_2$ \\
$d_2$ & $e_3$ \\
$d_3$ & $e_1$ \\
$d_4$ & $e_1$ \\\hline
\end{tabular}
\begin{tabular}{|c|c|c|c|c|}
\multicolumn{5}{c}{$\mathbf{R} \Join \mathbf{S}$}\\
\hline
$A$ & $B$ & $C$ & $D$ & $E$\\
\hline
$a_1$ & $b_1$ & $c$ & $d_1$ & $e_1$ \\
$a_1$ & $b_1$ & $c$ & $d_1$ & $e_2$ \\
$a_1$ & $b_2$ & $c$ & $d_2$ & $e_1$  \\
$a_1$ & $b_2$ & $c$ & $d_2$ & $e_2$  \\
$a_1$ & $b_2$ & $c$ & $d_2$ & $e_3$  \\
$a_2$ & $b_1$ & $c$ & $d_3$ & $e_1$  \\
$a_2$ & $b_2$ & $c$ & $d_4$ & $e_1$  \\\hline
\end{tabular}\\
}
\smallskip

In $R$, $A \indep B |_R C$, but in $R \Join S$, $\Pr[A = a_1, B = b_1 | C = c] = \frac{2}{7}$, whereas $\Pr[A = a_1| C = c] = \frac{5}{7}$
$\Pr[B = b_1 | C = c] = \frac{3}{7}$. Hence $\neg(A \indep B |_{R \Join S} C)$ (note that the join attribute $D$ is not a subset of the RHS of the CI $A \indep B |_R C$. 
%\babak{From this it is implied that the proposition gives a necessary condition (which is not true)}
\par
On the other hand, $A \indep B |_R CD$, which propagates to $R \Join S$ as  $A \indep B |_{R \Join S} CD$ since $D$ belongs to the RHS of this CI.
\end{proof}
However, as we will see in the next section (Theorem~\ref{thm:pmap-ci-propagates}), if the CIs are generated by an undirected graphical model, then \emph{all} CIs in the base relation propagate to the joined relation.

\subsection{CI propagation for Foreign-Key Joins and One-one Joins}\label{sec:join-ci-fk}
\textbf{Foreign key joins:~} Proposition~\ref{prop:join-1-tight} shows that not all CIs from a base relation propagates to the joined relation. However, if the joined attributes form a \emph{foreign key} in $R$ referring to the primary key in $S$, then all CIs in $R$ propagate to $R \Join S$.
\begin{proposition}\label{prop:join-fk}
If the join attributes $(R \cap S)$ in a foreign key in $R$ referring to the primary key in $S$, then for all CI $X \indep Y |_R Z$ in $\ci_R$, it holds that $X \indep Y|_{R \Join S} Z$.
\end{proposition}
The proof is in the appendix.
\cut{
\begin{proof}[of Proposition~\ref{prop:join-fk}]
With a primary key join, every tuple in $R$ joins with exactly one tuple in $S$. Hence for $T = R \Join S$, and for all $x, y, z$ values of $X, Y, Z$, if
$\frac{N_{R, xyz}}{N_{R, z}} = \frac{N_{R, xz}}{N_{R, z}} \times \frac{N_{R, yz}}{N_{R, z}}$, then
$\frac{N_{T, xyz}}{N_{T, z}} = \frac{N_{T, xz}}{N_{T, z}} \times \frac{N_{T, yz}}{N_{T, z}}$, since all frequencies in all numerators and denominators in $R$ is multiplied by 1 to obtain the frequencies in $T$.
\end{proof}
}

%\begin{example}\label{eg:not-prop}
%\red{Give an example of $R = students$ and $S = families$ or something similar.}
%\end{example}

\textbf{One-one joins:~} Note that the propagation rule in Proposition~\ref{prop:join-fk} is not symmetric, since some CI in $S$ not satisfying the conditions in Theorem~\ref{thm:ci-join} may not propagate to $R \Join S$. However, if $(R \cap S)$ is a key (superkey) of both $R$ and $S$, and if $R, S$ have the same set of keys, then the CIs from both relations propagate to $R \Join S$.

\begin{proposition}\label{prop:join-one-one}
If $\pi_{R \cap S} R= \pi_{R \cap S}S$ and if $(R \cap S)$ is a key in both $R$ and $S$, then all the CIs from both $\ci_R$ and $\ci_S$ propagate to $R \Join S$.
\end{proposition}
The proposition directly follows from Proposition~\ref{prop:join-fk} -- since$\pi_{R \cap S} R= \pi_{R \cap S}S$, $R, S$ must have the same set of primary keys, and therefore $|R| = |S| = |R \Join S|$.
Also note that the condition $\pi_{R \cap S} R= \pi_{R \cap S}S$ is necessary, otherwise some tuples may be lost in the join destroying the CI.
 %\red{remove if no space}
\cut{
Consider the $R$ relation instance in Proposition~\ref{prop:join-1-tight}, and consider an instance of $S$ containing three tuples $S(d_1, e_1)$, $S(d_2, e_2)$, $S(d_3, e_3)$. Since $d_4$ does not join with any tuple, the CI $A \indep B |_R C$ does not hold any more in $R \Join S$. The CI propagates in this particular example is solved if we take an outer join between $R, S$ instead of an inner join. However, if we add the tuple $S(d_5, e_5)$ in $S$, the CI does not propagate again.
}
Theorem~\ref{thm:pmap-ci-propagates} in Section~\ref{sec:undirected} states that for relations that are graph-isomorph, the CIs propagate. The above proposition states that even if the relation is not graph-isomorph, but if the join is one-one, then also the CIs propagate.
%\par
%Nevertheless, as we will see in the following section, propagation of all CIs to the joined relation is an important property even if it is achievable only in this special case: If we are doing causal analysis on an entity (people, students, land, patients), and if the information about these entities are collected from different sources, we will have relations satisfying this condition, where we can do fine-grained CI analysis using the graphical models if available.

\subsection{Application to Observational Studies}\label{sec:app-ci}
Whether a subset of variables $X \subseteq \attr$ satisfies strong ignorability, in general, is untestable even for a single relation since the test involves missing data in the form of potential outcomes $Y(0), Y(1)$ (in contrast to observed outcome $Y$), although there are sufficient conditions assuming a causal graphical model (\cite{PearlBook2000}, Section~\ref{sec:app-causal-nw}). In this section, first we show a negative result -- if the join attributes belong to $X$, then $X$ satisfies strong ignorability, but is not useful since the estimated average causal effect of $T$ on $Y$ will be zero (Section~\ref{sec:join-var-useless}). Then we give a positive result (Section~\ref{sec:avoid-join}) that if we are given a set of variables $X$  satisfying ignorability spanning multiple relations participating in a join with primary key-foreign keys, it suffices to condition on subset of $X$ restricted to the relation(s) containing $T$ and $Y$, which increases efficiency and reduces variance since the matched groups based on the same values of the covariates will be bigger. 
\subsubsection{Conditioning on join variables is not useful}\label{sec:join-var-useless}
\cut{
In practice, often  matching is used on observational data with all attributes in the relation as covariates $X$ to estimate the treatment effect of column $T$
 on outcome $Y$. This implicitly assumes the \emph{strong ignorability assumption} (Definition~\ref{def:SITA}) that \red{the columns $T$ and $Y(0), Y(1)$} are conditionally independent given $X$. 
 %We show the following:
 \par
 The following theorem gives a soundness property, that even if the attributes are collected from two relations $R, S$ 
 with arbitrary schema that participate in an arbitrary natural
  join (not necessarily foreign key or one-one joins), if the treatment $T$ and outcome $Y$ \red{(and therefore the potential outcomes $Y(0, Y(1))$)} belong to different relations, conditioned on all other attributes in $R \cup S$, or conditioned on a subset of attributes in $R \cup S$ that contains the join attributes $R \cap S$, the strong ignorability condition holds, and therefore, these covariates can be used in matching. This is independent of all CIs that hold in $R$ or $S$.
\babak{Pleas revise your theorem in the light of our discussion}
}
Here we show the following proposition, which states that conditioning on any subset of variables that includes the join variables is not useful for causal analysis despite satisfying strong ignorability.
\begin{proposition}\label{prop:app-causality-all}
If the treatment and the outcome (\ie, also the potential outcome) variables come from two different relations, \ie, $T \in R$ and $Y \in S$ (\ie, $Y(0), Y(1) \in S$), then given any subset of attributes $X \subseteq (RS \setminus TY)$ such that $X \supseteq (R \cap S)$, (i) $T$ and $Y(0), Y(1)$ are independent (thereby satisfying strong ignorability). However, (ii) $T$ and $Y$ are also independent, and therefore (iii) the average treatment effect of $T$ on $Y$ using $X$ is zero. 
\end{proposition}
The proof uses graphoid axioms, and is given in the appendix. Note that the columns $Y(0), Y(1)$ are hypothetical, since only $Y = Y(0) (1 - T) + Y(1) T$ is observed due to the fundamental problem of causal analysis from missing data. 
As a special case, the proposition shows that conditioning on $X = $ all other attributes in $R \Join S$ except $T, Y$ is not useful (which is often done in statistical causal analysis for applied problems involving a single relation), since the estimated average treatment effect will be zero. This further motivates the study of CI and understanding ignorability for joined relations (further discussed in Section~\ref{sec:app-causal-nw}).

\cut{
\babak{Speculations: 1) Assume T and Y are in two tables R and S. Assume there exists a set Z=WUG (W in R, U in R $\cap$ S and G in S)
such that strong ignorability hold for Z. Then string ignorability also hold for Z=WU or Z=GU. This is we don't need to joint at all!
2) If Y and T belong to tow tables with k-k relationship the above still holds. Not sure yet about arbitrary join. But this seems like the right way to follow.}
}

\cut{
\begin{proof}[of Theorem~\ref{prop:app-causality-all}]
Let% $U = (X \cap R) \setminus (\{T\} \cup (R \cap S))$, and $V = (X \cap S) \setminus (\{Y\} \cup (R \cap S))$, \ie,
$U, V$ denote the subset of attributes in $R$ and $S$ respectively in $X$ that do not belong to $T, Y, $ or $R \cap S$.
Then,
\begin{eqnarray*}
&&UT \indep VY |_{R\Join S} (R \cap S)~~~~\textrm{(Corollary~\ref{cor:cl-join-0})}\\
& \Rightarrow & UT \indep Y |_{R\Join S} (R \cap S) V~~~~\textrm{(weak union (\ref{equn:wu})}\\
& \Rightarrow & T \indep Y |_{R\Join S} (R \cap S) VU~~~~\textrm{(weak union (\ref{equn:wu})}\\
& \equiv & T \indep Y |_{R\Join S} X
\end{eqnarray*}
since $X = \{U\} \cup \{V\} \cup (R \cap S)$.
\end{proof}
}

\subsubsection{Avoiding joins}\label{sec:avoid-join}
Assume we are given $k$ relations $R_1, \ldots, R_k$, where the outcome $Y \in R_i$. 
Suppose $U$ is the universal table obtained by joining these relations.
%that are related with primary key-foreign-key constraints, and $U$ is the universal table obtained 
%by joining these relations. Assume we are also given a treatment variable $T$, an outcome variable $Y$, and a set of covariates $X$ scattered across $k$ relations, \ie, 
%$X=X_1 \cup \ldots \cup X_k$, $X_i \subseteq R_i$ (some $X_i$-s can be empty). 
Further, assume that we are given a set of covariates $X \subseteq U$ that satisfies strong ignorability, 
\ie, $Y(1),Y(0) \indep T |_U X$. 
Suppose $X_j = X \cap R_j$, $j \in [1, k]$. The following proposition says that, if $X_i$ contains foreign keys to all other relations, then it suffices to replace $X$ with $X_i$ for covariate adjustment, since $X$ and $X_i$ are c-equivalent (Definition~\ref{def:c-eq}).

\begin{proposition} \label{prop:cequi}
If $U = R_1 \Join \cdots \Join R_k$, $X \subseteq U$, $X_j = X \cap R_j$ for $j \in [1, k]$, $R_i$ contains $Y$, and %\red{$X_i \rightarrow U$}  holds in $U$, 
$X_i$ contains foreign keys to all relations $R_j$, $j \in [1, k]$, $j \neq i$,
then 
$X$ and $X_i$ are c-equivalent.
\end{proposition} 
%\sudeepa{The proof only works for $X_i \rightarrow U$. See if you can show this for $R_i \rightarrow U$ which will match Observation~\ref{obs:units-unique-outcome}}.

Intuitively, if $X_i$ contains the foreign keys to all other relations, then adjusting for $X_i$ is sufficient, since any two tuple in $U$ that have the same value of $X$, will have the same value of $X_i$. 
The proof uses entropy and mutual information (Section~\ref{sec:prelim}), and uses the fact that $I(X, Y | Z) = 0$ if and only if $X \indep Y | Z$ \cite{InfoTheoryBookCoverThomas}; the proof is given in the appendix. 
However, from Proposition~\ref{prop:app-causality-all}, if $T$ and $Y$ belong to two different relations, then conditioning on the join attributes (foreign keys) will result in a zero causal effect. Nevertheless, the above proposition implies that if $T$ and $Y$ belong to the same relation $R_i$, and if $\neg (T \indep Y_{R_i} X_i)$, then it is enough to condition on the foreign keys and the covariates in $R_i$, and the join can be avoided.

\smallskip
The implications of the results established in this section to causal inference is two-fold.  \emph{First,} they provide a principled
way to reduce the number of covariates needed for estimating causal effect. This can be done by starting with a set of CIs on base relations that
propagate to the joined relation (using Theorem \ref{thm:ci-join}, more CIs propagate if the relations are graph-isomorphic as discussed in the next section). Then Proposition~\ref{th:ceq} can be employed to infer $X'$, a smaller set of covariates that is c-equivalent with $X$. A special case is given in Proposition~\ref{prop:cequi}.  It is known that the quality of the matching estimators for causal effect
decreases with the number of covariates \cite{de2011covariate}. Hence, reducing the set of covariates is important for 
inferring robust causal conclusions. \emph{Second,} Proposition~\ref{prop:cequi} shows that under some circumstances it is not required to materialize the joined table involving all relations for collecting more covariates, and it suffices to focus on the subset of the given covariates $X$ in the relation containing $Y$.  This is not only useful from the efficiency point of view (the matching groups have to be performed on smaller covariates, and more matched groups are likely to be `valid' with at least one treatment and one control units), but is also interesting because it reveals that it is still possible to make causal inferences, when the values of the covariates in some of the base relations are not recorded or noisy (\eg, when the chema and a valid $X$ are given by a domain expert but the values in some of the other relations are unavailable).

\cut{
\sudeepa{I don't follow the last line in red. $X$ is given in the relations. So $X$ is observed. What is unobserved here.}
\babak{For instance we know that the flight data joined with the wetaher data on time/data/location attributes. however it might be the vase that the weather data is not available. Now we can still adjust for the weather by adjusting for data/time/location.}
%\sudeepa{S: Not sure how you are using CIs in base relations, and you will use Theorem \ref{thm:ci-join}. Can you expand this discussion?}
}

%\red{revisit}

\textbf{Further Questions:}
In this section we investigated join of two relations. Understanding the CIs for multiple relations may require investigating the query hypergraph. %is an interesting direction; whether additional CIs on the joined relation for general or special cases can be inferred can be studied.
 If the join is on multiple relations where  different relations share a subset of attributes from other relations, such complex interaction may prohibit certain CIs to hold on the joined relation. On the other hand, special structure of the \emph{query hypergraph} (a hypergraph on all attributes where the relations form the hyperedges), like the \emph{ayclicity property} \cite{Beeri+83}, may allow some CIs in the joined relation.

\section{Using Undirected Graphical Models on the Base Relations}\label{sec:undirected}
In the previous section, we discussed sufficient conditions to infer CIs in the joined relation when the  dependency models satisfied by the two base relations are arbitrary. However, suppose the CIs in the base relations are graph-isomorph, \ie, the base relations $R, S$ have P-maps $G_1, G_2$. The question we study in this relation is whether $G_1, G_2$ help generate an I-map $G$ of $R \Join S$. Since our key motivation is causal analysis, then we will be able to infer correct CIs on $R \Join S$ using $G$. The other question is whether the information that $R$ (or $S$) is graph-isomorph, helps propagate additional CIs from $R$ to the joined relation. In this section, we answer these two questions affirmatively when there is only one join attribute: \emph{all CIs from the base relation $R$ propagate to the joined relation $R \Join S$ if $R$ is graph-isomorph} (Theorem~\ref{thm:pmap-ci-propagates}), and %answer the first question for a special scenario: 
\emph{in addition, when both $G_1, G_2$ are connected, then union of $G_1, G_2$ produces an I-map of $R \Join S$} (Theorem~\ref{thm:union-imap}). The extension to multiple base relations and other research questions are discussed as future directions at the end of this section.

%Suppose the CIs in the base relations $R, S$ are graph-isomorph by undirected Markov networks, \ie, there are two graphs (P-maps) $G_1, G_2$ for $R$ and $S$ respectively such that: \emph{$X \indep Y |_{R} Z$ if and only if removing all vertices in $X$ from $G_1$ disconnects all vertices in $X$ from all vertices in $Y$ (similarly for $S$ and $G_2$).}
%Indeed, a trivial I-map for $R \Join S$ is the complete graph on $RS$ and a trivial d-map for $R \Join S$ is the empty graph. However, the goal is to construct a graph $G$ for $R \Join S$ that preserves as many valid CIs that hold in $R \Join S$ as possible, and discards the ones that do not hold (as many as possible).
%The questions we study in this section are
%\begin{itemize}
%\item If $G_1, G_2$ are P-maps for $R, S$, does a P-map always exist for $R \Join S$?
%\item How can we construct a non-trivial I-map for $R \Join S$ using $G_1, G_2$?
%\item How can we construct a non-trivial d-map for $R \Join S$ using $G_1, G_2$?
%\item Which CIs can we capture in these graphs for general schema and joins?
%\item Which CIs can we capture for special cases (\eg, foreign key and one-one joins)?
%\item What if $G_1, G_2$ are only I-maps or only d-maps?
%\end{itemize}

%Here we extend our notations for CIs from relations to undirected graphical models. In particular, we use the notation
Recall that
$$X \indep Y |_{G} Z$$
denotes vertex separation in an undirected graph $G(V, E)$, \ie, if $X, Y, Z \subseteq V$ are disjoint subsets of vertices, then removing $Z$ and all the incident edges on $Z$ from $G$ (denoted by $G_{-Z}$) disconnects all paths between all vertices in $X$ and all vertices in $Y$ in $G$ (or $Z$ is a cutset in $G$ between $X, Y$). Below we state a property used in our proofs:

\begin{observation}\label{obs:cutset-sup}
In an undirected graph $G(V, E)$, if $X \indep Y |_{G} Z$, then for all supersets $Z'$ of $Z$, $X \indep Y |_{G} Z'$.
\end{observation}

Although we overload the notation for independence $\indep$ to also denote vertex separation in graphs, note that, $X \indep Y |_{G} Z$ by itself does not say \emph{anything} about CIs of $X, Y$ given $Z$. In fact, the results in this section aim to prove that if $X \indep Y |_{G} Z$, then $X \indep Y |_{R \Join S} Z$.

\subsection{CIs from a Relation with a P-Map Propagates to the Joined Relation}
%\babak{I have not checked the proofs of this section yet.  But I think its a little bit disconnected from the causality application. I was wondering if we could bring the selectivity estimation using markov network in more detail. Because the result we obtain in this section essentially show that having multiple relations under some conditions we can push the learning of the markov network to the base relations.}
Proposition~\ref{prop:join-1-tight} shows that not all CIs in $R \Join S$  propagate to the joined relation $R \Join S$. Here we show that if $R$ has a P-map $G_1$ (Definition~\ref{def:idp-maps}), \ie, if $X \indep Y |_{R} Z \Leftrightarrow X \indep Y |_{G_1} Z$ for all disjoint subsets $X, Y, Z \subseteq R$, then all CIs in $R$ propagate to $R \Join S$ in arbitrary natural joins.
\begin{theorem}\label{thm:pmap-ci-propagates}
If $R$ has a P-map $G_1$, $D = R \cap S$ is the single join attribute, and $X \indep Y |_{R} Z$ for disjoint $X, Y, Z \subseteq R$, then $X \indep Y |_{R \Join S} Z$.
\end{theorem}

Theorem~\ref{thm:pmap-ci-propagates} is proved using the following lemma; both proofs are given in the appendix.
\begin{lemma}\label{lem:either-or}
If $R$ has a P-map $G_1$, $D = R \cap S$ is the single join attribute, and $X \indep Y |_{R} Z$ for disjoint $X, Y, Z \subseteq R$, then either $X \indep YD |_R Z$ or $XD \indep Y |_R Z$.
\end{lemma}

Note that in $R \Join S$, which may not have a P-map in general (Section~\ref{sec:future-directed} gives an example that can be extended to a join), %\red{need an example ref here},
we can only apply graphoid axioms, whereas since $R$ has a P-map $G_1$, we can apply both graphoid axioms as well as the necessary and sufficient conditions from Theorem~\ref{thm:PP-undirected}.
\par
We revisit why in Proposition~\ref{prop:join-1-tight} the CI $A \indep B |_{R} Z$ did not propagate to $R \Join S$.  In relation $R$ in the example, $A \indep B|_R C$. If $R$ had a P-map, the transitivity property of Theorem~\ref{thm:PP-undirected} will hold, and we will have either $A \indep D|_R C$ or $D \indep B|_R C$. However, in the example, both do not hold.  On the other hand, an Example~\ref{eg:pmap-prop} in the appendix shows that if we replace the $R$ instance with one that is generated by the P-map $A-B-C-D$, then the CI will propagate to $R \Join S$ with the same $S$ instance.

\subsection{An I-map for Joined Relation}\label{joined-imap}
Theorems~\ref{thm:pmap-ci-propagates} and \ref{thm:ci-join} give us sufficient conditions for some of the CIs that hold in the joined relation. The question we study in this section is, whether we can infer additional CIs when both $R$ and $S$ have P-maps $G_1$ and $G_2$. A natural choice is to consider the \emph{union graph} $G$ of $G_1, G_2$, where the set of vertices in $G$ is $R \cup S$ and the set of edges in the union of edges from $G_1, G_2$. We use the following two observations in our proofs.
\begin{observation}\label{obs:join-attr-cutset}
The join attributes $R \cap S$ form a cutset between $R \setminus S$ and $S \setminus R$ in $G$, \ie, all paths between the non-join attributes in two relations must go through $R \cap S$.
\end{observation}
\begin{observation}\label{obs:subset-indep}
If $X \indep Y|_{G} Z$ where $X, Y, Z \subseteq R$ and are disjoint, then $X \indep Y|_{G_1} Z$.
\end{observation}
The above observation follows from the fact that $G_1$ has a subset of edges of $G$, and if two subset of nodes are not connected in $G$, they cannot be connected in $G_1$.
In this section, we assume that (i) \emph{the join is on a single attribute $D$, \ie, $R \cap S = \{D\}$}, and (ii) \emph{the P-maps $G_1, G_2$ are connected}, \ie, there is a path from all vertices in $R$ and $S$ to $D$ in $G_1, G_2$ respectively. % \red{revisit these assumptions}.
\par
First we show that additional CIs inferred from the union graph $G$, not necessarily captured by Theorems~\ref{thm:pmap-ci-propagates} and \ref{thm:ci-join}, hold in the joined relation $R \Join S$.
\begin{lemma}\label{lem:two-one-prop}
Suppose $X, Y, Z$ are disjoint set of vertices in the union graph $G$ such that (i) $X$ and $Y$ are disconnected in the graph $G_{-Z}$, and (ii) $X, Z$ belong to one of $R, S$ and $Y$ belongs to the other relation, then $X \indep Y|_{R \Join S} Z$.
\end{lemma}
The proof is given in the appendix, which (along with all other proofs in this section) uses all four symmetry, decomposition, weak union, and contraction properties ((\ref{equn:sym})-(\ref{equn:con})) of graphoid axioms, as well as properties of vertex separation of graphs. To see an example of the application of Lemma~\ref{lem:two-one-prop}, consider two simple P-maps $A-B-D$ and $D-F$. Note that the independence $A \indep F|_{R \Join S} B$ does not directly follow from Theorems~\ref{thm:pmap-ci-propagates} and \ref{thm:ci-join}, % \red{(check if follows from graphoid axioms)},
but follows from Lemma~\ref{lem:two-one-prop}. %A more complex example is shown in Example~\ref{eg:imap-ci}.

Now we move to the general case of separation in $G$.
Suppose $X \indep Y |_G Z$ where $X = X_1, X_2$, $Y = Y_1Y_2$, $Z = Z_1Z_2$, and $X_1, Y_1, Z_1 \subseteq R$, $X_2, Y_2, Z_2 \subseteq S$.  In the following lemmas we consider different cases when some of $X_i, Y_i, Z_i$-s are empty.
% when the join attribute $D = R \cap S$ is a singleton set.
All proofs are in the appendix.
\begin{lemma}\label{lem:y2-z2-empty}
If $Y_2 = Z_2 = \emptyset$, then $X_1X_2 \indep Y_1 |_{R \Join S} Z_1$.
\end{lemma}

\begin{lemma}\label{lem:y1-z2-empty}
If $Y_1 = Z_2 = \emptyset$, then  $X_1X_2 \indep Y_2 |_{R \Join S} Z_1$.
\end{lemma}

\cut{
\begin{proof}[of Lemma~\ref{lem:y1-z2-empty}]
First we argue that the join attribute $D \in Z_1$.
Suppose not. Then in $G_{-Z_1}$, no path exists between $X_2$ and $Y_2$ in $G$.  Since $D \notin Z_1$, removing $Z_1$ does not remove any edge in $G_2$, implying that  no path exists between $X_2$ and $Y_2$ in $G_2$, which contradicts the assumption that $G_2$ is connected.
\par
Hence $D \in Z_1$. Assume $Z_1 = DW_1$ where $W_1 \subseteq R$. Since $X_1 \indep Y_2 |_{G} Z_1 ~\equiv~ X_1 \indep Y_2 |_{G} DW_1$, and $W_1$ belongs to $R$ or $G_1$, it holds that $Y_2 \indep X_2 |_{G} D$. Since $G_2$ is a subset of $G$,  $Y_2 \indep X_2 |_{G_2} D$, and since $G_2$ is a P-map of $S$,  $Y_2 \indep X_2 |_{S} D$, and by Theorem~\ref{thm:pmap-ci-propagates},
\begin{equation}\label{equn:y2-x2-d-rs}
Y_2 \indep X_2 |_{R \Join S} D
\end{equation}

By Corollary~\ref{cor:cl-join-0},
\begin{eqnarray*}
&&X_1W_1 \indep Y_2X_2 |_{R \Join S} D\\
& \Rightarrow & X_1W_1 \indep Y_2 |_{R \Join S} DX_2~~~~\textrm{(weak union (\ref{equn:wu}))}\\
& \Rightarrow & X_1X_2W_1 \indep Y_2 |_{R \Join S} D~~\textrm{((\ref{equn:y2-x2-d-rs}) and contraction (\ref{equn:con}))}\\
& \Rightarrow & X_1X_2 \indep Y_2 |_{R \Join S} DW_1~~~~\textrm{(weak union (\ref{equn:wu}))}\\
& \equiv & X_1X_2 \indep Y_2 |_{R \Join S} Z_1
\end{eqnarray*}
\end{proof}
}

%Next we consider the second scenario, when $Z_2 = \emptyset$ (and the equivalent cases); the proof is similar to Lemma~\ref{lem:y1-z2-empty} but uses Lemma~\ref{lem:y2-z2-empty}. %\red{move to appendix}

\begin{lemma}\label{lem:z2-empty}
If $Z_2 = \emptyset$, then  $X_1X_2 \indep Y_1Y_2 |_{R \Join S} Z_1$.
\end{lemma}

\begin{lemma}\label{lem:x2-y2-z1-empty}
If $X_2 = Y_2 = Z_1 = \emptyset$, then  $X_1 \indep Y_1 |_{R \Join S} Z_2$.
\end{lemma}

\begin{lemma}\label{lem:z1-z2-non-empty}
\begin{itemize}
\item[(A)] If $X_1 \indep Y_1 |_{G} Z_1Z_2$, then $X_1 \indep Y_1 |_{R \Join S} Z_1Z_2$.
\item[(B)] If $X_1 \indep Y_2 |_{G} Z_1Z_2$, then $X_1 \indep Y_2 |_{R \Join S} Z_1Z_2$.
\item[(C)] If $X_1X_2 \indep Y_1 |_{G} Z_1Z_2$, then $X_1 \indep Y_1Y_2 |_{R \Join S} Z_1Z_2$.
\item[(D)] If $X_1X_2 \indep Y_1Y_2 |_{G} Z_1Z_2$, then $X_1X_2 \indep Y_1Y_2 |_{R \Join S} Z_1Z_2$.
\end{itemize}
\end{lemma}

The theorem stating the main result of this section follows from the above lemmas (proof in the appendix):
\begin{theorem}\label{thm:union-imap}
Suppose $R$ and $S$ have P-maps $G_1$ and $G_2$ respectively that are connected graphs, $G$ is the union graph of $G_1, G_2$, and the join is on a single common attribute $R \cap S = \{D\}$. Then for disjoint subsets of vertices $X, Y, Z \in G$ if $X\indep Y |_{G} Z$, then  $X\indep Y |_{R \Join S} Z$, \ie, $G$ is an I-map for $R \Join S$.
\end{theorem}

\cut{
\begin{proof}[of Theorem~\ref{thm:union-imap}]
%If all of $X, Y, Z$ belongs to the
Suppose $X = X_1X_2$, $Y = Y_1Y_2$, $Z  = Z_1Z2$. Lemma~\ref{lem:z1-z2-empty} covers the (non-equivalent) cases when both $Z_1, Z_2$ are non-empty.  When $Z_2 = \emptyset$ (equivalently $Z_1$): (i) when both $X, Y$ contain both subsets, the result follows from Lemma~\ref{lem:z2-empty}; (ii) when one of $X,Y$ contain both subsets and the other contain one, the results follows from Lemmas~\ref{lem:y1-z2-empty} and \ref{lem:y2-z2-empty}; (iii) when both $X, Y$ contain one subset each, the result follows from Lemma~\ref{lem:two-one-prop} or Theorem~\ref{thm:pmap-ci-propagates} (in this case, the CI holds in the base relation since $G_1$ is a P-map, and therefore propagates to the joined relation).
%
%\begin{tabular}{|c|c|c|}
%\hline
%$X$ & $Y$ & $Z$ & Follows from..\\\hline\hline
%$X_1$ & $Y_1$ & $Z_1$ & \\\hline
%$X_1$ & $Y_1$ & $Z_1$ & \\\hline
%$X_1$ & $Y_1$ & $Z_1$ & \\\hline
%$X_1$ & $Y_1$ & $Z_1$ & \\\hline
%$X_1$ & $Y_1$ & $Z_1$ & \\\hline
%\end{tabular}
\end{proof}
}

\textbf{Further Questions.~} In this section we showed that CIs from graph-isomorph base relations propagate to joined relation, and the union graph is an I-map if the join is on a single attribute and the given P-maps are connected. Several other questions remain to be explored: (1) If $G_1, G_2$ are P-maps for $R, S$, does a P-map always exist for $R \Join S$ (or under what conditions a P-map exists)? Example~\ref{eg:no-dmap} shows that new CIs may be generated for some instances, but a schema-based argument will be needed to find a solution to this problem. (2) What if $G_1, G_2$ are only I-maps? In this section, only the proof of Theorem~\ref{thm:pmap-ci-propagates} uses the properties of P-map in arguing that if a CI does not hold in the base relation $R$, the vertex separation does not hold in $G_1$; all other proofs mentioning P-maps use properties of vertex separation in undirected graphs. This assumption seemed to be necessary in our proof of Theorem~\ref{thm:pmap-ci-propagates}, so the question is whether an alternative proof exists that only uses properties from I-map to infer Theorem~\ref{thm:union-imap} (which uses Theorem~\ref{thm:pmap-ci-propagates}).  (3) We also use the assumption that $G_1, G_2$ are connected in our proofs - what happens if they are not connected (when some variables are unconditionally independent)? (4) Similar to Section~\ref{sec:condl_indep}, another question is what can be inferred about multiple relations, and whether properties like acyclicity of schemas help in inferring CIs.   Note that if conditions in Theorem~\ref{thm:union-imap} is weakened, requiring only that $G_1, G_2$ are I-maps of $R, S$, then if a join order on $k$ exists where every join happens only on one attribute (a special case of join tree for acyclic schema \cite{Beeri+83}), then the result can be extended to multiple relations, since by combining two I-maps we get another I-map. (5) Extending the results to arbitrary number of join attributes with arbitrary connections is another question to answer. 
%(some of the lemmas extend to multiple attributes, whereas some require the existence of one common attribute between two graphs).
%\item How can we construct a non-trivial I-map for $R \Join S$ using $G_1, G_2$?
%\item How can we construct a non-trivial d-map for $R \Join S$ using $G_1, G_2$?
%\item Which CIs can we capture in these graphs for general schema and joins?
%\item Which CIs can we capture for special cases (\eg, foreign key and one-one joins)?
%\item What if $G_1, G_2$ are only I-maps or only d-maps?
%\end{itemize}

%\textbf{Combining P-maps does not generate a D-map}. Here we give an example with two simple relations where the join may introduce additional CIs:

%\input{directed}

\section{More Research Directions for the Framework}\label{sec:applications}
%The framework of causal inference for multi-relational observational data opens several new research directions in the area of database theory and database management.  
In addition to the questions stated in the previous sections, here we present three other fundamental directions that need to be explored to make causal analysis on multi-relational data robust and practical.

\subsection{Using Directed Graphical Models for CIs}\label{sec:future-directed}
The results in the previous sections suggest that additional knowledge on CIs on base relations help infer additional CIs in the joined relation, and that graphical models on base relation is a convenient technique to infer CIs on the joined relation. Although questions remain to be answered even for undirected graphical models (Section~\ref{sec:undirected}, \eg, for join on multiple attributes or with multiple relations), the undirected graphical model has inherent limitations itself in capturing some dependency relations, \eg, when $I(X, Z, Y)$ but $\neg I(X, ZW, Y)$ (in undirected graphs, supersets of cutsets are also cutsets). As a special case, consider the scenario when two variables $X, Y$ are independent (\ie, cannot have any path connecting them in an undirected graph), whereas they are dependent given a third variable $Z$ (which requires paths between $X, Z$ and $Y, Z$), for instance, when $X, Y$ denote the random toss of two coins, and $Z$ denotes the ring of a bell that rings only when the both coins output the same value \cite{Pearl-PR-book}. These two constraints cannot be satisfied simultaneously in any undirected graph, but they can be captured in a directed graphical model or \emph{Bayesian networks} by adding two directional edges from $X$ to $Z$ and $Y$ to $Z$. 
Bayesian networks are causal when the arrows reflect the direction of causality between variable: two edges from $X$ to $Z$ and $Y$ to $Z$ imply that both $X, Y$ causally affect $Z$ but not each other (mutually independent).
%\babak{I would say: ``A Bayesian network is called {\em causal} if the arrows capture and reflect the direction of causal relationship between variables." Bayesian networks are always directed and they might not shows causal relations."  }
%Directed Bayesian networks reflects the directions of \emph{causal influences} using arrows, which say that $X, Y$ %\babak{ are causal affect $Z$ but they have to causal relationship.} determine the outcome of $Z$, but are mutually independent themselves. 
\par
The inference of CIs in Bayesian networks is performed using \emph{d-separation} proposed by Pearl and Verma \cite{PearlBook1998,DBLP:conf/aaai/PearlV87}. The basic idea is the following (which is extended to general paths in d-separation): for three variables $X, Y, Z$, if the directions are $X \rightarrow Z \rightarrow Y$, $X \leftarrow Z \leftarrow Y$, or $X \leftarrow Z \rightarrow Y$, then observing $Z$ makes $X, Y$ conditionally independent (or the path is \emph{blocked}), but if the direction is $X \rightarrow Z \leftarrow Y$ ($Z$ is a \emph{collider}), and if $Z$ or any of its descendants are observed, $X$ and $Y$ become dependent. In  general, $Z$ is said to d-separate $X$ from $Y$ in a directed graph, if all paths from $X$ to $Y$ are blocked by $Z$.
Not all CIs can be captured using a directed graph (\eg, CIs captured by a diamond-shaped undirected Markov network). However, the set of \emph{chordal graphs}, where every cycle of length $\geq 4$ has a chord, can be represented by both undirected and directed graphs (chordal graphs have also been studied for acyclic schemas in \cite{DBLP:journals/jacm/BeeriFMY83}). 
\par
The following example shows that the union of P-maps $G_1, G_2$ of base relations $R, S$ is not a D-map for the joined relation $R \Join S$: % (unlike undirected graphs):
\begin{example}
Consider the example when $G_1 = (A, B)$, \ie, a directed edge from $A$ to $B$, and $G_2 = (C, B)$, \ie, a directed edge from $C$ to $B$. Taking the union of these two edges, $C$ becomes a collider between $A$ and $B$, suggesting that $A$ and $C$ are not independent given $B$. This contradicts with Corollary~\ref{cor:cl-join-0} that $A \indep C |_{R \Join S} B$. Since a CI in the model does not hold in the graph, it is not a D-map. 
\end{example} 

However, this does not say whether we can have an I-map by union (in the above case, we have a trivial I-map). 
Getoor et al. proposed the concept of \emph{Probabilistic Relational Model (PRM)} \cite{getoor2001selectivity}, their construction, and application to selectivity estimation for joins with primary  key-foreign keys. Maier et al. \cite{DBLP:conf/uai/MaierMAJ13} proposes the notion of \emph{relational d-separation}, to infer instance-based CIs in entity-relationship models involving multiple relations. But to the best of our knowledge, inferring schema-based d-separation for joins in Bayesian network, and combining the ones for base relations to get an I-map or a P-map for the joined relation, has not been explored in the literature.

\subsection{CIs involving Potential Outcomes and Causal Networks}\label{sec:app-causal-nw}
The CIs discussed in the paper consider columns from base or joined relations. However, for the strong ignorability condition, we need $Y(0), Y(1) \indep T | X$, where $Y(0), Y(1)$ are potential outcomes with missing data in the observed relation. In Proposition~\ref{prop:app-causality-all}  in a two-way join, we showed that conditioning on all attributes from both relations, or attributes containing the join variable satisfies ignorability but is not useful for estimating causal effects. 
The challenge is that while $Y(0), Y(1)$ should be independent of $T$ given $X$, $Y$ and $T$ cannot be independent given $X$ (otherwise ATE = 0). 
\par
In general, strong ignorability is not readily assertable from common knowledge \cite{PearlBook2000}. However, if the underlying variables are represented by a causal directed graph, Pearl  \cite{PearlBook2000} gives a sufficient condition called the \emph{backdoor criteria} for checking ignorability (also called unconfoundedness or admissibility or identifiability): in the DAG, no variable in $X$ is a descendant of $T$, and $X$ `blocks' all paths between $T$ and $Y$ that contains an arrow into $T$ (as done in d-separation, Section~\ref{sec:future-directed}). In other words, using Pearl's notations, $Y$ and $T$ are independent (using d-separation) given $X$ in the graph $G_{\underline{T}}$, which is obtained by deleting all edges emerging from $T$ (thereby taking care of the fact that $T$ should have a direct effect on $Y$). Understanding and extending backdoor criteria (and other observations from the causal graphical model from the vast literature by Pearl and co-authors) is an important direction to explore to understand ignorability for joined relations.

\subsection{Satisfying Basic Causal Assumptions for Arbitrary Joins}\label{sec:app-mult-rows}
Strong ignorability (Definition~\ref{def:SITA}) is one of the necessary conditions for observational causal studies. Another necessary condition is SUTVA (Definition~\ref{def:SUTVA}, required also for controlled experiments), which says that the treatment assignment to one unit does not affect the outcome of another unit. Another hidden assumption is that every unit constitutes one data record (one row in Table~\ref{tab:potential-outcome}). We discussed foreign key and one-one joins in Section~\ref{sec:condl_indep}. For foreign key joins, if treatment $T$ and outcome $Y$ belong to the table with the foreign key, we get one row for a unit with one $T$ and $Y$ value. The one-one join allows arbitrary selection of $Y$ and $T$. However, for many-many joins, $T$ and/or $Y$ may repeat in the joined table. In Observation~\ref{obs:units-unique-outcome} we discussed necessary conditions for units obtained by natural joins (tuples containing outcome $Y$ cannot repeat in joined relation). As discussed in Section~\ref{sec:valid-units-joins}, we need to investigate how joined relation can be post-processed to obtain valid units satisfying SUTVA; the same holds for relations originated from more complex queries.
For instance, in Example~\ref{eg:2}, if the information about both parents is stored in the $\mathtt{ParentsInfo}$ table, and the causal question is how much the $job$s of the parents affect the $gpa$s of the students,  for a student in the joined table $\mathtt{Students \Join Parent \Join ParentInfo}$, there may be two rows, both having the same outcome $Y = gpa$, but potentially different $T = \mathtt{parents\_income}$, which should be aggregated to obtain  valid units.
%This  leads to additional challenges of how to form the new units, and understanding whether causal analysis in such settings will be meaningful.     

\subsection{Weaker Notions of CIs}\label{sec:weaker-CIs}
Defining CIs in base and joined relations using probabilistic  interpretation of conditional probabilities is a natural choice, although it raises some conceptual questions. As discussed earlier, CIs should be properties of relations when the relations capture well-defined entities, and should not change when we integrate this relation with another relation (\eg, if we are using weather data to reason about flight delays, the CIs that hold among temperature, pressure, humidity should continue to hold in the integrated dataset after joining with flight dataset). However as Proposition~\ref{prop:join-1-tight} shows, not all CIs propagate using join, and further, Example~\ref{eg:no-dmap} in the appendix shows that new CIs (confined to a base relation) may be generated in the joined relation. 
%We can explore whether such CIs are \emph{spurious}, \ie, instance-based, and do not hold in all possible instances of the schema. 
However, since all the schema-based CIs do not propagate to the joined relation anyway, 
the question arises whether other notions of CIs in relations are meaningful that will remain the same whether or not a join has been performed. Indeed, we can consider the combined joint distribution of all attributes from all base relations \cite{getoor2001selectivity}, but the question of inferring CIs in the joined relation starting from the CIs in the base relations still is an interesting and useful question to answer beyond causal analysis (\eg, efficiently learning graphical models on joined relation or selectivity estimation in joined relation, Section~\ref{sec:QO}). 
%\textbf{Join may introduce new CIs:}

\cut{
\begin{example}\label{eg:no-dmap}
Consider relations $R(A, B, C, D, E)$ and $S(D, F)$. % where the P-maps of $R$ and $S$ are both given by complete graphs. 
The relation instance contains the tuples $(a_1, b_1, c, d_1, -)$, $(a_1, b_2, c, d_2, -)$, $(a_2, b_1, c, d_3, -)$, $(a_2, b_2, c, d_4, -)$ (with unique values of $E$ as `-') respectively $2, 3, 1, 1$ times. %Additional tuples with different values of $A, B,C, D, E$ are inserted to destroy all CIs in this relation. 
For $C = c$, $\Pr_R[A = a_1, B = b_1 | C = c] = \frac{2}{7}$, whereas $\Pr_R[A = a_1 | C = c] = \frac{5}{7}$ and $\Pr_R[A = b_1 | C = c] = \frac{3}{7}$, so $\neg (A \indep B |_R C)$. Now suppose $S$ contains 3, 2, 1, 1 tuples respectively of the form $(d_1, -)$, $(d_2, -)$, $(d_3, -)$, $(d_4, -)$, and therefore using the new frequencies in $J = R \Join S$,
$\Pr_J[A = a_1, B = b_1 | C = c] = \frac{6}{14}$, whereas $\Pr_R[A = a_1 | C = c] = \frac{12}{14}$ and $\Pr_R[A = b_1 | C = c] = \frac{7}{14}$, thus satisfying the CI $A \indep B |_{R \Join S} C$.
\end{example} 
}

 One natural choice is that of \emph{embedded multi-valued dependency (EMVD)}  proposed by Fagin \cite{DBLP:journals/tods/Fagin77}.
A \emph{multi-valued dependency (MVD)} $X \mvd Y$ holds in a relation $R$,  if  for each pair of tuples $t_1, t_2 \in R$ such that $t_1[X] = t_2[X]$, there is a tuple $t$ in $R$ such that $t[X] = t_1[X] = t_2[X]$, $t[Y] = t_1[Y]$,   and $t[W] = t_2[W]$, where $W = R \setminus XY$.
An \emph{EMVD} $X \mvd Y | Z$ holds, if in the relation $\pi_Z R$, the MVD $X \mvd Y$ holds.
Basically EMVD gives a weaker form of CIs when the frequencies of variables and tuples are ignored, where $X \mvd Y | Z$ can be represented using independence symbol $\indep*$ as 
$Y \indep* (Z \setminus XY) | X$. As a result, EMVD gives desired properties like if an EMVD $X \mvd Y |_J Z$ holds in the joined relation $J$, and if $Z \subseteq R$, then this EMVD will also hold in the base relation: $X \mvd Y |_R Z$. In addition, if the join does not destroy any tuple from any of the base relations, \ie, if the set of relations $R_1, \cdots, R_k$ is \emph{semi-joined reduced} (all tuples in all relations generate a tuple in the joined relation), then  $X \mvd Y |_R Z \Rightarrow X \mvd Y |_J Z$. Further, EMVDs form a semi-graphoid (Definition~\ref{def:graphoid}) allowing additional inference of CIs in the joined relation.
Nevertheless, CI is fundamentally connected to frequencies, and whether using EMVDs for causal studies is a meaningful option, or if there are other possible ways of defining the probability space on relations, has to be investigated further.  

%\begin{definition}\label{def:mvd}
%\red{Fagin, 1977, Delobel, 1973, Zaniolo, 1976} Given a relation$R$ with attributes $A$, and subsets of attributes $X, Y \subseteq A$, $R$ satisfies the \emph{multi-valued dependency} or \emph{MVD} (over the set of attributes $A$) $X \mvd Y$ if, for each pair of tuples $t_1, t_2 \in R$ such that $t_1[X] = t_2[X]$, there is a tuple $t$ in $R$ such that $t[X] = t_1[X] = t_2[X]$, $t[Y] = t_1[Y]$,   and $t[Z] = t_2[Z]$, where $Z = A \setminus (X \cup Y)$.
%\par
%\red{Fagin, 1977, Delobel, 1978, Rissanen, 1977} Given $X, Y, Z \subseteq A$, the relation $R$ satisfies the \emph{embedded multi-valued dependency} or \emph{EMVD} $X \mvd Y | Z$ if the multi-valued dependency $X \mvd Y$ holds on $\pi_{Z} R$, equivalently, if
%$$\pi_{X \cup Y \cup Z} R = \pi_{X \cup Y} R \Join \pi_{X \cup Z} R$$
%\end{definition}

\subsection{Conclusions}
To summarize, in this paper we proposed a formal framework on causal analysis on multi-relational data extending the Neyman-Rubin potential outcome model. We obtained preliminary results in understanding CIs that hold in the joined relation in general, in special joins, as well as when the base relations are graph-isomorph and the P-maps are given. We also discussed several research problems and directions in the previous sections and in this section. 
Apart from these directions, the other general questions on causal analysis, like covariate selection and making matching techniques more efficient. Finally, extending the framework to more complex queries beyond natural join is a challenging research direction. 
%Initial results on matching using database techniques have been obtained in \cite{FLAME2017, DBLP:journals/corr/SalimiS16}.

%\subsection{Efficiently Learning Causal Graphs}
%\red{Learning bigger causal graphs by learning smaller one -- although with the negative results not sure}
%
%\subsection{Selectivity Estimation for Query Optimization}
%\red{here the general case becomes useful}

%\input{conclusions}

%\section{TODOs}
%\begin{itemize}
%\item Conditional Independence -- instance based
%\item multi-way joins
%\item covariate selection -- condition on as few vars as possible
%\item weak CI using EMVD and approximate CI -- bias introduced
%\end{itemize}

\newpage
%\bibliographystyle{abbrv}
%\bibliography{bib_causality}

\begin{thebibliography}{10}

\bibitem{DBLP:journals/tods/Beeri80}
C.~Beeri.
\newblock On the membership problem for functional and multivalued dependencies
  in relational databases.
\newblock {\em {ACM} Trans. Database Syst.}, 5(3):241--259, 1980.

\bibitem{Beeri+83}
C.~Beeri, R.~Fagin, D.~Maier, and M.~Yannakakis.
\newblock On the desirability of acyclic database schemes.
\newblock {\em J. ACM}, 30(3):479--513, July 1983.

\bibitem{DBLP:journals/jacm/BeeriFMY83}
C.~Beeri, R.~Fagin, D.~Maier, and M.~Yannakakis.
\newblock On the desirability of acyclic database schemes.
\newblock {\em J. {ACM}}, 30(3):479--513, 1983.

\bibitem{bertossi2017abcauses}
L.~Bertossi and B.~Salimi.
\newblock Causes for query answers from databases: Datalog abduction,
  view-updates, and integrity constraints.
\newblock {\em To appear in International Journal of Approximate Reasoning.
  Corr Arxiv Paper cs.DB/1611.01711.}, 2017.

\bibitem{bertossi2017causes}
L.~Bertossi and B.~Salimi.
\newblock From causes for database queries to repairs and model-based diagnosis
  and back.
\newblock {\em Theory of Computing Systems}, 61(1):191--232, 2017.

\bibitem{chapin1947experimental}
F.~Chapin.
\newblock {\em Experimental Designs in Sociological Research}.
\newblock Harper; New York, 1947.

\bibitem{DBLP:conf/sigmod/ChapmanJ09}
A.~Chapman and H.~V. Jagadish.
\newblock Why not?
\newblock In {\em Proceedings of the 2009 ACM SIGMOD International Conference
  on Management of Data}, SIGMOD '09, pages 523--534, 2009.

\bibitem{william1992experimental}
W.~G. Cochran and G.~M. Cox.
\newblock {\em Experimental designs}.
\newblock Wiley Classics Library. Wiley, 1992.

\bibitem{cochran1973controlling}
W.~G. Cochran and D.~B. Rubin.
\newblock Controlling bias in observational studies: A review.
\newblock {\em Sankhy{\=a}: The Indian Journal of Statistics, Series A}, pages
  417--446, 1973.

\bibitem{InfoTheoryBookCoverThomas}
T.~M. Cover and J.~A. Thomas.
\newblock {\em Elements of Information Theory, 2nd Edition}.
\newblock Wiley, 2006.

\bibitem{cox1958planning}
D.~Cox.
\newblock {\em Planning of experiments}.
\newblock Wiley series in probability and mathematical statistics: Applied
  probability and statistics. Wiley, 1958.

\bibitem{DBLP:conf/pods/DalkilicR00}
M.~M. Dalkilic and E.~L. Robertson.
\newblock Information dependencies.
\newblock In {\em Proceedings of the Nineteenth {ACM} {SIGMOD-SIGACT-SIGART}
  Symposium on Principles of Database Systems, May 15-17, 2000, Dallas, Texas,
  {USA}}, pages 245--253, 2000.

\bibitem{de2011covariate}
X.~De~Luna, I.~Waernbaum, and T.~S. Richardson.
\newblock Covariate selection for the nonparametric estimation of an average
  treatment effect.
\newblock {\em Biometrika}, 98(4):861--875, 2011.

\bibitem{deshpande2001independence}
A.~Deshpande, M.~Garofalakis, and R.~Rastogi.
\newblock Independence is good: dependency-based histogram synopses for
  high-dimensional data.
\newblock {\em ACM SIGMOD Record}, 30(2):199--210, 2001.

\bibitem{DBLP:journals/tods/Fagin77}
R.~Fagin.
\newblock Multivalued dependencies and a new normal form for relational
  databases.
\newblock {\em {ACM} Trans. Database Syst.}, 2(3):262--278, 1977.

\bibitem{Fisher1935design}
R.~A. Fisher.
\newblock {\em The design of experiments}.
\newblock Oliver and Boyd, Oxford, England, 1935.

\bibitem{getoor2001selectivity}
L.~Getoor, B.~Taskar, and D.~Koller.
\newblock Selectivity estimation using probabilistic models.
\newblock In {\em ACM SIGMOD Record}, volume~30, pages 461--472. ACM, 2001.

\bibitem{GKT07-semirings}
T.~J. Green, G.~Karvounarakis, and V.~Tannen.
\newblock Provenance semirings.
\newblock In {\em Proceedings of the Twenty-sixth ACM SIGMOD-SIGACT-SIGART
  Symposium on Principles of Database Systems}, PODS '07, pages 31--40, 2007.

\bibitem{greenwood1945experimental}
E.~Greenwood.
\newblock {\em Experimental sociology: A study in method}.
\newblock King's crown Press, 1945.

\bibitem{Holland1986}
P.~W. Holland.
\newblock Statistics and causal inference.
\newblock {\em Journal of the American Statistical Association}, 81(396):pp.
  945--960, 1986.

\bibitem{dblpdata}
http://dblp.uni trier.de/xml/.
\newblock Dblp dataset.

\bibitem{datagov}
https://www.data.gov.
\newblock U.s. government's open data.

\bibitem{yelpdata}
https://www.yelp.com/dataset\_challenge.
\newblock Yelp dataset.

\bibitem{iacus2011causal}
S.~M. Iacus, G.~King, and G.~Porro.
\newblock Causal inference without balance checking: Coarsened exact matching.
\newblock {\em Political analysis}, page mpr013, 2011.

\bibitem{iacus2009cem}
S.~M. Iacus, G.~King, G.~Porro, et~al.
\newblock Cem: software for coarsened exact matching.
\newblock {\em Journal of Statistical Software}, 30(9):1--27, 2009.

\bibitem{KollerF-PGM-book}
D.~Koller and N.~Friedman.
\newblock {\em Probabilistic Graphical Models - Principles and Techniques}.
\newblock {MIT} Press, 2009.

\bibitem{uci-Lichman:2013}
M.~Lichman.
\newblock {UCI} machine learning repository,
  \url{http://archive.ics.uci.edu/ml}, 2013.

\bibitem{DBLP:conf/uai/MaierMAJ13}
M.~E. Maier, K.~Marazopoulou, D.~T. Arbour, and D.~D. Jensen.
\newblock A sound and complete algorithm for learning causal models from
  relational data.
\newblock In {\em Proceedings of the Twenty-Ninth Conference on Uncertainty in
  Artificial Intelligence, {UAI} 2013, Bellevue, WA, USA, August 11-15, 2013},
  2013.

\bibitem{MaierTOJ10}
M.~E. Maier, B.~J. Taylor, H.~Oktay, and D.~Jensen.
\newblock Learning causal models of relational domains.
\newblock In {\em Proceedings of the Twenty-Fourth {AAAI} Conference on
  Artificial Intelligence, (AAAI)}, 2010.

\bibitem{MeliouGMS2011}
A.~Meliou, W.~Gatterbauer, K.~F. Moore, and D.~Suciu.
\newblock The complexity of causality and responsibility for query answers and
  non-answers.
\newblock {\em {Proc. VLDB Endow. (PVLDB)}}, 4(1):34--45, 2010.

\bibitem{neapolitan2004learning}
R.~E. Neapolitan et~al.
\newblock {\em Learning bayesian networks}, volume~38.
\newblock Pearson Prentice Hall Upper Saddle River, NJ, 2004.

\bibitem{neyman1923}
J.~Neyman.
\newblock {\em {On the Application of Probability Theory to Agricul- tural
  Experiments. Essay on Principles. Section 9}}.
\newblock PhD thesis, Roczniki Nauk Rolniczych Tom X [in Polish], 1923.
\newblock translated in Statistical Science, 5, page 465-480.

\bibitem{Pearl-PR-book}
J.~Pearl.
\newblock {\em Probabilistic Reasoning in Intelligent Systems: Networks of
  Plausible Inference}.
\newblock Morgan Kaufmann Publishers Inc., San Francisco, CA, USA, 1988.

\bibitem{PearlBook1998}
J.~Pearl.
\newblock {\em Probabilistic reasoning in intelligent systems: networks of
  plausible inference}.
\newblock Morgan Kaufmann Publishers Inc., San Francisco, CA, USA, 1988.

\bibitem{PearlBook2000}
J.~Pearl.
\newblock {\em Causality: models, reasoning, and inference}.
\newblock Cambridge University Press, 2000.

\bibitem{pearl2014probabilistic}
J.~Pearl.
\newblock {\em Probabilistic reasoning in intelligent systems: networks of
  plausible inference}.
\newblock Morgan Kaufmann, 2014.

\bibitem{PearlPaz86}
J.~Pearl and A.~Paz.
\newblock Graphoids: Graph-based logic for reasoning about relevance relations
  or when would x tell you more about y if you already know z?
\newblock In {\em {ECAI}}, pages 357--363, 1986.

\bibitem{pearl2014confounding}
J.~Pearl and A.~Paz.
\newblock Confounding equivalence in causal inference.
\newblock {\em Journal of Causal Inference J. Causal Infer.}, 2(1):75--93,
  2014.

\bibitem{DBLP:conf/aaai/PearlV87}
J.~Pearl and T.~Verma.
\newblock The logic of representing dependencies by directed graphs.
\newblock In {\em Proceedings of the 6th National Conference on Artificial
  Intelligence. Seattle, WA, July 1987.}, pages 374--379, 1987.

\bibitem{RattiganMJ11}
M.~J. Rattigan, M.~E. Maier, and D.~Jensen.
\newblock Relational blocking for causal discovery.
\newblock In {\em Proceedings of the Twenty-Fifth {AAAI} Conference on
  Artificial Intelligence, ({AAAI})}, 2011.

\bibitem{Rosenbaum2005}
P.~R. Rosenbaum.
\newblock {\em Observational study}, pages 1451--1462.
\newblock Wiley, Hoboken, N.J., 2005.

\bibitem{RosenbaumRubin1983}
P.~R. Rosenbaum and D.~B. Rubin.
\newblock The central role of the propensity score in observational studies for
  causal effects.
\newblock {\em Biometrika}, 70(1):pp. 41--55, 1983.

\bibitem{rosenbaum1984reducing}
P.~R. Rosenbaum and D.~B. Rubin.
\newblock Reducing bias in observational studies using subclassification on the
  propensity score.
\newblock {\em Journal of the American statistical Association},
  79(387):516--524, 1984.

\bibitem{rosenbaum1985constructing}
P.~R. Rosenbaum and D.~B. Rubin.
\newblock Constructing a control group using multivariate matched sampling
  methods that incorporate the propensity score.
\newblock {\em The American Statistician}, 39(1):33--38, 1985.

\bibitem{FLAME2017}
S.~{Roy}, C.~{Rudin}, A.~{Volfovsky}, and T.~{Wang}.
\newblock {FLAME: A Fast Large-scale Almost Matching Exactly Approach to Causal
  Inference}.
\newblock {\em ArXiv e-prints}, 1707.06315, July 2017.

\bibitem{RoyS14}
S.~Roy and D.~Suciu.
\newblock A formal approach to finding explanations for database queries.
\newblock In {\em Proceedings of the 2014 ACM SIGMOD International Conference
  on Management of Data}, SIGMOD '14, pages 1579--1590, 2014.

\bibitem{rubin2006matched}
D.~Rubin.
\newblock {\em Matched Sampling for Causal Effects}.
\newblock Cambridge University Press, 2006.

\bibitem{rubin1973matching}
D.~B. Rubin.
\newblock Matching to remove bias in observational studies.
\newblock {\em Biometrics}, pages 159--183, 1973.

\bibitem{Rubin1974}
D.~B. Rubin.
\newblock Estimating causal effects of treatments in randomized and
  nonrandomized studies.
\newblock {\em {Journal of Educational Psychology}}, 66(5):688, 1974.

\bibitem{rubin1976multivariate}
D.~B. Rubin.
\newblock Multivariate matching methods that are equal percent bias reducing,
  i: Some examples.
\newblock {\em Biometrics}, pages 109--120, 1976.

\bibitem{Rubin2005}
D.~B. Rubin.
\newblock Causal inference using potential outcomes.
\newblock {\em Journal of the American Statistical Association},
  100(469):322--331, 2005.

\bibitem{SalimiTaPP16}
B.~Salimi, L.~Bertossi, D.~Suciu, and G.~{Van den Broeck}.
\newblock Quantifying causal effects on query answering in databases.
\newblock In {\em TaPP}, 2016.

\bibitem{salimi2017zaliql}
B.~Salimi, C.~Cole, D.~R. Ports, and D.~Suciu.
\newblock Zaliql: Causal inference from observational data at scale.
\newblock {\em Proceedings of the VLDB Endowment}, 10(12), 2017.

\bibitem{DBLP:conf/uai/ShenoyS88}
P.~P. Shenoy and G.~Shafer.
\newblock Axioms for probability and belief-function proagation.
\newblock In {\em {UAI} '88: Proceedings of the Fourth Annual Conference on
  Uncertainty in Artificial Intelligence, Minneapolis, MN, USA, July 10-12,
  1988}, pages 169--198, 1988.

\bibitem{Silverstein+2000}
C.~Silverstein, S.~Brin, R.~Motwani, and J.~Ullman.
\newblock Scalable techniques for mining causal structures.
\newblock {\em Data Min. Knowl. Discov.}, 4(2-3):163--192, July 2000.

\bibitem{stuart2010matching}
E.~A. Stuart.
\newblock Matching methods for causal inference: A review and a look forward.
\newblock {\em Statistical science: a review journal of the Institute of
  Mathematical Statistics}, 25(1):1, 2010.

\bibitem{Tzoumas2013}
K.~Tzoumas, A.~Deshpande, and C.~S. Jensen.
\newblock Efficiently adapting graphical models for selectivity estimation.
\newblock {\em The VLDB Journal}, 22(1):3--27, Feb 2013.

\end{thebibliography}

\appendix
\section{Proofs from Section~4}
\subsection{Proof of Lemma~\ref{lem:cl-join-0}}

\begin{proof}[of  of Lemma~\ref{lem:cl-join-0}]
Fix any value $z$ (as a set) of the attributes $Z = R \cap S$. Given $Z = z$, all tuples $r \in R$ with $r[Z] = z$ join with all tuples $s \in S$ with $s[Z] = z$, and with no other tuples. 
Let $X = R \setminus S$, $Y = S \setminus R$, $T = R \Join S$. For any values $X = x, Y = y$, $N_{T, xz} = N_{R, xz} \times N_{S, z}$, $N_{T, yz} = N_{S, yz} \times N_{R, z}$,
$N_{T, xyz} = N_{R, xz} \times N_{S, yz}$, and $N_{T, z} = N_{R, z} \times N_{S, z}$. Hence it follows that 
$$\frac{N_{T, xyz}}{N_{T, z}} = \frac{N_{T, xz}}{N_{T, z}} \times \frac{N_{T, yz}}{N_{T, z}}$$
\ie, $X \indep Y |_{T} Z$.
\end{proof}

\subsection{Proof of Lemma~\ref{lem:ci-join-1}}
\begin{proof}[of Lemma~\ref{lem:ci-join-1}]
By the definition of conditional independence, and since $X \indep Y |_{R \Join S} Z$, for all values $x, y, z$ of $X, Y, Z$ we have
\begin{equation}\label{equn:6}
\frac{N_{R, xyz}}{N_{R, z}} = \frac{N_{R, xz}}{N_{R, z}} \times \frac{N_{R, yz}}{N_{R, z}}
\end{equation}
Fix arbitrary $x, y, z$. 
Let $Z = Z_1 \cup Z_2$, where $Z_1 = R \cap S$ and $Z_2  = Z \setminus Z_1$. 
Let $Z_1 = z_1$ and $Z_2 = z_2$ in $z$.
Note that $N_{S, (Z_1, z_1)}$ is the number of tuples in $S$ with the value of the join attributes as $z_1$. Each tuple $r$ in $R$ with $Z = z$, \ie, with $Z_1 = z_1$, joins with all tuples in $S$ with $Z_1 = z_1$, and joins with no other tuples in  $S$. Hence multiplying the numerator and denominator all terms in equation (\ref{equn:6}) above, we get
$$\frac{N_{R \Join S, xyz}}{N_{R \Join S, z}} = \frac{N_{R \Join S, xz}}{N_{R \Join S, z}} \times \frac{N_{R \Join S, yz}}{N_{R \Join S, z}}$$
In other words, $X \indep Y |_{R \Join S} Z$.
\end{proof}

\subsection{Proof of Lemma~\ref{lem:ci-join-2}}
\begin{proof}[of Lemma~\ref{lem:ci-join-2}]
Since $X \supseteq (R	 \cap S)$,  suppose $X = W \cup V$ where $V = (R \cap S)$ and $W = Z \setminus W$, \ie, $V$ denotes the set of joined attributes. 
By the definition of conditional independence, and since $WV \indep Y |_{R \Join S} Z$, for all values $y, w, v, z$ of $Y, W, V, Z$, from (\ref{equn:6}) we have: 
\begin{equation}\label{equn:7}
N_{R, wvz} \times N_{R, yz} = N_{R, wvyz} \times N_{R, z}
\end{equation}
In the joined relation $T = R \Join S$, every tuple with a value $V = v$ in $R$ will join with all $N_{S, v}$ tuples in $S$ and with no other tuples. Hence,  
{\small
\begin{eqnarray}
\nonumber
&&N_{T, wvyz} \times N_{T, z}\\\nonumber
& = & N_{T, wvyz} \times (\sum_{w', v', y'} N_{T, w'v''y'z})\\\nonumber
& = & (N_{R, wvyz} \times N_{S, v}) \times \sum_{w', v', y'} (N_{R, w' v' y' z} \times N_{S, v'})\\\nonumber
& = & (N_{R, wvz} \times N_{R, yz}) \times N_{S, z} \times \sum_{w', v', y'} (N_{R, w'v'z} \times N_{R, y'z} \times N_{S, v'})\\\nonumber
%%%phantom ignores
&&\phantom{(N_{R, wvz} \times N_{R, yz}) \times N_{S, z} \times }~~~\textrm{(by (\ref{equn:7}))}\\\nonumber
& = & N_{R, wvz} \times N_{R, yz} \times N_{S, z} \times \sum_{w', v'} (N_{R, w'v'z} \times N_{S, v'} \times (\sum_{y'}N_{R, y'z}) )\\\nonumber
& = & N_{R, wvz} \times N_{R, yz} \times N_{S, z} \times \sum_{w', v'} (N_{R, w'v'z} \times N_{S, v'} \times N_{R, z}) \\
& = & N_{R, wvz} \times N_{R, yz} \times N_{S, z} \times N_{R, z} \times \sum_{w', v'} (N_{R, w'v'z} \times N_{S, v'})\label{equn:8}
\end{eqnarray}
}
And,

{\small
\begin{eqnarray}
\nonumber
&&N_{T, wvz} \times N_{T, yz}\\\nonumber
& = & (\sum_{y'} N_{T, wvy'z}) \times (\sum_{w', v'} N_{T, w'v''yz})\\\nonumber
& = & \sum_{y'} N_{R, wvy'z} \times N_{S, z}) \times \sum_{w', v'} (N_{R, w' v' y z} \times N_{S, v'})\\\nonumber
& = & N_{S, z} \times (\sum_{y'}N_{R, wvz} \times N_{R, y'z})  \times \sum_{w', v'} (N_{R, w'v'z} \times N_{R, yz} \times N_{S, v'})\\\nonumber
%%%phantom ignores
&&\phantom{(N_{R, wvz} \times N_{R, yz}) \times N_{S, z} \times }~~~\textrm{(by (\ref{equn:7}))}\\\nonumber
& = & N_{R, wvz} \times N_{R, yz} \times N_{S, z} \times (\sum_{y'} N_{R, y'z})  \times  \sum_{w', v'} (N_{R, w'v'z} \times N_{S, v'})\\\nonumber
& = & N_{R, wvz} \times N_{R, yz} \times N_{S, z} \times N_{R, z}\sum_{w', v'} (N_{R, w'v'z} \times N_{S, v'} \times N_{R, z}) \\
& = & N_{R, wvz} \times N_{R, yz} \times N_{S, z} \times N_{R, z} \times \sum_{w', v'} (N_{R, w'v'z} \times N_{S, v'})\label{equn:9}
\end{eqnarray}
}
From (\ref{equn:8}) and (\ref{equn:9}), for all $x = (w, v), y, z$
\begin{eqnarray}
\nonumber
&& N_{T, wvyz} \times N_{T, z} = N_{T, wvz} \times N_{T, yz}\\\nonumber
& \Rightarrow & \frac{N_{T, xyz}}{N_{T, z}} = \frac{N_{T, xz}}{N_{T, z}} \times \frac{N_{T, yz}}{N_{T, z}}
\end{eqnarray}
\ie, in $T = R \Join S$, $X \indep Y |_{R \Join S} Z$.
\end{proof}

\subsection{Proof of Proposition~\ref{prop:join-fk}}

\begin{proof}[of Proposition~\ref{prop:join-fk}]
With a primary key join, every tuple in $R$ joins with exactly one tuple in $S$. Hence for $T = R \Join S$, and for all $x, y, z$ values of $X, Y, Z$, if 
$\frac{N_{R, xyz}}{N_{R, z}} = \frac{N_{R, xz}}{N_{R, z}} \times \frac{N_{R, yz}}{N_{R, z}}$, then 
$\frac{N_{T, xyz}}{N_{T, z}} = \frac{N_{T, xz}}{N_{T, z}} \times \frac{N_{T, yz}}{N_{T, z}}$, since all frequencies in all numerators and denominators in $R$ is multiplied by 1 to obtain the frequencies in $T$. 
\end{proof}

\subsection{Proof of Proposition~\ref{prop:app-causality-all}}

\begin{proof}[of Proposition~\ref{prop:app-causality-all}]
Let% $U = (X \cap R) \setminus (\{T\} \cup (R \cap S))$, and $V = (X \cap S) \setminus (\{Y\} \cup (R \cap S))$, \ie, 
$U, V$ denote the subset of attributes in $R$ and $S$ respectively in $X$ that do not belong to $T, Y, $ or $R \cap S$. 
Below, we use $Y'$ as a placeholder for either $Y(0), Y(1)$ (to prove (i)), or $Y$ (to prove (ii)).
%Then, 
\begin{eqnarray*}
&&UT \indep VY' |_{R\Join S} (R \cap S)~~~~\textrm{(Corollary~\ref{cor:cl-join-0})}\\
& \Rightarrow & UT \indep Y' |_{R\Join S} (R \cap S) V~~~~\textrm{(weak union (\ref{equn:wu})}\\
& \Rightarrow & T \indep Y' |_{R\Join S} (R \cap S) VU~~~~\textrm{(weak union (\ref{equn:wu})}\\
& \equiv & T \indep Y' |_{R\Join S} X\\
%&  \equiv & T \indep Y(0), Y(1)  |_{R\Join S} X\\
\end{eqnarray*}
since $X = \{U\} \cup \{V\} \cup (R \cap S)$.
This shows that (i) $T \indep Y(0), Y(1)  |_{R\Join S} X$, and (ii) $T \indep Y  |_{R\Join S} X$.
\par
To show (iii), we show that if  $T \indep Y  |_{R\Join S} X$, then ATE = 0. %This holds because from (\ref{equn:ae})
If $T \indep Y |X$ then for all $X = x, Y = y, T = t$, $P(y|t, x)=P(y|x)$. Thus,  $E[Y|X]=E[Y|X,T=1]=E[Y|X,T=0]$. Therefore (see (\ref{equn:ae})),
$ATE$ = $E_X[E[Y(1)|T=1,X]] - E_X[E[Y(0)|T=0,X]]$ = $E_X[E[Y|T=1,X]] - E_X[E[Y|T=0,X]]$ = 0. This shows (iii).
\end{proof}

\subsection{Proof of Proposition~\ref{prop:cequi}}

	We will use the following properties of entropy \cite{InfoTheoryBookCoverThomas, DBLP:conf/pods/DalkilicR00}:

%The following useful properties will be used through this paper:
\begin{proposition} \label{obs:entropy_prop} %The followings hold
 For $X, Y, Z \subseteq \attr$
		\begin{itemize}
		\itemsep0em
		\item[(a)] $H(X) \geq 0$, $H(X|Y) \geq 0$, $I(X, Y) \geq 0$
	and  $I(X, Y|Z) \geq 0$
		\item[(b)]  $X \indep Y$ if and only if $I(X,Y)=0$, and $X \indep Y |Z$  if and only if $I(X,Y|Z)=0$.
	\item[(c)] If $X$ functionally determines $Y$, \ie, if $X \rightarrow Y$, then $H(Y|X)=0$.
	\item[(d)] If $H(Y|X)=0$ then for any $Z$, $H(Y|XZ)=0$
	\item[(e)] If $X \rightarrow Y$ (or, if $H(Y|X) = H(Y|XZ) = 0$) then for any $Z$, $I(Y,Z|X)$ = 0. 
		\end{itemize}

\end{proposition}

\begin{proof}
 Proofs of (a, b) can be found in \cite{InfoTheoryBookCoverThomas} (some observations are obvious). To see (c) (also shown in \cite{DBLP:conf/pods/DalkilicR00}), note that if $X$ functionally determines $Y$
	then for any $x \in X$ and $y \in Y$ $P(Y=y,X=x)=P(X=x)$ (follows from the definition of a functional dependency).
	Thus $H(X,Y)=H(X)$ and therefore, $H(Y|X)=H(X,Y)-H(X)=0$. (d) is obtained from the 
	non-negativity of mutual information. Since $I(Y, Z|X) \geq 0$ then $H(Y|X)-H(Y|XZ) \geq 0$, i.e., 
	$H(Y|X) \geq H(Y|XZ)$. Now if
	$H(Y|X)=0$ then $H(Y|XZ)\leq 0$. Since $H(Y|XZ) \geq 0$, $H(Y|XZ) = 0$. (e) follows from (c) and (d), since $I(Y,Z|X) = H(Y|X) - H(Y|XZ)$. 
\end{proof}	

Now we prove Proposition~\ref{prop:cequi}.
\begin{proof}[of Proposition~\ref{prop:cequi}]
To prove the claim, we show that one of the conditions in Theorem \ref{th:ceq} 
is satisfied by $X$  and $X_i$. In particular, we show that both 
\begin{equation}\label{equn:equiv-1}
T \indep X_i|_U  X
\end{equation} and 
\begin{equation}\label{equn:equiv-2}
Y \indep X|_U X_iT
\end{equation} 
 hold. 

First we prove (\ref{equn:equiv-1}). Since $X_i \subseteq  X$, $X$ functionally
determines $X_i$ (this is a trivial functional dependency in $U$). 
%Now from Proposition~\ref{obs:entropy_prop}(c), it follows that $H(X_i|X)=0$. 
Now, from  Proposition~\ref{obs:entropy_prop}(e) it follows that, $I(X_i, T|X)$.  Therefore, $T \indep X_i|_U  X$ (Proposition~ \ref{obs:entropy_prop}(c)), \ie,   (\ref{equn:equiv-1}) holds. 

Next we show (\ref{equn:equiv-2}). 
Since $H(A|B) = H(AB) - H(B)$, 
\begin{equation}\label{equn:equiv-3}
H(X \setminus X_i | X_i) = H(X) - H(X_i) = H(X|X_i)
\end{equation}
Let $FK$ denote the set of foreign keys from $R_i$ to all $R_j$, $j \in [1, k]$, $j \neq i$.
%=\{FK_1, \ldots, FK_n\} \setminus FK_i$. Note that $FK_{-i}$ 
Hence $FK$ functionally determines $X-X_i$ in $U$, \ie,
%(**) implied in the following steps: 

\begin{eqnarray*} \tiny
&& H(X \setminus X_i|FK)=0 ~~~~~\textrm{(Proposition \ref{obs:entropy_prop}(c))} \nonumber\\  % \label{equn:given} \nonumber\\ 
& \Rightarrow & H(X \setminus X_i|FK X_i)=0 ~~~ \textrm{(Proposition \ref{obs:entropy_prop}(d))} \nonumber \\  
& \Rightarrow & H(X \setminus X_i|X_i)=0  ~~~~ \textrm{(Since } {FK \subseteq X_i)} \nonumber \\  
& \Rightarrow & H(X|X_i)=0   ~~~~ \textrm{{(From } (\ref{equn:equiv-3}))} \nonumber \\  
& \Rightarrow & H(X|X_i,T)=0  ~~~~  \textrm{(Proposition \ref{obs:entropy_prop}(d))} \nonumber \\  
& \Rightarrow & I(X,Y|X_iT)=0 ~~~~ \textrm{(Proposition \ref{obs:entropy_prop}(e))} \nonumber \\  
& \Rightarrow & X\indep Y| X_i, T   ~~~~~ \textrm{(Proposition \ref{obs:entropy_prop}(b))}\nonumber   
\end{eqnarray*}
\end{proof}

\cut{
Since $X_i \rightarrow U$ and $X \subseteq U$, $X_i \rightarrow X$. Hence,
%$X_i T \rightarrow $
%$I(X,Y|X_iT)=0$

\begin{eqnarray*} \tiny
%&& H(X|R_i)=0 ~~~~~\textrm{(Proposition \ref{obs:entropy_prop}(c))} \label{equn:given} \nonumber\\ 
%& \Rightarrow & H(X \setminus X_i|FK_{-i}X_i)=0 ~~~ \textrm{(Proposition \ref{obs:entropy_prop}(d))} \nonumber \\  
%& \Rightarrow & H(X \setminus X_i|X_i)=0  ~~~~ \textrm{(Since } {FK_{-i} \subset X_i)} \nonumber \\  
&& H(X|X_i)=0   ~~~~  \textrm{(Proposition \ref{obs:entropy_prop}(c))} \nonumber \\  %\textrm{{(From } (\ref{equn:equiv-3}))} \nonumber \\  
& \Rightarrow & H(X|X_i,T)=0  ~~~~  \textrm{(Proposition \ref{obs:entropy_prop}(d))} \nonumber \\  
& \Rightarrow & I(X,Y|X_iT)=0 ~~~~ \textrm{(Proposition \ref{obs:entropy_prop}(e))} \nonumber \\  
& \Rightarrow & X\indep Y| X_i, T   ~~~~~ \textrm{(Proposition \ref{obs:entropy_prop}(b))}\nonumber   
\end{eqnarray*}
}

\section{Proofs from Section~5}

\subsection{Proof of Lemma~\ref{lem:either-or}}

\begin{proof}[of Lemma~\ref{lem:either-or}]
Suppose not, \ie, assume the contradiction that $\neg (X \indep YD |_R Z)$ \emph{and} $\neg(XD \indep Y |_R Z)$. 
Note that $D$ is a single vertex in $G_1$. Since $G_1$ is a P-map for $R$, it follows that $\neg (X \indep YD |_{G_1} Z)$ \emph{and} $\neg(XD \indep Y |_{G_1} Z)$. Since $\neg (X \indep YD |_{G_1} Z)$, in $G_1$, there is a  path from $X$ to $DY$ that does not use any vertex in $Z$. Since $X \indep Y|_R Z \equiv X \indep Y|_{G_1} Z$, removing $Z$ disconnects $X$ from $Y$, hence there must be a path $p_1$ from $X$ to $D$ that does not use vertices from $Z$. Similarly, using $\neg(XD \indep Y |_{G_1} Z)$, there is a path $p_2$ from $D$ to $Y$ that does not use vertices from $Z$. Combining $p_1$ and $p_2$ (and making it a simple path by removing vertices if needed), there is path from $X$ to $Y$ in $G_1$ that does not use vertices from $Z$, contradicting the given assumption that $X \indep Y|_{G_1} Z$, and equivalently $X \indep Y|_{R} Z$. Hence either $X \indep YD |_R Z$ or $XD \indep Y |_R Z$. 
\par
An alternative proof can be obtained using the transitivity property (\ref{equn:tran-1}). Since $X \indep Y |_{G_1} Z$, by strong union (\ref{equn:su-1})
\begin{equation}
X \indep Y |_{G_1} Z \label{equn:bsu}
\end{equation}
Also,
{\small
\begin{eqnarray} 
	&& X \indep Y |_{G_1} Z \nonumber\\ %\label{equn:given}\\ 
	& \Rightarrow & X \indep D |_{G_1} Z~~  \textrm{or}~~ Y \indep D |_{G_1} Z  ~~~\textrm{(Transitivity (\ref{equn:tran-1}))} \nonumber\\
	& \Rightarrow & X \indep D |_{G_1} ZY~~ \textrm{or}~~ Y \indep D |_{G_1} ZX \label{equn:btra} \\
	&& ~~~~~~~~~\textrm{~(Strong union (\ref{equn:su-1}))} \nonumber\\
	%& \Rightarrow &X \indep Y |_{G_1} ZD~~~\textrm{((\ref{equn:given}) and Strong union (\ref{equn:su-1}))} \label{equn:bsu}\\
		& \Rightarrow & X \indep YD |_{G_1} Z~~ \textrm{or}~~  XD \indep Y |_{G_1} Z \nonumber\\
		& &~~~~~~~~~\textrm{((\ref{equn:bsu}),(\ref{equn:btra}) and Intersection (\ref{equn:int-1}))} \nonumber
\end{eqnarray}
}
\end{proof}

\cut{
\babak{ The following proof work for an arbitrary singleton D disjoint from X,Y and Z. I don't know why the transitivity axiom restricted to singleton $\gamma$ in Pearl's book.
}
}

\subsection{Proof of Theorem~\ref{thm:pmap-ci-propagates}}

Now we prove Theorem~\ref{thm:pmap-ci-propagates} using Lemma~\ref{lem:either-or}.
\begin{proof}[of Theorem~\ref{thm:pmap-ci-propagates}]
If joined attributes $D \subseteq X, Y, $ or $Z$, then by Theorem~\ref{thm:ci-join}, $X \indep Y |_{R \Join S} Z$.
\par
Otherwise, assume $D \not\subseteq X, Y$, and $Z$. By  Lemma~\ref{lem:either-or}, $X \indep YD |_R Z$ or $XD \indep Y |_R Z$. Without loss of generality, suppose $X \indep YD |_R Z$. Then by Theorem~\ref{thm:ci-join}, $X \indep YD|_{R \Join S} Z$. Then by the decomposition property of graphoid axioms (\ref{equn:dec}), $X \indep Y|_{R \Join S}Z$.  
\end{proof}

\subsection{Example~\ref{eg:pmap-prop}: CIs from Graph-Isomorph Relations Propagate to Joined Relation}

\begin{example}\label{eg:pmap-prop} 
%\red{move to appendix}.
Suppose $R = (ABCD)$ is given by graph $G_1 = A - B - C - D$, and $S = (DE)$ is given by $D-E$. % \red{draw graphs}. 
Here we give an instance of $R$ that conforms to $G_1$: \ie, $A \indep CD |_R B$, $B \indep D |_R C$, $A \indep D |_R BC$ -- all can be verified from $R$; $S$ remains the same:\\

{\small
\begin{tabular}{|c|c|c|c|}
\multicolumn{4}{c}{$\mathbf{R}$}\\
\hline
$A$ & $B$ & $C$ & $D$\\
\hline
$a_1$ & $b_1$ & $c$ & $d_1$ \\
$a_1$ & $b_2$ & $c$ & $d_2$ \\
$a_2$ & $b_1$ & $c$ & $d_3$ \\
$a_2$ & $b_2$ & $c$ & $d_4$ \\
$a_1$ & $b_1$ & $c$ & $d_3$ \\
$a_1$ & $b_2$ & $c$ & $d_4$ \\
$a_2$ & $b_1$ & $c$ & $d_1$ \\
$a_2$ & $b_2$ & $c$ & $d_2$ \\\hline
\end{tabular}
\begin{tabular}{|c|c|}
\multicolumn{2}{c}{$\mathbf{S}$}\\
\hline
$D$ & $E$\\
\hline
$d_1$ & $e_1$ \\
$d_1$ & $e_2$ \\
$d_2$ & $e_1$ \\
$d_2$ & $e_2$ \\
$d_2$ & $e_3$ \\
$d_3$ & $e_1$ \\
$d_4$ & $e_1$ \\\hline
\end{tabular}
\begin{tabular}{|c|c|c|c|c|}
\multicolumn{5}{c}{$\mathbf{R} \Join \mathbf{S}$}\\
\hline
$A$ & $B$ & $C$ & $D$ & $E$\\
\hline
$a_1$ & $b_1$ & $c$ & $d_1$ & $e_1$ \\
$a_1$ & $b_1$ & $c$ & $d_1$ & $e_2$ \\
$a_1$ & $b_2$ & $c$ & $d_2$ & $e_1$  \\
$a_1$ & $b_2$ & $c$ & $d_2$ & $e_2$  \\
$a_1$ & $b_2$ & $c$ & $d_2$ & $e_3$  \\
$a_2$ & $b_1$ & $c$ & $d_3$ & $e_1$  \\
$a_2$ & $b_2$ & $c$ & $d_4$ & $e_1$  \\
$a_2$ & $b_1$ & $c$ & $d_1$ & $e_1$ \\
$a_2$ & $b_1$ & $c$ & $d_1$ & $e_2$ \\
$a_2$ & $b_2$ & $c$ & $d_2$ & $e_1$  \\
$a_2$ & $b_2$ & $c$ & $d_2$ & $e_2$  \\
$a_2$ & $b_2$ & $c$ & $d_2$ & $e_3$  \\
$a_1$ & $b_1$ & $c$ & $d_3$ & $e_1$  \\
$a_1$ & $b_2$ & $c$ & $d_4$ & $e_1$  \\
\hline
\end{tabular}\\
}

\smallskip
Also, note that, $A \indep C |_{R \Join S} B$: in $R \Join S$, $\Pr[A = a_1, B = b_1 | C = c] = \frac{3}{14}$, whereas $\Pr[A = a_1| C = c] = \frac{7}{14}$ and
$\Pr[B = b_1 | C = c] = \frac{6}{14}$, \ie, the CI now propagates to $R \Join S$.
\end{example}

\subsection{Proof of Lemma~\ref{lem:two-one-prop}}

\begin{proof}[of Lemma~\ref{lem:two-one-prop}] Without loss of generality, assume $X, Z \subseteq R$ and $Y \subseteq S$. Let $D = R \cap S$ denote the singleton join attribute. 
There are different cases:
\par
\textbf{(i) $D \in Z$, $D \notin X, Y$:} If $Z = D$, the lemma follows from Theorem~\ref{thm:ci-join}. Hence assume $Z = DZ_1$, where $Z_1 \subseteq R$. By  Theorem~\ref{thm:ci-join}, $XZ_1 \indep Y |_{R \Join S} D$. By weak union property of graphoid axioms, $X \indep Y |_{R \Join S} DZ_1$, or $X \indep Y |_{R \Join S} Z$.
\par
\textbf{(ii) $D \in X$ (similarly $Y$), $D \notin Y, Z$} If $D = X$, by Theorem~\ref{thm:ci-join} the lemma follows. Hence assume $X = DX_1$, where $X_1 \subseteq R$ and $DX_1 \indep Y |_{G} Z$. We claim that this case cannot arise. Suppose not. Then in $G$, no path exists between $D$ and $Y$ in $G_{-Z}$. However, $Z \in R$ whereas $D, Y \in S$, $G_2$ is connected by assumption, and the connectivity of $D$ and $Y$ is not affected by removing $Z$. 
\par
\textbf{(iii) $D \not\in X, Y, Z$:} %(otherwise by Theorem~\ref{thm:ci-join} the lemma is proved. 
Since $X \indep Y |_{G} Z$, all paths between $X$ and $Y$ in $G$ go through $Z$. By Observation~\ref{obs:join-attr-cutset}, all paths between $X$ and $Y$ also go through $D$. Therefore, 
\begin{equation}\label{equn:xd-i-z-g}
X \indep D |_{G} Z
\end{equation}
 since otherwise, there is a path from $X$ to $D$ that do not go through $Z$, and in conjunction with a path between $D$ to $Y$ in $G_2$ (we assumed that both $G_1, G_2$ are connected), we get a path between $X$ and $Y$ in $G$ that does not go through $Z$ violating the assumption that $X \indep Y |_{G} Z$. 
\par
From (\ref{equn:xd-i-z-g}), since $G_1$ has a subset of edges of $G$ (if a path exists in $G_1$ it must exist in $G$), we have
\begin{eqnarray}
\nonumber
& & X \indep D |_{G_1} Z\\\nonumber
&\Rightarrow  &X \indep D |_{R} Z ~~~~~~~\textrm{(since $G_1$ is a P-map of $R$)}\\
& \Rightarrow & X \indep D |_{R\Join S} Z~~~~~~~\textrm{(from Theorem~\ref{thm:pmap-ci-propagates})}\label{equn:x-d-z-rs}
\end{eqnarray}
The last step follows from the fact that all of $X, Y, Z$ belong to $R$.
Now, from Corollary~\ref{cor:cl-join-0} we have 
\begin{eqnarray}
\nonumber
& & XZ \indep Y |_{R \Join S} D\\\nonumber
&\Rightarrow  &X \indep Y |_{R\Join S} D Z \label{equn:x-y-dz-rs}
\end{eqnarray}
where the last step follows from weak union of graphoid axioms (\ref{equn:wu}).
From (\ref{equn:x-d-z-rs}) and (\ref{equn:x-y-dz-rs}), applying the contraction property of the graphoid axioms (\ref{equn:con}) (assume $\X = X, \Y = D, \Z = Z, \W = Y$), 
\begin{eqnarray*}
&&X \indep YD |_{R \Join S} Z\\
& \Rightarrow & X \indep Y |_{R \Join S} Z
\end{eqnarray*}
by the decomposition property of the graphoid axioms (\ref{equn:dec}).
This proves the lemma.
\end{proof}

\subsection{Proof of Lemma~\ref{lem:y2-z2-empty}}

\begin{proof}[of Lemma~\ref{lem:y2-z2-empty}]
We consider all possible cases w.r.t. the join attribute $D$.
\par
\textbf{(i) $D \in Z_1, D \notin X_1, X_2, Y_1$}: 
(i-a) Suppose $Z_1 = D$.  Since $X_1X_2 \indep Y_1 |_{G} Z_1$, we have $X_1 \indep Y_1 |_{G} Z_1$, by Observation~\ref{obs:subset-indep},  $X_1 \indep Y_1 |_{G_1} Z_1$, since $G_1$ is a P-map of $R$, $X_1 \indep Y_1 |_{R} Z_1$, by Theorem~\ref{thm:ci-join}, 
\begin{equation}\label{equn:x1-y1-rs-z-1}
X_1 \indep Y_1 |_{R \Join S} Z_1 \equiv X_1 \indep Y_1 |_{R \Join S} D
\end{equation}
Further, by Corollary~\ref{cor:cl-join-0},
\begin{equation}\label{equn:y1-x2-rs-d-1}
Y_1 \indep X_2 |_{R \Join S} D
\end{equation}
By Corollary~\ref{cor:cl-join-0},
$X_1 Y_1 \indep X_2 |_{R \Join S} D$, by weak union (\ref{equn:wu}), $X_1 \indep X_2 |_{R \Join S} DY_1$. Using (\ref{equn:x1-y1-rs-z-1}) and contraction (\ref{equn:con}),  $X_1 \indep Y_1X_2 |_{R \Join S} D$. By weak union (\ref{equn:wu}), $X_1 \indep Y_1 |_{R \Join S} DX_2$. By contraction (\ref{equn:con}) and (\ref{equn:y1-x2-rs-d-1}), $X_1X_2 \indep Y_1 |_{R \Join S} D = X_1X_2 \indep Y_1 |_{R \Join S} Z_1$.
\par
(i-b) Otherwise, suppose $D \neq Z_1$, and $Z_1 = DW_1$, where $W_1 \subseteq R$. 
Since $X_1X_2 \indep Y_1 |_{G} DW_1$, we have $X_1 \indep Y_1 |_{G} DW_1$, by Observation~\ref{obs:subset-indep},  $X_1 \indep Y_1 |_{G_1} DW_1$, since $G_1$ is a P-map of $R$, $X_1 \indep Y_1 |_{R} DW_1$, by Theorem~\ref{thm:ci-join}, 
\begin{equation}\label{equn:x1-y1-rs-z-2}
X_1 \indep Y_1 |_{R \Join S} DW_1
\end{equation}
Further, by Corollary~\ref{cor:cl-join-0}, $Y_1W_1 \indep X_2 |_{R \Join S} D$, and by weak union (\ref{equn:wu}), 
\begin{equation}\label{equn:y1-x2-rs-d-2}
Y_1 \indep X_2 |_{R \Join S} DW_1
\end{equation}
By Corollary~\ref{cor:cl-join-0},
$X_1 Y_1 W_1 \indep X_2 |_{R \Join S} D$, by weak union (\ref{equn:wu}), $X_1 \indep X_2 |_{R \Join S} (DW_1)Y_1$. Using (\ref{equn:x1-y1-rs-z-2}) and contraction (\ref{equn:con}),  $X_1 \indep Y_1X_2 |_{R \Join S} DW_1$. By weak union (\ref{equn:wu}), $X_1 \indep Y_1 |_{R \Join S} (DW_1)X_2$. By contraction (\ref{equn:con}) and (\ref{equn:y1-x2-rs-d-2}), $X_1X_2 \indep Y_1 |_{R \Join S} DW_1 = X_1X_2 \indep Y_1 |_{R \Join S} Z_1$.
\par
\textbf{(ii) (ii-a) $D \in X_1, D \notin X_2, Y_1, Z_1$}: If $D = X_1$, \ie, if $DX_2 \indep Y_1 |_{G} Z_1$, then $D \indep Y_1|_{G_1} Z_1$ (Observation~\ref{obs:subset-indep}), since $G_1$ is a P-map $D \indep Y_1|_{R} Z_1$ and by Theorem~\ref{thm:pmap-ci-propagates}, 
\begin{equation}\label{equn:y1-d-z1-1}
Y_1 \indep D|_{R \Join S} Z_1
\end{equation}
By Corollary~\ref{cor:cl-join-0}, $Y_1Z_1 \indep X_2 |_{R \Join S} D$, and by weak union (\ref{equn:wu}),
\begin{equation}\label{equn:y1-x2-dz1-1}
Y_1 \indep X_2|_{R \Join S} DZ_1
\end{equation}
Combining (\ref{equn:y1-d-z1-1}) and (\ref{equn:y1-x2-dz1-1}) by contraction (\ref{equn:con}), we have $Y_1 \indep DX_2|_{R \Join S} Z_1 \equiv Y_1 \indep X_1X_2|_{R \Join S} Z_1$.
\par
(ii-b)
Otherwise, $X_1 = DW_1$, where $W_1 \subseteq R$. Then $DW_1X_2 \indep Y_1 |_{G} Z_1$, then $DW_1 \indep Y_1|_{G_1} Z_1$ (Observation~\ref{obs:subset-indep}), since $G_1$ is a P-map $DW_1 \indep Y_1|_{R} Z_1$ and by Theorem~\ref{thm:pmap-ci-propagates}, 
\begin{equation}\label{equn:y1-d-z1-2}
Y_1 \indep DW_1|_{R \Join S} Z_1
\end{equation}
By Corollary~\ref{cor:cl-join-0}, $Y_1Z_1W_1 \indep X_2 |_{R \Join S} D$, and by weak union (\ref{equn:wu}),
\begin{equation}\label{equn:y1-x2-dz1-2}
Y_1 \indep X_2|_{R \Join S} (DW_1)Z_1
\end{equation}
Combining (\ref{equn:y1-d-z1-2}) and (\ref{equn:y1-x2-dz1-2}) by contraction (\ref{equn:con}), we have $Y_1 \indep DW_1X_2|_{R \Join S} Z_1 \equiv Y_1 \indep X_1X_2|_{R \Join S} Z_1$.
\par
\textbf{(iii) $D \in X_2, D \notin X_2, Y_1, Z_1$}: (iii-a)If $D = X_2$, \ie, if $DX_1 \indep Y_1 |_{G} Z_1$, then $DX_1 \indep Y_1|_{G_1} Z_1$ (Observation~\ref{obs:subset-indep}), since $G_1$ is a P-map $DX_1 \indep Y_1|_{R} Z_1$ and by Theorem~\ref{thm:pmap-ci-propagates}, 
\begin{equation*}
DX_1 \indep Y_1|_{R \Join S} Z_1 ~\equiv~ X_1X_2 \indep Y_1|_{R \Join S} Z_1
\end{equation*}
\par
(iii-b)
Otherwise, $X_2 = DW_2$, where $W_2 \subseteq S$. Then $DW_2X_1 \indep Y_1 |_{G} Z_1$, then $DX_1 \indep Y_1|_{G_1} Z_1$ (Observation~\ref{obs:subset-indep}), since $G_1$ is a P-map $DX_1 \indep Y_1|_{R} Z_1$ and by Theorem~\ref{thm:pmap-ci-propagates}, 
\begin{equation}\label{equn:y1-d-z1-3}
Y_1 \indep DX_1|_{R \Join S} Z_1
\end{equation}
By Corollary~\ref{cor:cl-join-0}, $Y_1Z_1X_1 \indep W_2 |_{R \Join S} D$, and by weak union (\ref{equn:wu}),
\begin{equation}\label{equn:y1-x2-dz1-3}
Y_1 \indep W_2|_{R \Join S} (DX_1)Z_1
\end{equation}
Combining (\ref{equn:y1-d-z1-3}) and (\ref{equn:y1-x2-dz1-3}) by contraction (\ref{equn:con}), we have $Y_1 \indep (DX_1)W_2|_{R \Join S} Z_1 \equiv Y_1 \indep X_1X_2|_{R \Join S} Z_1$.
\par
\textbf{(iv) $D \in Y_1, D \notin X_1, X_2, Z_1$}: (iv-a) Suppose $Y_1 = D$, \ie, $X_1X_2 \indep D|_{G}Z_1$. Since $D \notin Z_1$, removing $Z_1$ cannot disconnect $X_2$ with $D$ since $G_2$ is connected.
\par
(iv-b) Otherwise, $Y_1 = DW_1$,  \ie, $X_1X_2 \indep DW_1 |_{G}Z_1$. Therefore, $X_1 \indep DW_1 |_{G_1}Z_1$
and by Theorem~\ref{thm:pmap-ci-propagates}
\begin{equation}\label{equn:iv-1}
X_1 \indep DW_1 |_{R \Join S}Z_1
\end{equation}
Also $X_2 \indep DW_1 |_{G} Z_1$. By Lemma~\ref{lem:two-one-prop},
\begin{equation}\label{equn:iv-2}
X_2 \indep DW_1 |_{R \Join S}Z_1
\end{equation}
  $X_1W_1 Z_1 \indep X_2 |_{R \Join S} D$ $\Rightarrow X_1 \indep X_2 |_{R \Join S} DW_1Z_1$ (weak union).
  By contraction and (\ref{equn:iv-1}), $X_1 \indep DW_1X_2 |_{R \Join S} DW_1Z_1$. By weak union, $X_1 \indep DW_1 |_{R \Join S} X_2Z_1$.  By contraction and (\ref{equn:iv-2}), $X_1X_2 \indep DW_1 |_{R \Join S} Z_1$.
  \par
\textbf{(v) $D \notin X_1, X_2, Y_1$}:
Since $X_1X_2 \indep Y_1 |_{G} Z_1$, then by definition of independence in a graph, $X_1 \indep Y_1 |_{G} Z_1$, and since $G_1$ is a subgraph of $G$, $X_1 \indep Y_1 |_{G_1} Z_1$. Since $G_1$ is a P-map of $R$, $X_1 \indep Y_1 |_{R} Z_1$, by the strong union property of Theorem~\ref{thm:PP-undirected} $X_1 \indep Y_1 |_{R} DZ_1$, and by Theorem~\ref{thm:pmap-ci-propagates}, 
\begin{equation}\label{equn:x1-y1-dz1-rs}
X_1 \indep Y_1 |_{R \Join S} DZ_1
\end{equation}
Also since $X_1X_2 \indep Y_1 |_{G} Z_1$, by decomposition property of graphoid axioms (\ref{equn:dec}), $X_2 \indep Y_1 |_{G} Z_1$. We claim that $Y_1 \indep DX_2|_{G}Z_1$. Suppose not. Since $X_2 \indep Y_1 |_{G} Z_1$, there is a path $p_1$ between $D$ and $Y_1$ that does not go through $Z_1$. Since $D$ and $X_2$ are connected by at least one path $p_2$ in $G_2$ (which does not go through $Z_1$ since $Z_1 \subseteq R$), by combining $p_1$ and $p_2$ (and simplifying to get a simple path), we get a path from $Y$ to $X_2$ that does not go through $Z_1$, contradicting the assumption that $X_1X_2 \indep Y_1 |_{G} Z_1$. Hence $Y_1 \indep DX_2|_{G}Z_1$, and since $Y_1, Z_1$ is in $R$ and $DX_2$ is in $S$, by Lemma~\ref{lem:two-one-prop},
\begin{equation}\label{equn:y1-dx2-z1-rs}
Y_1 \indep DX_2|_{R \Join S} Z_1
\end{equation}
\par
Next note that $X_1Y_1 \indep X_2 |_{G} DZ_1$, since $D$ itself is a cutset between $X_1Y_1$ and $X_2$ (Observation~\ref{obs:join-attr-cutset}). Since $X_1Y_1 \subseteq R$, $X_2 \subseteq S$,  and $DZ_1 \subseteq R$, by Lemma~\ref{lem:two-one-prop}, 
\begin{eqnarray*}
&& X_1Y_1 \indep X_2 |_{R\Join S} DZ_1\\
& \Rightarrow & X_1 \indep X_2 |_{R\Join S} (DZ_1)(Y_1)~~~~\textrm{(weak union (\ref{equn:wu}))}\\
& \Rightarrow & X_1 \indep X_2Y_1 |_{R\Join S} (DZ_1)~~\textrm{((\ref{equn:x1-y1-dz1-rs}) and contraction (\ref{equn:con}))}\\
& \Rightarrow & X_1 \indep Y_1 |_{R\Join S} Z_1 (DX_2)~~~~\textrm{(weak union (\ref{equn:wu}))}\\
& \Rightarrow & X_1(DX_2) \indep Y_1 |_{R\Join S} Z_1~~\textrm{((\ref{equn:y1-dx2-z1-rs}) and contraction (\ref{equn:con}))}\\
& \Rightarrow & X_1X_2 \indep Y_1 |_{R\Join S} Z_1~~~~\textrm{(decomposition (\ref{equn:dec}))}
\end{eqnarray*}
\end{proof}

\subsection{Proof of Lemma~\ref{lem:y1-z2-empty}}
\begin{proof}[of Lemma~\ref{lem:y1-z2-empty}]
Since $X_1X_2 \indep Y_2 |_{G} Z_1$, no path exists between any vertices in $X_1, X_2$ and $Y_2$ in $G_{-Z_1}$.
First we argue that the join attribute $D \in Z_1$. 
Suppose not. In $G_{-Z_1}$, no path exists between $X_2$ and $Y_2$ in $G$.  Since $D \notin Z_1$, removing $Z_1$ does not remove any edge in $G_2$, implying that  no path exists between $X_2$ and $Y_2$ in $G_2$, which contradicts the assumption that $G_2$ is connected. 
\par
Hence $D \in Z_1$. Assume $Z_1 = DW_1$ where $W_1 \subseteq R$. 
Since $X_1X_2 \indep Y_2 |_{G} Z_1$, we have 
\begin{equation}\label{equn:y2-x2-2}
X_2 \indep Y_2 |_{G} DW_1
\end{equation}
%Since %$X_1 \indep Y_2 |_{G} Z_1 ~\equiv~  $X_1 \indep Y_2 |_{G} DW_1$, and 
Note that 
$W_1$ belongs to $R$ or $G_1$, and does not have any common vertex in $G_2$. In other words, removing $W_1$ cannot affect the connectivity between $X_2, Y_2$. Therefore,  it holds that 
\begin{equation}\label{equn:y2-x2-3}
Y_2 \indep X_2 |_{G} D
\end{equation}
 \ie, removing $D$ disconnects $X_2, Y_2$ in $G$. Since $G_2$ is a subset of $G$,  $Y_2 \indep X_2 |_{G_2} D$, and since $G_2$ is a P-map of $S$,  $Y_2 \indep X_2 |_{S} D$, and by Theorem~\ref{thm:pmap-ci-propagates}, 
\begin{equation}\label{equn:y2-x2-d-rs}
Y_2 \indep X_2 |_{R \Join S} D
\end{equation}

By Corollary~\ref{cor:cl-join-0}, 
\begin{eqnarray*}
&&X_1W_1 \indep Y_2X_2 |_{R \Join S} D\\
& \Rightarrow & X_1W_1 \indep Y_2 |_{R \Join S} DX_2~~~~\textrm{(weak union (\ref{equn:wu}))}\\
& \Rightarrow & X_1X_2W_1 \indep Y_2 |_{R \Join S} D~~\textrm{((\ref{equn:y2-x2-d-rs}) and contraction (\ref{equn:con}))}\\
& \Rightarrow & X_1X_2 \indep Y_2 |_{R \Join S} DW_1~~~~\textrm{(weak union (\ref{equn:wu}))}\\
& \equiv & X_1X_2 \indep Y_2 |_{R \Join S} Z_1
\end{eqnarray*}
\end{proof}

\subsection{Proof of Lemma~\ref{lem:z2-empty}}
The proof is similar to Lemma~\ref{lem:y1-z2-empty} but uses Lemma~\ref{lem:y2-z2-empty}.

\begin{proof}[of Lemma~\ref{lem:z2-empty}]
First we argue that the join attribute $D \in Z_1$. 
Suppose not. Then in $G_{-Z_1}$, no path exists between $X_2$ and $Y_2$ in $G$.  Since $D \notin Z_1$, removing $Z_1$ does not remove any edge in $G_2$, implying that  no path exists between $X_2$ and $Y_2$ in $G_2$, which contradicts the assumption that $G_2$ is connected. 
\par
Hence $D \in Z_1$. Assume $Z_1 = DW_1$ where $W_1 \subseteq R$. Since $X_1 \indep Y_2 |_{G} Z_1 ~\equiv~ X_1 \indep Y_2 |_{G} DW_1$, and $W_1$ belongs to $R$ or $G_1$, it holds that $Y_2 \indep X_2 |_{G} D$. Since $G_2$ is a subset of $G$,  $Y_2 \indep X_2 |_{G_2} D$, and since $G_2$ is a P-map of $S$,  $Y_2 \indep X_2 |_{S} D$, and by Theorem~\ref{thm:pmap-ci-propagates}, 
\begin{equation}\label{equn:y2-x2-d-rs-2}
Y_2 \indep X_2 |_{R \Join S} D
\end{equation}
Since $X_1X_2 \indep Y_1Y_2 |_{G} Z_1$, by definition,   $X_1X_2 \indep Y_1 |_{G} Z_1$, and by Lemma~\ref{lem:y2-z2-empty}, 
\begin{equation}\label{equn:x1x2-y1-z1-rs}
X_1X_2 \indep Y_1 |_{R \Join S} Z_1 
\end{equation}
By Corollary~\ref{cor:cl-join-0}, 
\begin{eqnarray*}
&&X_1Y_1W_1 \indep Y_2X_2 |_{R \Join S} D\\
& \Rightarrow & X_1Y_1W_1 \indep Y_2 |_{R \Join S} DX_2~~~~\textrm{(weak union (\ref{equn:wu}))}\\
& \Rightarrow & X_1X_2Y_1W_1 \indep Y_2 |_{R \Join S} D~~\textrm{((\ref{equn:y2-x2-d-rs-2}) and contraction (\ref{equn:con}))}\\
& \Rightarrow & X_1X_2 \indep Y_2 |_{R \Join S} DW_1Y_1~~~~\textrm{(weak union (\ref{equn:wu}))}\\
& \Rightarrow & X_1X_2 \indep Y_2Y_1 |_{R \Join S} DW_1~~\textrm{((\ref{equn:x1x2-y1-z1-rs}) and contraction (\ref{equn:con}))}\\
& \equiv & X_1X_2 \indep Y_2Y_1 |_{R \Join S} Z_1
\end{eqnarray*}
\end{proof}

\subsection{Proof of Lemma~\ref{lem:x2-y2-z1-empty}}
\begin{proof}[of Lemma~\ref{lem:x2-y2-z1-empty}]
Given $X_1 \indep Y_1 |_{G} Z_2$. We claim that $D \in Z_2$, otherwise, since $G_1$ is connected, removing vertices from $G_2$ cannot disconnect $X_1, Y_1$ in $G_1$. 
\par
We assume $Z_2 = DW_2$: if $Z_2 = D$, then $X_1 \indep Y_1 |_{G_1} D \equiv X_1 \indep Y_1 |_{R} D$ ($G_1$ is a P-map of $R$), and therefore by Theorem~\ref{thm:pmap-ci-propagates}:
\begin{equation}\label{equn:x1-y1-d-4}
X_1 \indep Y_1 |_{R \Join S} D 
\end{equation}
\par
By Corollary~\ref{cor:cl-join-0}, $X_1Y_1 \indep W_2 |_{R \Join S} D$. By weak union, $X_1 \indep W_2 |_{R \Join S} DY_1$. Combining with (\ref{equn:x1-y1-d-4}) and using contraction (\ref{equn:con}), $X_1 \indep Y_1W_2 |_{R \Join S} D$. By weak union, $X_1 \indep Y_1 |_{R \Join S} DW_2 \equiv X_1 \indep Y_1 |_{R \Join S} Z_2$.
\end{proof}

\subsection{Proof of Lemma~\ref{lem:z1-z2-non-empty}}
To prove this lemma, we will need additional lemmas:

\begin{lemma}\label{lem:dxz1z2}
\begin{itemize}
\item[(A)] If $D \indep X_1 |_{G} Z_1Z_2$, , then $D \indep X_1 |_{R\Join S} Z_1Z_2$.
\item[(B)] If $D \indep X_1X_2 |_{G} Z_1Z_2$, , then $D \indep X_1 |_{R\Join S} Z_1Z_2$.
\end{itemize}
\end{lemma}
\begin{proof}
\textbf{(A)}
$D \indep X_1 |_{G} Z_1Z_2$ $\Rightarrow D \indep X_1 |_{G_1} Z_1$, and since $G_1$ is a P-map of $R$ and by Theorem~\ref{thm:pmap-ci-propagates}, 
\begin{equation}\label{equn:DX-1}
D \indep X_1|_{R \Join S} Z_1
\end{equation}
By Corollary~\ref{cor:cl-join-0}), $X_1 Z_1 \indep Z_2 |_{R \Join S} D$. By weak union (\ref{equn:wu}), $X_1 \indep Z_2 |_{R \Join S} Z_1D$. Combining with \ref{equn:DX-1} by contraction (\ref{equn:con}),  $X_1 \indep DZ_2 |_{R \Join S} Z_1$. By weak union again, $X_1 \indep D |_{R \Join S} Z_1Z_2$.
\par
\textbf{(B)}
$D \indep X_1X_2 |_{G} Z_1Z_2$ $\Rightarrow D \indep X_1|_{G} Z_1$, By Observation~\ref{obs:cutset-sup}, $D \indep X_1|_{G} Z_1X_2$. By (A) above, 
\begin{equation}\label{equn:DX-2}
D \indep X_1 |_{R \Join S} Z_1X_2
\end{equation}
Similarly, $D \indep X_2|_{G} Z_2$ and 
\begin{equation}\label{equn:DX-3}
D \indep X_2 |_{R \Join S} Z_1Z_2
\end{equation}
By Corollary~\ref{cor:cl-join-0}), $X_1Z_1 \indep X_2Z_2 |_{R \Join S} D$. By weak union (\ref{equn:wu}), $X_1 \indep Z_2 |_{R \Join S} Z_1DX_2$. Combining with \ref{equn:DX-2} by contraction (\ref{equn:con}),  $X_1 \indep DZ_2 |_{R \Join S} Z_1X_2$. By contraction and (\ref{equn:DX-3}), $X_1X_2 \indep D |_{R \Join S} Z_1Z_2$.
\end{proof}

\begin{lemma}\label{lem:dx1-y1-y2}
Suppose the join attribute $D \notin X_1, Y_1, Y_2, Z_1, Z_2$.
\begin{itemize}
\item[(A)] If $DX_1 \indep Y_1 |_{G} Z_1Z_2$, then $DX_1 \indep Y_1 |_{R\Join S} Z_1Z_2$.
\item[(B)] If $DX_1 \indep Y_2 |_{G} Z_1Z_2$, then $DX_1 \indep Y_2 |_{R\Join S} Z_1Z_2$.
\end{itemize}
\end{lemma}
\begin{proof}
\textbf{(A)} Since $DX_1 \indep Y_1 |_{G} Z_1Z_2$, $D \indep Y_1 |_{G} Z_1Z_2$, and by Lemma~\ref{lem:dxz1z2}, 
\begin{equation}\label{equn:PP-1}
D \indep Y_1 |_{G} Z_1Z_2
\end{equation}
Also, $X_1 \indep Y_1 |_{G} Z_1Z_2$. Hence $X_1 \indep Y_1 |_{G} Z_1Z_2D$ (Observation~\ref{obs:cutset-sup}), and by Lemma~\ref{lem:y2-z2-empty},
\begin{equation}\label{equn:PP-2}
X_1 \indep Y_1 |_{R \Join S} Z_1Z_2D
\end{equation}
Combining (\ref{equn:PP-1}) and (\ref{equn:PP-2})  by contraction, $DX_1 \indep Y_1 |_{R \Join S} Z_1Z_2$.
\par
\textbf{(B)} Since $DX_1 \indep Y_2 |_{G} Z_1Z_2$, $D \indep Y_2 |_{G} Z_1Z_2$, and by Lemma~\ref{lem:dxz1z2}, 
\begin{equation}\label{equn:QQ-1}
D \indep Y_2 |_{G} Z_1Z_2
\end{equation}
Also, $X_1 \indep Y_2 |_{G} Z_1Z_2$. Hence $X_1 \indep Y_2 |_{G} Z_1Z_2D$ (Observation~\ref{obs:cutset-sup}), and by Lemma~\ref{lem:y1-z2-empty},
\begin{equation}\label{equn:QQ-2}
X_1 \indep Y_2 |_{R \Join S} Z_1Z_2D
\end{equation}
Combining (\ref{equn:QQ-1}) and (\ref{equn:QQ-1})  by contraction, $DX_1 \indep Y_2 |_{R \Join S} Z_1Z_2$.
\end{proof}

\begin{lemma}\label{lem:d-z1-z2}
Suppose the join attribute $D \notin X_1, Y_1, Y_2, Z_1, Z_2$.
\begin{itemize}
\item[(A)] If $X_1\indep Y_1 |_{G} DZ_1Z_2$, then $X_1\indep Y_1 |_{R \Join S} DZ_1Z_2$.
\item[(B)] If  $X_1\indep Y_2 |_{G} DZ_1Z_2$, then $X_1\indep Y_1 |_{R \Join S} DZ_1Z_2$.
\end{itemize}
\end{lemma}
\begin{proof}
\textbf{(A)} If $X_1\indep Y_1 |_{G} DZ_1Z_2$, then $X_1\indep Y_1 |_{G} DZ_1$. Also $X_1\indep Z_2 |_{G} DZ_1$. Hence $X_1\indep Y_1Z_2 |_{G} DZ_1$. %\red{another obs?}. 
By Lemma~\ref{lem:y2-z2-empty}, $X_1\indep Y_1Z_2 |_{R \Join S} DZ_1$. By weak union, $X_1\indep Y_1 |_{G} DZ_1Z_2$.
\par
\textbf{(B)} If $X_1\indep Y_2 |_{G} DZ_1Z_2$, then $X_1\indep Y_2Z_2 |_{G} DZ_1$ ($D$ itself disconnects $X_1$ from $Z_2$).  By Lemma~\ref{lem:two-one-prop}, $X_1\indep Y_2Z_2 |_{R \Join S} DZ_1$. By weak union, $X_1\indep Y_2 |_{G} DZ_1Z_2$.
\end{proof}

Now we prove Lemma~\ref{lem:z1-z2-non-empty}. There are four non-equivalent cases as stated in the  lemma.

\begin{proof}[of Lemma~\ref{lem:z1-z2-non-empty}]
\textbf{(A)} If $D \in Z_1Z_2$, it follows from Lemma~\ref{lem:d-z1-z2}. If $D \in X_1$ or $Y_1$, it follows from Lemma~\ref{lem:dxz1z2} and \ref{lem:dx1-y1-y2}. Hence we assume $D \notin Z_1Z_2, X_1, Y_1$.  $X_1 \indep Y_1 |_{G} Z_1Z_2$ $\Rightarrow$ $X_1 \indep Y_1 |_{G} Z_1$. This is because of the fact that $D \notin Z_1, Z_2$, $Z_2 \subseteq S$, and no path between $X_1, Y_1$ in $G$ can go through $Z_2$. In turn, $X_1 \indep Y_1 |_{G_1} Z_1$ (Observation~\ref{obs:subset-indep}). Hence $X_1 \indep Y_1 |_{G_1} DZ_1$ (Observation~\ref{obs:cutset-sup}), since $G_1$ is a P-map of $R$, by Theorem~\ref{thm:pmap-ci-propagates},
\begin{equation}\label{equn:A-1}
X_1 \indep Y_1 |_{R \Join S} DZ_1
\end{equation} 
By Corollary~\ref{cor:cl-join-0}), $X_1Y_1Z_1 \indep Z_2 |_{R \Join S} D$. By weak union (\ref{equn:wu}), $X_1 \indep Z_2 |_{R \Join S} DY_1Z_1$. Combining with (\ref{equn:A-1}) by contraction (\ref{equn:con}), $X_1 \indep Y_1Z_2 |_{R \Join S} DZ_1$. By weak union, 
\begin{equation}\label{equn:A-2}
X_1 \indep Y_1 |_{R \Join S} DZ_1Z_2
\end{equation} 
We claim that either $X_1 \indep D |_{G} Z_1Z_2$ or $Y_1 \indep D |_{G} Z_1Z_2$. Indeed, if both fail, then there is a path from $X_1$ to $Y_1$ in $G_{-Z_1Z_2}$ violating the assumption that $X_1 \indep Y_1 |_{G} Z_1Z_2$. 
\par
If $X_1 \indep D |_{G} Z_1Z_2$,by Lemma~\ref{lem:dxz1z2},  $\Rightarrow  X_1 \indep D |_{R \Join S} Z_1Z_2$. Combining with (\ref{equn:A-2}) by contraction,   $X_1 \indep DY_1 |_{R \Join S} Z_1Z_2$, and by decomposition, $X_1 \indep Y_1 |_{R \Join S} Z_1Z_2$. 
\par
If $Y_1 \indep D |_{G} Z_1Z_2$, by similar argument, $DX_1 \indep Y_1 |_{R \Join S} Z_1Z_2$, and by decomposition, $X_1 \indep Y_1 |_{R \Join S} Z_1Z_2$. 
\par
\textbf{(B)} 
 If $D \in Z_1Z_2$, it follows from Lemma~\ref{lem:d-z1-z2}. If $D \in X_1$ or $Y_1$, it follows from Lemma~\ref{lem:dxz1z2} and \ref{lem:dx1-y1-y2}.
Hence assume $D \notin Z_1Z_2, X_1, Y_2$.  If $X_1 \indep Y_2 |_{G} Z_1Z_2$, either $X_1 \indep D |_{G} Z_1$ or  $Y_2 \indep D |_{G} Z_2$, otherwise, a path exists between $X_1$ and $Y_2$ through $D$ that is not in $Z_1Z_2$. 
\par
Without loss of generality, assume $X_1 \indep D |_{G} Z_1$. Then $X_1 \indep D |_{G} Z_1Z_2$ (Observation~\ref{obs:cutset-sup}), and by Lemma~\ref{lem:dxz1z2}, 
\begin{equation}\label{equn:B-1}
X_1 \indep D |_{R \Join S} Z_1Z_2
\end{equation} 
\par
By Corollary~\ref{cor:cl-join-0}), 
$X_1Z_1 \indep Y_2Z_2|_{R \Join S}D$ $\Rightarrow X_1 \indep Y_2 |_{R \Join S}DZ_1Z_2$ (weak union). Combining with (\ref{equn:B-1}) by contraction,  $\Rightarrow X_1 \indep Y_2D|_{R \Join S}Z_1Z_2$, and by decomposition, $X_1 \indep Y_2|_{R \Join S}Z_1Z_2$.
\par
\textbf{(C)} 
Since $X_1X_2\indep Y_1 |_{G} Z_1Z_2$, $X_1 \indep Y_1 |_{G} Z_1Z_2$, and by case (A) above, 
\begin{equation}\label{equn:CC-1}
X_1 \indep Y_1 |_{R \Join S} Z_1Z_2
\end{equation}
Also  $X_2 \indep Y_1 |_{G} Z_1Z_2$, therefore, $X_2 \indep Y_1 |_{G} Z_1Z_2X_1$ (Observation~\ref{obs:cutset-sup}), and by case (B) above,
\begin{equation}\label{equn:CC-2}
X_2 \indep Y_1 |_{R \Join S} (Z_1Z_2)X_1
\end{equation}
Applying contraction (\ref{equn:con}) on (\ref{equn:CC-1}) and (\ref{equn:CC-2}), $X_1X_2 \indep Y_1 |_{R \Join S} Z_1Z_2$.
\par
\textbf{(D)} 
Since $X_1X_2\indep Y_1Y_2 |_{G} Z_1Z_2$, $X_1X_2 \indep Y_1 |_{G} Z_1Z_2$, and by case (C) above, 
\begin{equation}\label{equn:DD-1}
X_1X_2 \indep Y_1 |_{R \Join S} Z_1Z_2
\end{equation}
Also  $X_1X_2 \indep Y_2 |_{G} Z_1Z_2$, therefore, $X_1X_2 \indep Y_2|_{G} Z_1Z_2Y_1$ (Observation~\ref{obs:cutset-sup}), and by case (B) above,
\begin{equation}\label{equn:DD-2}
X_1X_2 \indep Y_2 |_{R \Join S} (Z_1Z_2)X_1
\end{equation}
Applying contraction (\ref{equn:con}) on (\ref{equn:DD-1}) and (\ref{equn:DD-2}), $X_1X_2 \indep Y_1Y_2 |_{R \Join S} Z_1Z_2$.

\par
(D) Consider the remaining case that $D \notin Z_1Z_2, X_1, X_2, Y_1, Y_2$.  If $X_1X_2 \indep Y_1Y_2 |_{G} Z_1Z_2$, either (i) $X_1 \indep D |_{G} Z_1$, or  (ii) $Y_1 \indep D |_{G} Z_2$ and $Y_2 \indep D |_{G} Z_2$, otherwise, a path exists between $X_1$ and either $Y_1$ or $Y_2$ through $D$ that is not in $Z_1Z_2$. 
\end{proof}

\subsection{Proof of Theorem~\ref{thm:union-imap}}

\begin{proof}[of Theorem~\ref{thm:union-imap}]
%If all of $X, Y, Z$ belongs to the 
Suppose $X = X_1X_2$, $Y = Y_1Y_2$, $Z  = Z_1Z2$. Lemma~\ref{lem:z1-z2-non-empty} covers the (non-equivalent) cases when both $Z_1, Z_2$ are non-empty.  When $Z_2 = \emptyset$ (equivalently $Z_1$): (i) when both $X, Y$ contain both subsets, the result follows from Lemma~\ref{lem:z2-empty}; (ii) when one of $X,Y$ contain both subsets and the other contain one, the results follows from Lemmas~\ref{lem:y1-z2-empty} and \ref{lem:y2-z2-empty}; (iii) when both $X, Y$ contain one subset each, the result follows from Lemma~\ref{lem:two-one-prop} or Theorem~\ref{thm:pmap-ci-propagates} (in this case, the CI holds in the base relation since $G_1$ is a P-map, and therefore propagates to the joined relation). 
%
%\begin{tabular}{|c|c|c|}
%\hline
%$X$ & $Y$ & $Z$ & Follows from..\\\hline\hline
%$X_1$ & $Y_1$ & $Z_1$ & \\\hline
%$X_1$ & $Y_1$ & $Z_1$ & \\\hline
%$X_1$ & $Y_1$ & $Z_1$ & \\\hline
%$X_1$ & $Y_1$ & $Z_1$ & \\\hline
%$X_1$ & $Y_1$ & $Z_1$ & \\\hline
%\end{tabular}
\end{proof}

\section{Proofs from Section~6}

\subsection{Example~\ref{eg:no-dmap}: New CIs in the Joined Relation}
\begin{example}\label{eg:no-dmap}
Consider relations $R(A, B, C, D, E)$ and $S(D, F)$. % where the P-maps of $R$ and $S$ are both given by complete graphs. 
The relation instance contains the tuples $(a_1, b_1, c, d_1, -)$, $(a_1, b_2, c, d_2, -)$, $(a_2, b_1, c, d_3, -)$, $(a_2, b_2, c, d_4, -)$ (with unique values of $E$ as `-') respectively $2, 3, 1, 1$ times. %Additional tuples with different values of $A, B,C, D, E$ are inserted to destroy all CIs in this relation. 
For $C = c$, $\Pr_R[A = a_1, B = b_1 | C = c] = \frac{2}{7}$, whereas $\Pr_R[A = a_1 | C = c] = \frac{5}{7}$ and $\Pr_R[A = b_1 | C = c] = \frac{3}{7}$, so $\neg (A \indep B |_R C)$. Now suppose $S$ contains 3, 2, 1, 1 tuples respectively of the form $(d_1, -)$, $(d_2, -)$, $(d_3, -)$, $(d_4, -)$, and therefore using the new frequencies in $J = R \Join S$,
$\Pr_J[A = a_1, B = b_1 | C = c] = \frac{6}{14}$, whereas $\Pr_R[A = a_1 | C = c] = \frac{12}{14}$ and $\Pr_R[A = b_1 | C = c] = \frac{7}{14}$, thus satisfying the CI $A \indep B |_{R \Join S} C$.
\end{example}

\section{Additional Related Work}\label{sec:related}
\textbf{Recent work in the intersection of causality and databases.~} While causality has been used as a motivation to explain interesting  observations in the area of databases \cite{MeliouGMS2011, RoyS14}, actual causal inference as done in statistical studies using techniques from databases has drawn attention only very recently
%The concept of causality based on the potential outcome model is being studied in the context of database techniques only recently 
\cite{FLAME2017, salimi2017zaliql}. In \cite{FLAME2017}, Roy et al. studied efficient matching methods that include a large number of covariates, while ensuring that each group contains at least one treated and one control units, and such that the selected covariates predict the outcome well. Since the goals conflict with each other, the proposed technique aim to match as many units as possible using as many covariates as possible, then drops the `less useful' covariates to match more units in the next round. 
In \cite{salimi2017zaliql}, Salimi and Suciu proposed a framework for supporting various existing causal inference techniques efficiently inside a database-based engine, for both online and offline settings. 
The problem of matching units with the same values of covariates has a strong connection with the \emph{group-by} operator used in SQL queries, therefore, both \cite{FLAME2017, salimi2017zaliql} use database queries to efficiently implement the matching algorithms using standard relational database management systems. 
However, both \cite{FLAME2017, salimi2017zaliql} consider the problem of efficient causal inference on a single relation, whereas the main focus of the framework proposed in the current paper is extending causal inference to multiple relations. 

\textbf{Matching for observational studies.~} 
Matching has been studied since 1940s for observational studies \cite{chapin1947experimental,greenwood1945experimental}. 
In the 1970s and 1980s, a large literature on different dimension reduction approaches to matching was developed %\cite{rubin1973use, 
(\eg, \cite{rubin1973matching, rubin1976multivariate, cochran1973controlling}). One of the most famous approaches that is still 
prevalent in current research
is matching using \emph{propensity score}, the distribution of treatment assignment conditional on background covariates. Rosenbaum and Rubin  \cite{RosenbaumRubin1983} demonstrated that under certain assumptions the propensity score is a \emph{balancing score}  (within each matched group, treated and control units are independent) that allowed for the unbiased estimation of average causal effects, 
%; Rosenbaum and Rubin  \cite{RosenbaumRubin1983} showed that propensity score is a \emph{balancing score} (within each matched group, treated and control units are independent), 
and also is the most coarsened balancing score, thereby allowing valid matching of many units. 
%Propensity score matching also belongs to another set of techniques for matching called \emph{subclassification} \cite{rubin2006matched}.
Beyond naive matching, the propensity score has been integrated into \emph{subclassification} \cite{rubin2006matched, rosenbaum1984reducing} and \emph{multivariate matching schemes} \cite{rosenbaum1985constructing} for observational data. 
Relatively recent \emph{coarsened exact matching} avoids fitting complicated propensity score models by coarsening or discretizing covariates in such a way that the newly constructed covariates allow for exact matching \cite{iacus2011causal}. %,iacus2011multivariate} 
%(although has limitations like requiring knowledge on the scale and being sensitive to bin boundaries). 
A general overview of  matching techniques can be found in \cite{stuart2010matching}.

\textbf{Causality and causal graphs in artificial intelligence.~}
The study of causality in the artificial intelligence community is based on the notion of counterfactuals,
%The notion of \emph{counterfactual} causes, 
\cut{
which can be traced back to
Hume~\cite{hume1748, Menzies:Causation2008}, 
%explains causality in an intuitive way: if the first event (cause) had not occurred, then the second
%event (effect) would not have occurred. The 
whereas the best known counterfactual analysis
of causation is due to Lewis~\cite{Lewis1973}. The 
}
where the basic idea is that if the first event (cause) had not occurred, then the second
event (effect) would not have occurred (in contrast to the quantitative measure of ATE in the causal analysis in statistics).
%Several philosophers~\cite{hitchcock2001, hitchcock2007, woodward2003} explored an alternative approach to counterfactuals that employs structural equations. 
The work by Pearl and others~\cite{PearlBook2000} %and Halpern and Pearl's work~\cite{HalpernP01}
refined the generally accepted aspects of causality into a rigorous
definition using structural equations, which can also be viewed as \emph{causal networks}.
%\par
%\cut{
%\em {Causal
%  Networks}~\cite{pearl2010foundations} can be traced back to
%the %pioneering 
%work of Wright in 1923~\cite{wright1923theory} and were later popularized by
%Pearl~\cite{PearlBook2000},
%} 
In causal networks, the causal effects and
counterfactuals are modeled using a mathematical operator
$do(X=x)$ to simulate the effect of an intervention  (called \emph{do-calculus}). 
\cut{
Specifically, the distribution $P(Y=y |do(T=1))$
describes the outcome of the population from a hypothetical experiment
in which we administer treatment uniformly to the entire
population. Then, the average treatment effect of $T$ on $Y$ is
measured with $E(Y |do(T=1))- E(Y |do(T=0))$. Any expression in the
$do$ calculus can be written in terms of the potential outcome notion
\cite{Richardson}.  For example, $P(Y=y |do(T=1))$ denotes the same
distribution as $P(Y(1)=y)$. 
}
Pearl proposed a sufficient condition
called {\em back-door criterion} to identify the distribution 
$P(Y=y |do(T=1))$ in order to estimate the distribution on the causal effect of the treatment \cite{PearlBook2000}.

\textbf{Other work on causality in databases.~} 
Motivated by the notion of causality and intervention by Pearl \cite{PearlBook2000}, Meliou et al. \cite{MeliouGMS2011}, studied the problem of finding and ranking input tuples as `causes' of query answers. Roy and Suciu \cite{RoyS14} studied finding \emph{explanations} for aggregate query answers, a summary of tuples that have a large impact on the query answers, with a similar motivation. 
In \cite{Silverstein+2000}, Silverstein \etal\ studied the problem of efficiently determining
causal relationship %(\ie, not simple association or correlation) 
for mining market basket data. 
%The authors identified the limitations, both in terms of scalability and appropriateness, of Bayesian learning ideas to mining
%large databases and proposed solutions to handle these limitations. 
Maier \etal~\cite{RattiganMJ11,MaierTOJ10} identified ways to distinguish
statistical association from actual causal dependencies in the relational
domain. \\

To the best of our knowledge, none of the work till date has considered the problem of extending causal analysis and potential outcome model to multiple relations with a rigorous study of the underlying assumptions like SUTVA and strong ignorability, which is the main contribution of the framework proposed in this paper. 

%, and the papers surveyed in \cite{FLAME2017, DBLP:journals/corr/SalimiS16}.

%\red{CHECK OR REMOVE}
%causality in stat

%causality in ai
%
%causality in databases

\end{sloppypar}

\end{document}